\documentclass[10pt]{amsart}
\usepackage{lineno}

\usepackage[margin=2.2cm]{geometry}
\usepackage{fullpage}
\usepackage{tikz}
\usepackage{tikz-cd}
\usetikzlibrary{matrix,arrows,shapes,cd}
\usepackage{pgf}
\usepackage{url}
\usepackage{longtable}
\usepackage{enumerate}
\usepackage{amsthm}
\usepackage{amssymb}
\usepackage{amsmath}
\usepackage{amscd}
\usepackage{eucal}
\usepackage{mathrsfs}
\usepackage{upgreek}
\usepackage{mathtools}
\usepackage{dsfont}
\usepackage{yfonts}
\usepackage[T1]{fontenc}
\usepackage{bbm}
\usepackage{graphics}
\usepackage{array}
\usepackage[retainorgcmds]{IEEEtrantools}
\usepackage{booktabs}
\usepackage{enumerate}
\usepackage{feynmf}
\usepackage{amsaddr}

\usetikzlibrary{trees}
\usetikzlibrary{decorations.pathmorphing}
\usetikzlibrary{decorations.markings}
\tikzset{
	photon/.style={decorate, decoration={snake}, draw=red},
	electron/.style={draw=blue, postaction={decorate},
		decoration={markings,mark=at position .55 with {\arrow[draw=blue]{>}}}},
	gluon/.style={decorate, draw=magenta,
		decoration={coil,amplitude=4pt, segment length=5pt}},
	sderiv/.style={postaction={decorate},
		decoration={markings,mark=at position .3 with {\arrow{>}}}},
	tderiv/.style={postaction={decorate},
		decoration={markings,mark=at position .7 with {\arrow{<}}}},
	stderiv/.style={postaction={decorate},
		decoration={markings,mark=at position .7 with {\arrow{<}},mark=at position .3 with {\arrow{>}}}}
}
\definecolor{see}{RGB}{67,75,179}
\definecolor{darksee}{RGB}{42,44,148}
\definecolor{honey}{RGB}{232,180,129}
\definecolor{lighthoney}{RGB}{255,254,220}
\definecolor{citecol}{rgb}{0.5,0,0} 
\definecolor{black1}{RGB}{130,150,209}

\DeclareSymbolFont{bbold}{U}{bbold}{m}{n}
\DeclareSymbolFontAlphabet{\mathbbold}{bbold}

\usepackage{soul}
\usepackage[colorinlistoftodos,textsize=footnotesize,color=blue!30,textwidth=2cm]{todonotes}

\setstcolor{red}
\definecolor{see}{RGB}{67,75,179}
\usepackage[bookmarks=true]{hyperref}
\hypersetup
{
	bookmarksopen=true,
	pdfmenubar=true, 
	pdfhighlight=/O, 
	colorlinks=true, 
	pdfpagelayout=SinglePage, 
	pdffitwindow=true, 
	linkcolor=see, 
	citecolor=red!75!darkgray, 
	urlcolor=see 
}
\newcommand{\Tcal}{\mathcal{T}}

\newcommand{\WF}{\mathrm{WF}}         
\newcommand{\id}{\mathrm{id}}               

\newcommand{\Ga}{\Gamma}

\newcommand{\La}{\Lambda}

\newcommand{\TT}{\Tcal}

\newcommand{\be}{\begin{equation}}
\newcommand{\ee}{\end{equation}}

\DeclareMathOperator{\supp}{\mathrm{supp}}      

\newcommand{\no}[1]{:\!#1\!:}

\theoremstyle{plain}
\newtheorem{thm}{Theorem}[section]
\newtheorem{df}[thm]{Definition}
\newtheorem{prop}[thm]{Proposition}

\newtheorem{lemma}[thm]{Lemma}

\theoremstyle{definition}
\newtheorem{rem}[thm]{Remark}

\title{Equilibrium states for the massive Sine-Gordon theory in the Lorentzian signature}

\author[1]{\small{Dorothea Bahns}}
\address{ University of G\"ottingen, \\
Mathematical Institute \\
\normalfont{\texttt{dbahns@mathematik.uni-goettingen.de}}}
\author[12]{\small{Nicola Pinamonti}}
\address{ University of Genova, \\Department of Mathematics \\ and INFN sez. Genova\\
	\normalfont{\texttt{pinamont@dima.unige.it}}}
\author[3]{\small{Kasia Rejzner}}
\address{University of York,  \\
	Department of Mathematics \\
\normalfont{\texttt{kasia.rejzner@york.ac.uk}}
}

\date{\today}

\begin{document}
 \sloppy


\maketitle
\begin{abstract}
In this paper we investigate the massive Sine-Gordon model in the ultraviolet finite regime in thermal states over the two-dimensional Minkowski spacetime.
We combine recently developed methods of perturbative algebraic quantum field theory with techniques developed in the realm of constructive quantum field theory over Euclidean spacetimes to construct the correlation functions of the equilibrium state of the Sine-Gordon theory in the adiabatic limit.
First of all, the observables of the Sine-Gordon theory are seen as functionals over the free configurations and are obtained as a suitable combination of the $S-$matrices of the interaction Lagrangian restricted to compact spacetime regions over the free massive theory.
These $S-$matrices are given as power series in the coupling constant with values in the algebra of fields over the
free massive theory. Adapting techniques like conditioning and inverse conditioning to spacetimes with Lorentzian signature, we prove that these power series converge when evaluated on a generic field configuration. 
The latter observation implies convergence in the strong operator topology in the GNS representations of the considered states. 
In the second part of the paper, adapting the cluster expansion technique to the Lorentzian case, we prove that the correlation functions of the interacting equilibrium state at finite temperature (KMS state) can be constructed also in the adiabatic limit, where the interaction Lagrangian is supported everywhere in space. 

\end{abstract}
\tableofcontents

\section{Introduction}

The programme of constructive quantum field theory (CQFT) was initiated 80 years ago, yet, to this day, only a handful of rather special models seem to have a chance to be constructible in 4 dimensions (for reviews, see~\cite{GrW20, Riv16}).
%
Yang-Mills and gauge theory on the other hand, have so far eluded all attempts and we have to be content with perturbation theory. With this paper we propose to go back to a 2-dimensional model that has been studied for decades \cite{GJS,GJS74,GR75a,GR75b}, and to combine the 
techniques
developed then with recent methods that emerged in the context of perturbative algebraic quantum field theory (pAQFT) \cite{DF,BDF,Book} in the hope that such a combination of 
results
might bring the whole programme of CQFT forward. We investigate the massive Sine Gordon model in the ultraviolet finite regime 
  ($a^2<4\pi /\hbar$ in  \eqref{eq:sinegordon}) in thermal states and prove that the adiabatic (infinite volume) limit exists. 
Several results regarding Sine-Gordon models have been obtained in the past,
in particular, for finite volume or non vanishing mass,
the construction of the Sine-Gordon model in the Euclidean case for $a^2<4\pi/\hbar$ has been obtained
by Fr\"ohlich and Seiler in 
\cite{Froe76,FS} and
 it has been extended 
to $a^2< 8\pi/\hbar$ by Dimock and Hurd in \cite{DimockHurd,DimockHurd1},
see also \cite{BGN,NRS} for earlier works using a different approach.
Further recent works on the subject are reported in \cite{BFM1}, \cite{bauerschmidt2022coleman}.
%
%
We stress that while most work done in the past used the auxiliary Euclidean framework,  we are able to remain purely in the Lorentzian framework. This, in fact, turns out to simplify arguments and proofs. Most notably, we will see how the  time slice property~\cite{ChF08} and causality allow us to reduce the problem of taking the adiabatic limit to a one-dimensional problem, thus 
simplifying the estimates in the cluster expansion used in the Euclidean case \cite{FS}.
Our construction might also shed new light on Osterwalder-Schrader positivity and the Wick rotation. It would be interesting to investigate the model's analytically continued correlation functions, calculated in thermal states, and study the resulting Euclidean theory.
Finally, it should be noted that our approach to the construction of the algebra of observables carries over to any globally hyperbolic spacetime. 

To be specific, we construct a quantized $*-$algebra of the massive Sine-Gordon theory on the two-dimensional Minkowski space $M = (\mathbb{R}^2,\eta)$ with the Minkowski metric $\eta= (-,+)$.
We recall that the classical Lagrangian of the Sine-Gordon theory is 
\begin{equation}\label{eq:sinegordon}
\mathcal{L} :=   -\frac{1}{2}\partial\varphi \partial\varphi -\frac{m^2}{2} \varphi^2 - \lambda \cos(a\varphi)
\end{equation}
where $\varphi$ represents the Sine-Gordon field, $m$ is the mass of the field, $\lambda$ is a coupling constant and $a$ is a parameter. We assume  
$a^2 < 4\pi/\hbar$, so the perturbative series does not require renormalization (ultraviolet finite regime). 
%
%
Working in the framework of pAQFT, the quantization of 
theory is achieved in two separate steps. First, one constructs a $*-$algebra $\mathcal{A}_I$ which contains fields observables as a formal power series deformation of a classical Poisson algebra. In a second step, representations of this algebra are considered, via the introduction of states (normalized positive linear functionals on $\mathcal{A}_I$). This is fundamentally different from other approaches to QFT, where a vacuum state is assumed along with the fields. The dynamics enters the construction of both the algebra and the states. In this paper, we consider equilibrium states (KMS states) and already in the case of free theories, the representations obtained at different temperatures are not equivalent.

In the case of fields which satisfy a nonlinear equation of motion,  we shall use perturbation theory to construct interacting fields. The idea is then to construct $\mathcal{A}$, the $*-$algebra of free fields (the case $\lambda=0$ in \eqref{eq:sinegordon}), 
we then find a map which permits to represent interacting fields, namely elements of $\mathcal{A}_I$ in $\mathcal{A}$. 
This map exists for the off-shell algebra, namely $\mathcal{A}$ is constructed by implementing the free, canonical commutation relations but without assuming that linear fields satisfy the equation of motion. In other words, in $\mathcal{A}$ it is possible to find an ideal $\mathcal{I}$ generated by the free equation of motion applied to test functions
and instead of taking the quotient of $\mathcal{A}$ with $\mathcal{I}$ at the beginning, we shall consider states over $\mathcal{A}$ which vanish over $\mathcal{I}$, hence we are effectively taking the quotient at the very end of the quantization process.

The map which permits to represent off-shell interacting fields within the algebra of linear fields is constructed using ideas of Bogoliubov and St\"uckelberg (\cite{BS,SR}
and in particular constructing and using the time ordered exponential of the interacting Lagrangian called the $S$-matrix.
Usually, in higher dimensions, the $S-$matrix is known as a formal power series in the coupling constant $\lambda$ or in the case of polynomial interaction Lagrangian as a power series in $\hbar$. 
For Sine-Gordon theory \eqref{eq:sinegordon} in two dimensions, this $S$-matrix can be shown to converge for various choices of the parameters $m$, $\lambda$ and $a$ and in various forms, see e.g. \cite{DimockHurd,DimockHurd1} for the case $a< (8\pi/\hbar)$. 

To obtain the interacting fields perturbatively, we shall construct the $S-$matrix of the interacting Lagrangian 
\begin{equation}\label{eq:LI}
\mathcal{L}_I = \lambda \cos(a\varphi)
\end{equation}
as the time ordered exponential of $\mathcal{L}_I$. It is long known that $S$ converges in the case of a Euclidean theory \cite{Froe76,FS}. 
The first step in this construction is the restriction of the support of the interaction Lagrangian multiplying it with a suitable cutoff function $g$. 
It is also known that $S$ converges  for the massless two-dimensional Lorentzian case provided the support of  $g\mathcal{L}_I$ is restricted to a sufficiently small region \cite{BR16, BFR17}. In \cite{BFR17}, it is proved that in the GNS representation of the Derezinski-Meissner state \cite{DM06}, the series that defines the $S-$matrix converges 
in the strong operator topology to a unitary operator.  
 
In the first part of this paper we extend this construction to the interacting Lagrangian supported on arbitrary compact regions (Theorem \ref{thm:convergence-V} below). In order to achieve this result we make use of techniques like conditioning and inverse conditioning developed for the Euclidean case, see e.g. (\cite{Froe76, FS, GR75a, GR75b}) and extend them to the Lorentzian case. 
These results are used to prove   absolute  convergence of the $S-$matrix when the interaction Lagrangian is supported on a small region   uniformly on the field configuration. See Theorem \ref{thm:convergence}, which is a direct extension of the results of \cite{BR16}.
Observe that the analysis performed in this part of the work can be extended also to generic globally hyperbolic two dimensional spacetime.
We then proceed to analyze adiabatic limits, i.e. limits where the cutoff function $g$ is removed taking the limit $g\to1$ in a suitable sense. As a first step,
in Theorem \ref{thm:convergence-V} we 
prove that the $S-$matrix converges also when $g$ is supported on a generic domain.

In section \ref{se:thermal}, we investigate the equilibrium state for the interacting theory with positive mass at generic temperature $\beta^{-1}$. When the interaction Lagrangian is of compact support, the construction is an extension of the results of Fredenhagen and Lindner \cite{FL} which are based on ideas of Araki \cite{Araki73}. The investigation of the adiabatic limit, however, is rather difficult for these states and requires a detailed analysis of the decay properties of the free $n-$point functions at large spatial separation. 
To present this part in details, we start by discussing some properties of thermal states for the free theory in subsection 
\ref{se:thermal-state-free-theory}. In particular, we argue that similar results as those given in Theorem \ref{thm:convergence-V} can be obtained also for an algebra whose product is constructed with the two-point function of the free equilibrium state.
The Fredenhagen-Lindner construction uses two key properties, namely, the time-slice axiom for the interacting theory \cite{ChF08} and the
fact that 
(smooth) alteration of the interaction Lagrangian in the past of $\Sigma_0$, the time $t=0$ hypersurface, results in an unitary transformation of the algebra of interacting observables supported in the future of $\Sigma_0$ \cite{BF0, HW, HW01}. These observations and the way in which they are used to compute adiabatic limits, are recalled in subsection  \ref{se:strategy}. Further properties of the interacting equilibrium state, and in particular their independence on the time cutoff function are given in  subsection \ref{se:thermal-state-interacting-theory}.
The interacting equilibrium state $\omega^{\beta,V}$ is presented in \eqref{eq:interacting-state} and it is given in terms of the equilibrium state of the free theory $\omega^\beta$   and of $U(i\beta)$, the relative partition function associated to the free and interacting dynamics , as
\[
\omega^{\beta,V}(F) = \frac{\omega^\beta(F U(i\beta))}{\omega^\beta(U(i\beta))}, \qquad F \in \mathcal{A}_I. 
\]
 
We then discuss the construction of the generating functional $G(f)$ given in \eqref{eq:generating-functional} of the equilibrium state for the interacting theory at finite temperature when the interaction Lagrangian is localized in a compact region with cutoff function of the form  $g=\chi h$  with $\chi$ and $h$ the time and space cutoffs, respectively, as discussed in  subsection \ref{se:thermal-state-interacting-theory}.
$G(f)$ is obtained from the ratio above
with $F$ the exponential of the field localized by a cutoff function $f$, and with numerator and denominator given in \eqref{eq:FandZ}.
Furthermore, the explicit form of the generating functional can be simplified with the following observations.
The state evaluated on observables supported in the future of $\Sigma_0$ does not depend on the particular form of the time cutoff function $\chi$ in the past of $\Sigma_0$.
If the time cutoff function is chosen to be a Heaviside step function, no further divergences are introduced in $G(f)$. Furthermore, in this case, the corresponding interacting Hamiltonian given in \eqref{eq:interacting-hamiltonian}
 is supported at time zero only. 
Using the time slice axiom of the free theory, we can also move the support of $f$ to be closer and closer to the hypersurface $\Sigma_0$ without altering the form of the generating functional $G(f)$. 
Therefore, instead of considering $e^{i\Phi(f)}$ in  \eqref{eq:generating-functional} with general $f$, we can equivalently compute expectation values of fields entirely supported at the hypersurface $\Sigma_0$ because  \eqref{eq:time-slice-free} holds up to contributions which vanish on shell and the state satisfies the free equation of motion. 

Taking into account these observations, we have that the explicit form of the generating functional can be obtained using \eqref{eq:Fexp-Zexp} which is an expression given in terms of objects entirely supported on $\Sigma_0$. 
In this expression, the covariance   $C^\beta_m\in \mathcal{D}'((0,\beta)\times \Sigma_0)$ 
appears. It corresponds to the analytic continuation of the two-point function of the free KMS state given in \eqref{eq:defC}.
In \eqref{eq:Fexp-Zexp}, the spatial support of the interacting Hamiltonian (at time $t=0$) is restricted by the spatial   cutoff function $h$.
The last step in the construction of the equilibrium state for the Sine-Gordon theory consists in taking the limit where the spatial cut off function $h$ tends to $1$ everywhere on $\Sigma_0$. 
To study this limit, we adapt to one dimension 
the cluster expansion originally proposed by 
Glimm, Jaffe and Spencer \cite{GJS} for a two-dimensional Euclidean theory and used in the context of Euclidean Sine-Gordon models by Fr{\"o}hlich and Seiler \cite{FS}.
%
We recall that the cluster expansion, originally proposed in \cite{GJS}, has been improved and generalized over the years in various works, see e.g. \cite{BF2, BK, AR95}.

In order to apply these techniques, the decay properties of the covariances with or without boundary conditions are discussed in section \ref{se:covariance-bc}. 
The details of the cluster expansion are given in section \ref{se:cluster-expansion} where we reduce to the one-dimensional case the cluster expansion given in \cite{GJS} and in \cite{FS}.
In particular, in Proposition \ref{pr:cluster-exp-00} and Proposition \ref{pr:cluster-exp} we summarize some properties of the cluster expansion and in Theorem \ref{thm:adiabatic-limit}, we state the validity of the adiabatic limit in our context.
In order to prove the finiteness of the adiabatic limit, the decay properties for $h$ with large support as stated in Lemma \ref{le:cluster-convergence-2} are crucial. The proof of this Lemma descends from some more technical lemmata collected in section \ref{se:lemmata}.

\subsection{Functional Approach}
In the functional approach to QFT, the observables of the theory are constructed as (formal power series of) functionals 
on smooth field configurations. 
The set of these objects with a certain regularity is then equipped with suitable products and involutions in order to obtain the $*-$algebra of field observables. The commutation relations among fields are encoded at algebraic level in the product.
Further details on this construction can be found in \cite{BF0, BDF, BFV, HW, HW01, LesHauches}, see also the recent books on the subject \cite{Book,Due19}.
Here we recall the main definitions. We shall formulate the formalism more generally than needed here, assuming $M$ to be a general globally hyperbolic manifold.

Starting point is the choice of possible field configurations. In our case of spin 0 fields, the set of {\bf field configurations} is simply the set of real-valued smooth functions
\[
\mathcal{C} := \left\{\varphi \in C^{\infty}(M;\mathbb{R})  \right\}. 
\]
Observables of the theory are contained in the set of {\bf microcausal functionals} which is the set of smooth functionals over $\mathcal{C}$ that are of compact support and whose wave front set is microcausal, i.e.
\[
\mathcal{F}_{\mu c} := \left\{F:\mathcal{C} \to \mathbb{\mathbb{C}} \;\left|\;   F^{(n)} \in \mathcal{E}'(M^{n}),  \text{WF}(  F^{(n)}) \cap  (\overline{V}^+ \cup \overline{V}^-)   = \emptyset   \right. , n \in \mathbb{N}\ \right\},
\]
where $F^{(n)}$ denotes the $n$-th functional derivative of $F$,  $\text{WF}(d)$ denotes the wave front set of a distribution $d\in\mathcal{D}'(M^n)$, see e.g. \cite{Hoer1} for a definition,  and where $V_+$ and $V_-$ are the set of future-pointing and past-pointing covectors in  $T^*M^n$, respectively.
The set of {\bf regular functionals} over $\mathcal{C}$ is the subset
\[
\mathcal{F}_{\text{reg}} := \left\{F \in \mathcal{F}_{\mu c} \;\left|\;  F^{(n)} \in C^{\infty}_0(M^n) , n \in \mathbb{N}\ \right.\right\}. 
\]
Another relevant subset of $\mathcal{F}_{\mu c}$ is the set of {\bf local functionals} 
\[
\mathcal{F}_{\text{loc}} :=  \left\{F \in \mathcal{F}_{\mu c}   \;\left|\;  F^{(1)} \in C^{\infty}(M),\; \text{supp}F^{(n)}\subset D_n\subset M^n   \right.
, n \in \mathbb{N}\ \right\},
\]
where $D_n$ denotes the diagonal in $M^n$. A functional $F\in\mathcal{F}_{\mu c}$ is said to be {\bf polynomial} if only finitely many functional derivatives are non-vanishing. We denote by $\mathcal{F}^p_{\mu c}$ the set of polynomial microcausal functionals and by $\mathcal{F}_{\text{reg}}^p$ the set of regular polynomial microcausal functionals.

Restricting the domain of integration by a compactly supported smooth function (``smearing'' as it is usually called), the interaction Lagrangian and fields themselves,
gives rise to elements of $\mathcal{F}_{\text{loc}}$ in the following sense:
The smeared {\bf linear field} is the functional $\Phi(f)$ on $\mathcal C$
\begin{equation}\label{eq:linear-field}
\Phi(f):\varphi \mapsto \langle f, \varphi\rangle =    \int_M  f(x)\varphi(x) d\mu_x ,  \qquad f \in C^\infty_0(M)
\end{equation}
where $\mu$ is the measure induced by the metric of the spacetime $M$. The linear field is obviously  a local regular functional. The notation usually suppresses the $\varphi$-dependence, i.e. we write $\Phi(f)$ instead of $\Phi(f)(\varphi)$. 
In a similar way, the interaction Lagrangian gives rise to a local functional, which we define by means of the smeared {\bf vertex operator}, i.e. the exponential of the field
\begin{equation}\label{eq:vertex}
V^g_a :=\int_M  e^{\mathrm{i} a\varphi(x)}   g(x)  d\mu_x    ,\qquad a\in \mathbb{R},\; g\in C^\infty_0(M),
\end{equation}
where again, in the notation  we suppress the $\varphi$-dependence.

If we equip $\mathcal{F}_{\mu c}$ or $\mathcal{F}_{\text{reg}}$ with suitable products and with a suitable involution, we obtain the classical or quantum algebra which contains the observables of the free theory.
More precisely, the involution we are considering coincides with the complex conjugation, i.e.
\[
F^{*}(\varphi) :=  \overline{F(\varphi)}.   
\]
We now equip $\mathcal F_{\mu c}$ with the pointwise product 
defined as
\[
F\cdot G = \mathcal{M} (F\otimes G)
\]
where $\mathcal{M}$ is the pullback of the map $\imath:\mathcal{C}\to\mathcal{C}\times \mathcal{C}$, $\imath(\varphi):=(\varphi,\varphi)$ on $\mathcal{F}_{\mu c}\otimes\mathcal{F}_{\mu c}$. Effectively, this means that $F \cdot G (\varphi) = F(\varphi) G(\varphi)$ with pointwise multiplication in $\mathbb C$.
The algebra $\mathcal{A} = (\mathcal{F}_{\mu c},\cdot )$ is thus a commutative $*-$algebra. We equip it with a topology by saying that a sequence $A_n \in \mathcal{A}$ converges to a functional $A$ on $\mathcal C$ 
if for any field configuration $\varphi$ and for any $k\in \mathbb{N}$, the functional derivatives $A_{n}^{(k)}(\varphi)$ converge to $A^{(k)}(\varphi)$ in the topology of distributions of compact support. 

Observe that the field configurations which are considered in the construction of this algebra are not necessarily solutions of any field equation, hence this algebra is called \textbf{off-shell algebra of free classical fields}. The \textbf{on-shell algebra of fields} is obtained by taking the quotient with respect to $\mathcal{I}$, the ideal formed by elements which vanish when evaluated on a field configuration $\varphi$ that is a solution of the chosen equation of motion. We avoid computing this quotient as it is easier to work off shell to obtain the correlation functions of the theory. We will simply ensure that the chosen state that will define the representation of the off-shell algebra, hence the correlation functions, is compatible with taking the quotient, i.e. that the ideal $\mathcal{I}$ generated by the equation of motion is in the kernel of the state.

The canonical quantization of the free theory  is realized by deforming the pointwise product to a non-commutative product in terms of a formal power series with a formal parameter $\hbar$, with $\hbar=0$ corresponding to the classical
theory (pointwise product). 
The deformed products we consider are of the from 
\begin{equation}\label{eq:product-exponential}
F\star_w G = \mathcal{M}\; e^{\Gamma_w}(F\otimes G), \qquad   \Gamma_w = \int d\mu_xd\mu_y w(x,y) \frac{\delta}{\delta \varphi(x)} \otimes \frac{\delta}{\delta \varphi(y)}
\end{equation}
where 
$w\in \mathcal{D}'(M^2)$ is a suitable distribution of order $\hbar$. 

In order to encode the canonical commutation relations,
making the deformed product the correct  {\bf ordinary quantum product}, $w$ is chosen to be equal to $\frac{i}{2}\hbar\Delta$ where $\hbar$ is the Planck constant and $\Delta$ is the retarded minus advanced fundamental solution of the free equation of motion, also called the Pauli-Jordan commutator function or the {\bf causal propagator}. This choice makes the product $\star_{i\hbar\Delta/2}$ (also denoted by $\star$ or even omitted when this does not create confusion),
well defined on polynomial regular functionals and that each term in the series is again (up to a power in $\hbar$) a regular polynomial functional. Observe that the resulting series' only contain finitely many terms and hence, we can treat $\hbar$ as a {\em{number}}, not as a formal parameter, justifying the notation $(\mathcal F_{\normalfont\text{reg}}^p, \star)$ instead of $(\mathcal F_{\normalfont\text{reg}}^p[[\hbar]], \star)$ for the  $*-$algebra of \textbf{free quantum  fields}. In particular, we have implemented the canonical commutation relations, 
\begin{equation}\label{eq:CCR}
[\Phi(f_1),\Phi(f_2)]_\star =   \mathrm{i}\hbar \langle f_1, \Delta f_2    \rangle  , \qquad f_1,f_2 \in C^\infty_0(M).    
\end{equation}
We observe that $\star_{i\hbar\Delta/2}$ cannot be directly extended to microcausal functionals because the singularities of $\Delta$ are not necessarily compatible with the singularities of the functional derivatives of elements of $\mathcal{F}_{\mu c}$.
Hence, elements like the smeared interaction Lagrangian \eqref{eq:LI}, the vertex operators \eqref{eq:vertex} and even non-linear polynomial local functionals are not contained in $\mathcal F_{\normalfont\text{reg}}^p$. 
On the other hand, physically relevant states like the vacuum defined on the quantum algebra $(\mathcal{F}^p_{\text{reg}},\star)$  cannot be extended to  
$\mathcal{F}_{{\mu c}}$.
In order to overcome both problems, a normal ordering procedure is used.
The procedure of \textbf{normal ordering} on regular functionals is to consider formal power series
\[
\no{F}_w   =     \alpha_w^{-1}F \ \in \mathcal{F}_{\text{reg}}[[\hbar]], \qquad F\in\mathcal{F}_{\text{reg}}, 
\] 
where 
\begin{equation}\label{eq:deformation}
\alpha_{w}F = e^{{\tilde\Gamma}_{w}} F,\qquad     \tilde\Gamma_{w} = \frac{1}{2}\int d\mu_xd\mu_y w(x,y) \frac{\delta^2}{\delta \varphi(x) \delta \varphi(y)}
\end{equation}
and where $w\in \mathcal{D}^\prime(M\times M)$, the so-called \textbf{2-point function}, is a suitable distribution, again of order $\hbar$, which 
is assumed to have the following properties: \label{page:admissiblew} 
$w$ is a solution of the free equation of motion up to smooth remainder term,
its antisymmetric part coincides   with
$\frac{i}{2}\hbar\Delta$ so that the canonical commutation relations \eqref{eq:CCR} are still valid   if $\star$ is replaced by $\star_w$, and moreover, the symmetric part of $w$ is real and it is chosen in such a way that  the wave front set of $w$ satisfies the so-called microlocal spectrum condition.
We recall that a bi-distribution $u$ satisfies the {\bf microlocal spectrum condition} if 
\begin{equation}\label{eq:microlocal}
\text{WF} (u)  =  \{ (x,y;k_x,k_y) \in T^{*}M^2 \setminus \{0\} \mid (x,k_x)\sim(y,-k_y), k_x \triangleright 0 \}
\end{equation} 
where $(x,k_x)\sim(y,k_y)$ means that  $x$ and $y$ are connected by a null geodesic $\gamma$, with the vectors  $\eta^{-1}k_x$  and $\eta^{-1}k_y$ being tangent to $\gamma$  at $x$ and $y$, respectively, and with $k_y$ being the parallel transport of $k_x$ along $\gamma$, and where  $k_x \triangleright 0$ means that $k_x$ is future pointing.  
Observe that the existence of such functions $w$ is guaranteed by a deformation argument on general globally hyperbolic spacetimes, see e.g. \cite{FNW}.

A state $\omega$ on $(\mathcal{F}^p_{\text{reg}},\star)$ whose $2-$point function $\omega(\Phi(f)\star\Phi(g))$ satisfies the microlocal spectrum condition given in \eqref{eq:microlocal} and whose truncated $n$-point functions with $n\neq 2$ are smooth functions, is called \textbf{Hadamard state} \cite{Rad,BFK96}.
In this paper we shall always consider normal ordering computed with respect to $w$ given by the massive vacuum two-point function, see subsection~\ref{subsection:States}. Normal-ordered fields have finite expectation values when tested on Hadamard states, however, as mentioned above, they are formal power series, in general. 
Restricting to polynomial regular functionals, on the other hand, $\alpha_w$ introduced in \eqref{eq:star-isomorphism} maps to honest regular functionals (if, as explained above, $\hbar $ is treated as a number) and in fact, realizes a $*$-isomorphism between $(\mathcal{F}^p_{\text{reg}},\star)$ and $(\mathcal{F}^p_{\text{reg}},\star_w)$ where $\star_w$ is again given in terms of formula \eqref{eq:product-exponential}.
Contrary to $\star$, the product in  $(\mathcal{F}^p_{\text{reg}},\star_w)$ can be straightforwardly extended to polynomial local functionals.
Furthermore, taking into account this extension, $(\mathcal{F}^p_{\text{reg}},\star)$ can also be extended accordingly, but the resulting space is not a space of (formal power series over) functionals anymore. It includes formal objects like $\no{F}_w$ with $F\in \mathcal{F}^p_{\text{loc}}$, which are constructed as limits of suitable sequences in regular functionals, see e.g. \cite{BDF}, and
then mapped to ordinary polynomial local functionals by $\alpha_w$ (i.e. $\alpha_w(\no{F}_w)=F$, understood as an element of
  $\mathcal{F}^p_{\text{loc}})$.

Now, the star-products $\star_w$ constructed via a 2-point function $w$ with the properties stated above, are well defined on polynomial microcausal functionals, but in fact they can be further extended to other elements of $\mathcal{F}_{\mu c}$, such as the exponential of the fields. 
In this section, we will denote by $\mathcal F$ the set of functionals on which $\star_w$ is defined and does not yield only a formal power series, but an honest functional in $\mathcal F$ if $\hbar$ is treated like a number. Then $\mathcal{A}=(\mathcal{F},\star_w)$ is a unital $*$-algebra which contains observables of the theory, such as exponentials of the field. We will discuss this in  detail in subsection~\ref{se:algebras-star-product} where we shall define a set $\mathcal{F}^V$ of functionals which contains $V_a^g$ and which forms an algebra with respect to the  various products used in the present paper. 
Observe that for different choices of $w$, the resulting products $\star_w$ are equivalent 
to each other in the sense of star products 
and thus produce isomorphic algebras. The isomorphism of $(\mathcal{F},\star_{w^\prime})$ with $(\mathcal{F},\star_{w})$ is realized by
\begin{equation}\label{eq:star-isomorphism}
\alpha_{w-w^\prime}F = e^{{\tilde\Gamma}_{w-w^\prime}} F,
\end{equation}
where $\tilde\Gamma_{v}$ is given in \eqref{eq:deformation} and $w-w^\prime$ is smooth.

\subsection{States}\label{subsection:States}
A state is a positive normalized   linear functional on $({\mathcal{F}^p_{\text{reg}},\star})$ and, if it is of Hadamard form, it can be extended to normal-ordered fields and it can thus be represented as a
  positive normalized linear functional on  $\mathcal{A}=(\mathcal{F},\star_w)$.  
%

  If $w$ is a positive bidistribution,
an example of a state on $\mathcal{A}=(\mathcal{F},\star_w)$ is given by evaluation at the vanishing field configuration $\varphi=0$, 
\[
\omega(A) := {\rm{ev}}_0(A):= A(0), \qquad A \in (\mathcal{F},\star_w) \ , 
\]
whose 2-point function $\omega(\Phi(f){\star_w}\Phi(g))$ coincides with $w$. This state is quasifree, i.e. its $n$-point functions are determined by the 2-point functions, and vanish for odd $n$. Since a generic element $A\in (\mathcal{F},\star_w)$ is not invariant under the action of the $*-$homomorphism $\alpha_{w-w'}$ that intertwines  $(\mathcal{F},\star_w)$ and $(\mathcal{F},\star_{w'})$, the state   ${\rm{ev}}_0$  depends on the choice of $w$.
%

\subsection{Bogoliubov map}\label{se:BogoliubovMap}

A direct construction of the product among interacting fields is usually not available, however, it is possible to use ideas of Bogoliubov and St\"uckelberg (\cite{BS,SR}, see \cite{Book,Due19} for reviews)
to represent interacting fields as power series with coefficients in $\mathcal{F}_{\mu c}$ (with formal parameter given by the coupling constant). 

The starting point of this construction is in the splitting of the Lagrangian in a free and an interacting part. 
In two dimensions there are  interacting Lagrangians like $\mathcal{L}_I$ in \eqref{eq:LI}, which give rise to series which are in fact convergent at least in the weak sense, namely when evaluated on a generic field configuration $\varphi\in\mathcal{C}$.  
As a first step, in order to avoid infrared divergences, one obtains the interacting potential as a local field smearing $\mathcal{L}_I$ with a suitable test function $g\in C^\infty_0(M)$,   $g\geq 0$,
$ V^g := \frac{1}{\lambda}\int_M  \mathcal{L}_I   g   d\mu_x$.
In our case 
\begin{equation}\label{eq:interacting-potential}
 V^g :=  \frac{1}{2}  \left( V_a^g  + V_{-a}^g   \right)  = \int  \cos(a\varphi(x))   g(x)   d\mu_x, \qquad g\in C^\infty_0(M),
\end{equation}
where $V^g_a$ is given in \eqref{eq:vertex}.
We then have to introduce the time ordering map $\mathcal{T}$ on multilocal fields by means of a family of multilinear maps 
\[
\mathcal{T}_n : \mathcal{F}_{\text{loc}}^{\otimes n} \to \mathcal{F}_{\mu c}, \quad \textrm{with} \ \mathcal T_0=1, \ \mathcal T_1=\id \, ,
\]
where $\TT({\prod_{i=1}^n} F_i)=\mathcal T_n({\prod_{i=1}^n}^\otimes F_i)$ for $F_i \in \mathcal F_{\textrm{loc}}$, see~\cite{FR}.
 The $\mathcal{T}_n$ are such that $\TT$ is symmetric and respects the causal factorisation property, which requires that if $\{F_i\}\subset \mathcal{F}_{\text{loc}}$ and $\{G_j\}\subset \mathcal{F}_{\text{loc}}$ 
\[
\mathcal{T}({\prod_i} F_i  {\prod_j} G_j  ) =   \mathcal{T}({\prod_i} F_i)  \star   \mathcal{T}({\prod_j} G_j  ),  \qquad       F_i  \gtrsim G_j ,\; \forall i,j.
\]
where $F  \gtrsim G $ means that $\text{supp} (F)$ does not intersect $J^{-} \supp (G)$, the past of $\supp (G)$.

Usually, the construction of a time ordering map is not straightforward and it requires a renormalization procedure due to divergences which are present when the factors have overlapping supports.
In two-dimensional theories, the divergences present in the propagators are of logarithmic nature and thus these divergences can be avoided at least
when the local functionals do not contain derivatives of the field and are polynomial in the field configurations. As we shall see in subsection \ref{se:algebras-star-product} and as was observed long ago for the Euclidean case, the same is true also in the case  of the smeared interaction Lagrangian $\int \mathcal{L}_I g d\mu$ \eqref{eq:interacting-potential} or $V_a^g$ 
 in \eqref{eq:vertex} with a sufficiently small $a$.
Hence, in these cases, 
the causal factorisation property determines completely the time ordering map.
 
More concretely, on polynomial local functionals which are constructed with fields that do not contain derivatives,
\[
\mathcal{F}_{\text{loc}}^{\Phi} = \left\{ F \in \mathcal{F}^{p}_{\text{loc}} \;\left|\;   F = \int A(\varphi(x))   f(x) d\mu_x, \; f\in C^\infty_0(M;\mathbb{R}),\; A\in \mathcal{R}_\mathbb{R}[X] \right.\right\},
\] 
where $\mathcal{R}_\mathbb{R}[X]$ is the set of polynomials in   one variable $X$ with coefficients in $\mathbb{R}$,
  the time ordering map does not require renormalization
and induces a well defined associative product $\cdot_T$ when $M$ is 2-dimensional Minkowski space, see
\cite{FR}. Observe that a definition of the time-ordered product consistent with causal factorisation and with the first-order
condition $\mathcal{T}_1(F) = F$ is
\begin{equation}\label{eq:time-ordering}
F_1\cdot_{T} \dots \cdot_{T}  F_n := 
F_1\cdot_{w^{\text{F}}} \dots \cdot_{w^{\text{F}}}  F_n := 
\mathcal{T}^{w^{\text{F}}} (F_1\otimes \dots \otimes  F_n) = \mathcal{M} e^{\sum_{i<j}\Gamma_{w^{\text{F}}}^{ij}}   \left( F_1\otimes \dots \otimes  F_n \right),   \qquad F_i \in \mathcal{F}_{\text{loc}}^{\Phi}
\end{equation}
where $\Gamma_{w^{\text{F}}}^{ij} = \int d\mu_xd\mu_y w^{\text{F}}(x,y) \frac{\delta}{\delta \varphi^i(x)} \otimes \frac{\delta}{\delta \varphi^j(y)}$
and $\delta/\delta\varphi^i$ is the functional derivatives acting on the $i-$th element of the tensor product.
Here, $w^{\text{F}}$ is the Feynman propagator associated to a two-point function $w$ as described on page \pageref{page:admissiblew}, 
\begin{equation}\label{eq:feynmann}
w^{\text{F}}= w \mathrm{i}\hbar\Delta^A, 
\end{equation}
where $\Delta^A$ is the advanced fundamental solution of the free equation of motion.
Notice that $\mathcal{A}^\Phi = (\mathcal{F}_{\mu\text{c}}^{\Phi}, \cdot_{w^{\text{F}}})$ forms a commutative algebra. 
Here $\mathcal{F}_{\mu\text{c}}^{\Phi}$ is the smaller subset of $\mathcal{F}_{\mu c} $ which contains $\mathcal{F}_{\text{loc}}^{\Phi}$ and where $\cdot_{w^{\text{F}}}$ is defined.
As we shall see in subsection \ref{se:algebras-star-product} the time ordered product can be extended to more singular objects,   like products of $V^g_a$, without using any renormalization procedure.
 
Once the time ordering map is at disposal, we may introduce the {\bf $S-$matrix} of the local functionals as the time ordered exponential of its argument, namely 
\[
S(V)  := \sum_n  \frac{\mathrm{i}^{n}}{n!} \frac{\lambda^n}{\hbar^n} \mathcal{T}^{w^{\text{F}}}
(V^{\otimes n}), \qquad V\in\mathcal{F}_{\text{loc}}.
\]
Observe that by definition of the time ordering, the $S-$matrix is a formal power series in $\lambda$ and a  Laurent series in $\hbar$.
The {\bf  relative $S-$matrix} is then defined as  
\[
S_V(F) := S(V)^{-1}\star S(V+F), \qquad V,F\in\mathcal{F}_{\text{loc}}
\]
where the inverse is calculated in the sense of formal power series.
Finally, with the relative matrix at disposal we get the {\bf Bogoliubov map} of a local functional $F$ which is 
\[
R_V(F) := \left. -\mathrm{i}\frac{\hbar}{\lambda} \frac{d}{d\mu}  S_V(\mu F) \right|_{\mu=0}, \qquad V,F\in\mathcal{F}_{\text{loc}}.
\]
Observe that $R_V(\Phi)$, for $\Phi$ as in \eqref{eq:linear-field}, solves the 
interacting equation of motion with interaction term given by $V^{(1)}$.
The algebra of interacting fields is thus the smaller subalgebra of $\mathcal{A}[[\lambda]]$ which contains $\{R_V(F)\mid F\in \mathcal{F}_{\text{loc}}\}$.
Notice that even if $F$ is a local functional, $R_V(F)$ in $\mathcal{A}$ is in general not local.

\subsection{Algebras generated by vertex operators and linear fields}\label{se:algebras-star-product}

In this section we shall see that some of the formal power series presented in the previous sections are actually ordinary elements of a well defined algebra when we specialize the construction for a $2-$dimensioanl Minkowski space   $M$.  

We are actually interested in finding a set of functionals which contains 
the vertex operators $V_a^f$ given in \eqref{eq:vertex}, the linear fields $\Phi(f)$ given in $\eqref{eq:linear-field}$
and their adjoints for $f\in C^{\infty}_0(M)$ and $a$ sufficiently small. 
Furthermore, the $\star_w-$products, $\cdot_{w^{\text{F}}}$ and $\cdot_{\text{Re}w}$ need to be internal in the set we are looking for, when $w$ has the form of two-point functions or Feynman propagators used in the next sections, namely $w\in \mathcal{D}'(M^2)$ is such that
\begin{equation}\label{eq:form-of-w}
\begin{aligned}
w(x,x') &= \hbar H_0^\mu(x,x') + \frac{\mathrm{i} \hbar}{2} \Delta_0(x-x') +s(x,x')\\
&:= -\frac{\hbar}{4\pi}  \log\left(\frac{(x-x')^2+i0^+(x-x')^0}{4\mu^2}\right)  +s(x,x')
\end{aligned}
\end{equation}
where $\mu$ is a fixed positive length scale, $H_0^\mu$ and $\Delta_0$ 
 
are real distributions whose explicit form is  
\begin{equation} \label{eq:form-of-H0-Delta}
\begin{aligned}
H_0^\mu(x,x') &= \mathcal{H}_0^\mu(x-x') =  - \frac{1}{4\pi} \log\left| \frac{(x-x')^2}{4\mu^2}\right|
= - \frac{1}{4\pi} \log\left| \frac{-{(x-x')^0}^2+{(x-x')^1}^2}{4\mu^2}\right|
, 
\\
\Delta_0(x) &= \frac{1}{2} \left(\theta(-x^0 -|x^1|)- \theta(x^0 -|x^1|)   \right) .
\end{aligned}
\end{equation}
Furthermore, 
 $s$ in \eqref{eq:form-of-w} is a continuous function on $M^2$ such that if $(x,x';k,k')\in\WF(w)$ then $k+k'= 0$. 
 Notice that $\Delta_0$ is the causal propagator of the massless theory and that, denoting by $\Delta_m$ the causal propagator of the massive theory, we have that
 $\Delta_m-\Delta_0 = (1-I_0(m\sqrt{x^2}))\Delta_0$ where $I_0$ is the modified Bessel function of the first kind.  $\Delta_m-\Delta_0$  is actually a continuous function whose wave front set satisfies the condition required for $s$ in  \eqref{eq:form-of-w}.
 Furthermore, 
as we shall see in section \ref{se:convergence-S(V)} and in section \ref{se:thermal}   the symmetric part of both the massive, massless, thermal two-point functions used in this paper have the form $\hbar H_0^\mu+s'$ where $s'$ is a symmetric continuous function whose wave front set satisfies the property required for $s$ in \eqref{eq:form-of-w}. Hence, in all these cases the two-point functions have the form given in \eqref{eq:form-of-w}. 
Further properties of $H_0^\mu$ will be discussed after \eqref{eq:distH}.

\begin{df}\label{df:FV}
Let $\mathcal{F}^{V}$ be the vector space formed 
by finite linear combinations of elements of the form 
\begin{equation}\label{eq:Gabnmf}
G_{a,n,m}(f) = \int_{M^{n+m}}   e^{\mathrm{i}\sum_{j}  a_j\varphi(x_j)} e^{- \sum_{i<j}  a_i a_j \hbar {H}_0^\mu(x_i,x_j) } 
f(X,Y) \prod_{i=1}^nd\mu_{x_i}       \prod_{j=1}^m   \varphi(y_j) d\mu_{y_j}
\end{equation}
where $a$ is in $(-\alpha,\alpha)^n\subset \mathbb{R}^n$ with $\alpha^2<\frac{4\pi}{\hbar}$ 
and $f$ is a compactly supported bounded function on $M^{n+m}$, $X=(x_1,\dots,x_n)$ and $Y=(y_1,\dots,y_m)$. Furthermore, $f$ is chosen in such a way that for every $l$,
\begin{equation}\label{eq:f-WF}
\WF(f)\subset L_{n+m} := \left\{(x_1,\dots, x_{n+m}; k_1, \dots, k_{n+m})\in T^*M^{n+m}\setminus \{0\}\;\left|\; \sum_{i=1}^{n+m} k_i=0\right.\right\}
\end{equation}
and
 ${H}_0^\mu$ is given in \eqref{eq:form-of-w}.
\end{df}
  \noindent
As we shall discuss later, we notice that ${H}_0^\mu$  which
appears in \eqref{eq:form-of-w} and in Definition \ref{df:FV} is  the symmetric Hadamard function of the massless field given in \eqref{eq:distH}. 
\begin{prop}\label{prop:FVinFmuc}
It holds that $\mathcal{F}^{V}\subset \mathcal{F}_{\mu c}$.
\end{prop}
\begin{proof}
Observe that operating as in the proof of   Theorem 6 of \cite{BR16}, (see also Lemma \ref{le:Cauchy-Determinant} below) we can show that $e^{- \sum_{i<j}  a_i a_j \hbar {H}_0^\mu(x_i,x_j) }$ is a locally integrable function in $M^n$. Hence, when multiplied with a   bounded compactly supported functions $f$ it gives a distribution of compact support (which is of order $0$).
By direct inspection we then have that for every $l$, $G_{a,n,m}(f)^{(l)}\in \mathcal{E}'(M^l)$.
Notice that $\varphi$ is smooth and if
$(x,x';k,k')\in\WF(H^\mu_0)$ then $k+k'=0$, hence, by an application of the second part of Theorem 8.2.10 in \cite{Hoer1},
 $\WF(e^{-\sum_{i<j}a_ia_j H^\mu_0 }f)\subset L_{m+n}$.
Now, by applying Theorem 8.2.12 in \cite{Hoer1} we get
 $\WF(G_{a,n,m}(f)^{(l)})\subset L_l$.
This implies that elements of $\mathcal{F}^{V}$ are microcausal functionals, actually, recalling the form of $L_l$ given in Definition \ref{df:FV} we have that  
\[ 
L_l \cap (\overline{V}^+ \cup \overline{V}^-)   = \emptyset 
\]
because if $(x;k_1,\dots, k_l)\in \overline V^+$, $\sum_i k_i$ is a future pointing covector which vanishes only if $k_i=0$ for every $i$ and thus  $(x;k_1,\dots, k_l) \not\in L_l$. Hence, $\WF(G_{a,n,m}(f)^{(l)})\cap \overline{V}^+  = \emptyset$. The same holds replacing $V^+$ with $V^-$. We have thus that $\mathcal{F}^V\subset \mathcal{F}_{\mu c}$.
\end{proof}

We shall now see that the various products we are using in the rest of the paper are internal in $\mathcal{F}^V$.
Hence, thanks to the following proposition, we shall define the $*-$algebra in which we operate as $\mathcal{A}_w^V:=(\mathcal{F}^V,\star_w)$, where $w$ is as in \eqref{eq:form-of-w}.

\begin{prop}\label{prop:algebra-well-defined}
The operations $*$, $\star_{w}$ in \eqref{eq:product-exponential}, $\cdot_{\text{\rm Re}w}$ and $\cdot_{w^{{\rm F}}}$ in \eqref{eq:time-ordering} with $w$ of the form given in \eqref{eq:form-of-w} and $w^{{\rm F}}$ in \eqref{eq:feynmann} are closed in $\mathcal{F}^{V}$.
Hence $\mathcal{A}_w^V=(\mathcal{F}^V,\star_w)$ is a $*-$algebra over the complex numbers. 
\end{prop}
\begin{proof}
The $*$ operation acts as complex conjugation on $G_{a,n,m}(f)\in\mathcal{F}^V$ and, hence, since ${H}^\mu_0$ is a real distribution, we have that $G_{a,n,m}(f)^*$ is in $\mathcal{F}^V$. The same holds for the linear combinations.

We prove now that $\star_{w}$ is internal in $\mathcal{F}^V$. 
Notice that 
$G_{a,n,m}(f)$
can be obtained as linear combinations of elements of the form 
$\left. \partial_{b_1}\dots \partial_{b_{m'}} G_{(a,b),n+m',0}(f') \right|_{b=0}$ with $a\in (-\alpha,\alpha)^n$, $b\in (-\alpha,\alpha)^{m'}$ and 
for various $m'\in\{0,\dots,m\}$ and with $f'$ which is constructed out of $f$ and $a$.
It is thus sufficient to check that the product of two elements of the from $G_{a,n,0}(f)$ and $G_{a',n',0}(f')$ is internal in $\mathcal{F}^V$. 

To this end, recall that $w$ given in \eqref{eq:form-of-w} has the form 
\[
w = \hbar H_0^\mu +  \frac{\mathrm{i} \hbar}{2} \Delta_0   + s
\]
where $s$ is a continuous function and $\Delta_0$ is a bounded function.

Applying \eqref{eq:product-exponential}
we get 
\[
G_{a,n,0}(f) \star_w G_{a',n',0}(f')
= \mathcal{M}\sum_{l\geq 0 } \frac{1}{l!} \Gamma^l_w(G_{a,n,0}(f) \otimes G_{a',n',0}(f')).
\]
Notice that this series converges because
$\frac{\delta}{\delta  \varphi(y)} e^{\mathrm{i} b \varphi(x)} = \mathrm{i} b e^{\mathrm{i} b\varphi(x)}  \delta(x,y)$
and hence
\[
\Gamma_w(e^{\mathrm{i}b \varphi(x)}\otimes e^{\mathrm{i}b' \varphi(y)}) 
= -(bb' w(x,y))^l e^{\mathrm{i}b \varphi(x)}\otimes e^{\mathrm{i}b' \varphi(y)}
\]
where $w$ is a smooth function almost everywhere, furthermore, $e^{-bb' w}$ is locally integrable at the discontinuity points for $b,b'\in (-\alpha, \alpha)$.
Hence, 
it holds that 
\[
G_{a,n,0}(f) \star_w 
G_{a',n',0}(f') = 
 G_{(a,a'),n+n',0}(h)
\]
where $h$ is a function on $M^{n}\times M^{n'}$ obtained from $f$ and $f'$ in the following way
\[
h(X,Y) := e^{- \sum_{i}\sum_j a_ia_j'(  w(x_i,y_j)-\hbar H_0^{\mu}(x_i,y_j) )} f(X) f'(Y)
\]
with $X=(x_1,\dots,x_n)\in M^n$ and $Y=(y_1,\dots, y_{n'})\in M^{n'}$.
Since $w-\hbar H^{0}_\mu$ is locally bounded, $f$ and $f'$ are compactly supported, we have that $h$ is a bounded compactly supported function. Furthermore, if 
$(x,x';k,k')\in\WF(w)$ then $k+k'=0$ and the same holds for $H^\mu_0$. 
Hence by the second part of Theorem 8.2.10  in \cite{Hoer1} about multiplication of distributions
we have that $\WF(h)\subset L_{n+n'}$ 
thus completing the proof that $\star_w$ is internal in $\mathcal{F}^V$. 
The proof that the same holds for  $\cdot_{{\rm Re}w}$ and $\cdot_{w^{\rm F}}$ can be obtained in a similar way noticing that
both
${\rm Re}w-H^\mu_0$ and $w^{\rm F} - H^\mu_0$ are locally bounded and that,
if $(x,x';k,k')\in\WF({\rm Re}w)$ or if $(x,x';k,k')\in\WF(w^{\rm F})$, $k+k'=0$. 
 
\end{proof}

We have seen that the vertex operators $V^g_a$ and the smeared interaction Lagrangian $V^g$ are elements of $\mathcal{A}^V_w$.
At the same time all the time ordered product of $n$ factors of the form $V^g$ gives an element of $\mathcal{A}^V_w$. 
Hence we may write $S(V^g)$ as a power series   in the coupling constants whose elements are in  $\mathcal{A}^V_w$. 
Later we shall prove that when evaluated  on a generic field configuration the series converges absolutely.

From now on we shall work directly in $\mathcal{A}^V_w = (\mathcal{F}^V,\star_w)$ where all the elements are well defined.
To simplify the notation, whenever it will not be strictly necessary, we shall denote $\mathcal{A}^V_w$ by $\mathcal{A}^V$. 
Furthermore, local fields in $\mathcal{F}^V$ are now understood as being normal ordered with respect to $w$. The quantum states on $(\mathcal{F}_{\text{reg}},\star)$   need to be pulled back to $(\mathcal{F}_{\text{reg}},\star_w)$ accordingly and there extended to   $\mathcal{A}^V$. In particular, the functional ${\rm ev}_0$ on $\mathcal{A}^V_w$ determines a well defined state which is the quasi free state constructed out of the two-point function $w$.

\section{Convergence of the S-matrix of the massive theory}\label{se:convergence-S(V)}

Consider $M$ a two dimensional Minkowski space with Lorentzian signature $(-,+)$.   Let us denote by $\omega_m$ the vacuum state. 
Hence, the canonical two-point function $w^\infty_m$ of the massive Klein Gordon field is given by the vacuum expectation value of the product of two fields, $w^\infty_m(f,g) :=  \omega_m(\Phi(f)\star\Phi(g))$ and it takes form
 
\[
w^\infty_m(f,g) = \lim_{\epsilon\to 0^+}  \hbar\int  f(y)g(x+y) \Delta^+_m(x+i\epsilon e^0)  d\mu_xd\mu_y, \qquad f,g\in C^\infty_0(M)
\]
where, $e^0=(1,0)$ and the regularized integral Kernel $\Delta^+_m$ is given in terms of the function
\begin{equation}\label{eq:massive2pt}
\Delta^+_m(x):=  \frac{1}{2\pi}K_0(m\sqrt{ x^2 }),
\end{equation}
here $K_0$ is the Mac Donald function also known as modified Bessel Function of the second kind
and
$x^2$ is the Lorentzian square $x^2= -{(x^0)}^2 + {(x^1)}^2$.

The Feynman propagator $\Delta^{\rm F}_m$ is obtained as the time ordered version of $\Delta^+_m$,
\[
\Delta^{\rm F}_m(x) = \Delta^+_m(x) + \mathrm{i} \Delta^A_m(x),
\]
with the advanced fundamental solution $\Delta^A_m$ of the massive Klein-Gordon equation.
With this canonical choice of the time ordering, the $S-$matrix of the interaction Lagrangian $V^g$ given in \eqref{eq:interacting-potential}  can thus explicitly be computed,
\begin{equation}\label{eq:S-matrix}
S(V^g) = \sum_n S_n^{m}(V^g)(\varphi) = \sum_{n} \frac{1}{n!}\left(\frac{\mathrm{i} \lambda}{\hbar}\right)^n\mathcal{T}^{\hbar\Delta^{\rm F}_m}_n({V^{g}}^{\otimes n})
\end{equation}
where the time ordered product of $n-$copies of the interaction Lagrangian can be further expanded as
\[
\mathcal{T}^{\hbar\Delta^{\rm F}_m}_n({V^{g}}^{\otimes n}) = \frac{1}{2^n}\sum_{k=0}^n  \begin{pmatrix} n \\ k \end{pmatrix} \mathcal{T}^{\hbar\Delta^{\rm F}_m}_n({V^g_a}^{\otimes k}\otimes {V^g_{-a}}^{\otimes(n-k)}).
\]
Recalling the definition of the smeared vertex operator $V_a$ given in \eqref{eq:interacting-potential} and 
  the map $\mathcal{T}$ given in \eqref{eq:time-ordering}
we get the following by directly applying the map $\mathcal{T}$ given in \eqref{eq:time-ordering} to  
$V^g_{a_1}\otimes \dots   \otimes V^g_{a_n}$
\[
\mathcal{T}^{\hbar\Delta^{\rm F}_m}_n(V^g_{a_1}\otimes \dots   \otimes V^g_{a_n}) = \int e^{\mathrm{i}\sum_k a_k\varphi(x_k)} e^{-\sum_{1\leq i<j\leq n} a_ia_j\hbar\Delta^{\rm F}_m(x_i,x_j)}g(x_1)\dots g(x_n) d\mu_{x_1}\dots d\mu_{x_n}
\]
where $a_k$ is equal to either $a$ or $-a$ 
and where with a little abuse of notation $\Delta^{\rm F}_m(x,y)=\Delta^{\rm F}_m(x-y)$.
Notice $e^{-\sum_{1\leq i<j\leq n} a_ia_j\hbar\Delta^{\rm F}_m(x_i,x_j)}$ is a locally integrable function because 
 $\Delta^{\rm F}_m=H^\mu_0+s$ where $s$ is a locally bounded function and the local integrability of  $e^{-\sum_{1\leq i<j\leq n} a_ia_j\hbar H^\mu_0(x_i,x_j)}$ can be proved as in the proof Theorem 6 of \cite{BR16} or as in Theorem \ref{thm:convergence} and in Lemma \ref{le:Cauchy-Determinant} below.
 Hence, since $\varphi$ is a smooth function and $g$ is compactly supported and smooth $\mathcal{T}^{\hbar\Delta^{\rm F}_m}_n(V^g_{a_1}\otimes \dots   \otimes V^g_{a_n})$ is well defined for every value of $\hbar$.
In the next we want to prove the convergence of the sum given in \eqref{eq:S-matrix} whenever $g$ has sufficiently small compact support. Later we shall remove the hypothesis on the smallness of the support. 
In the case $m=0$ the convergence of the corresponding $S-$matrix has been obtained in Theorem 6 in \cite{BR16}.

The idea we want to implement is the following. We notice that 
\begin{equation}\label{eq:stima-Re}
\left|\mathcal{T}^{\hbar\Delta^{\rm F}_m}_n(V^g_{a_1}\otimes \dots   \otimes V^g_{a_n}) \right|
\leq
\int e^{-\sum_{1\leq i<j\leq n} a_ia_j\hbar H_m(x_i,x_j)}g(x_1)\dots g(x_n) d\mu_{x_1}\dots d\mu_{x_n}
\end{equation}
where ${H_m}$ is the real part of $\Delta^{\rm F}_m$. Its form is discussed in \eqref{eq:Hm} and \eqref{eq:hm}.
We furthermore observe that the right hand side of \eqref{eq:stima-Re} is the expectation value in a state ${\rm{ev}}_0(F)=F(0)$
of the $n-$th power of $V^g$ computed 
with respect to the commutative product $\cdot_{\hbar H_m}$   and 
it is positive for every $n$.
Hence, we get
\begin{equation}\label{eq:step-local-g}
\left| S_n^{m}(V^g)(\varphi) \right| \leq 
\frac{1}{n!}\frac{\lambda^n}{\hbar^n}{\rm{ev}}_0(\mathcal{T}^{\hbar H_m}_n({V^g}^{\otimes n}))  
\end{equation}
where $\mathcal{T}^{\hbar H_m}$ is the map defined in \eqref{eq:time-ordering} constructed with   $w^F$ replaced by $\hbar H_m$. 
As discussed above, this is a deformation of the pointwise product which is like the time ordering map with $H_m$ at the place of $\Delta^{\rm F}_m$.

\begin{rem}\label{rem:estimates} If we have an estimate such as \eqref{eq:step-local-g} for every $n$,
and if the right hand side gives rise to a convergent series, the left hand side gives an absolutely  convergent series, and the result is less than or equal to the limit on the right hand side. Observe that we will state such inequalities even when considering formal power series, as a matter of notation, meaning that if the series converge, the inequality holds. A similar argument is used in Proposition \ref{pr:conditioning} to which we refer for further details. 
\end{rem}

Hence, in this sense, we have
\begin{equation}\label{eq:estimate-S(V)}
|S(V^g)|   \leq \sum_n |S^m_n(V^g)| \leq {\rm{ev}}_0\left( \exp_{\hbar{H_m}}\left(\frac{\lambda}{\hbar}V^g\right)\right)
\end{equation}
where $\exp_{\hbar H_m}$ is the time ordered exponential constructed with the time ordering map $\mathcal{T}^{\hbar H_m}$. 
We furthermore observe that 
$\exp_{\hbar H_m}\left(\frac{\lambda}{\hbar}V^g\right)$ is the exponential of $\frac{\lambda}{\hbar}V^g$ 
with respect to $\cdot_{\hbar H_m}$. Hence, we can thus bound $|S(V)|$ with a power series of $\frac{\lambda}{\hbar}V$ 
computed in the commutative $*-$algebra
$(\mathcal{F}^V,\cdot_{\hbar H_m})$
 evaluated in state ${\rm{ev}}_0$.
Now, adapting similar methods developed for the case of Euclidean background to the Lorentzian one, shall show how to control the expectation value of the exponential computed in the massive theory with the massless one.

\subsection{The Cauchy-Schwarz inequality, conditioning and inverse conditioning.}

We draw from the techniques of conditioning and inverse conditioning to control the massive theory by the massless one, cf. \cite{Froe76} and  \cite{GR75a,GR75b}. In Proposition \ref{pr:conditioning} and Proposition \ref{pr:inverse-conditioning2}   in the appendix we recall the main points of these techniques  as far as they are relevant for our constructions. 
We start by recalling that a bi-distribution $w:C_0^\infty(M;\mathbb{R})^2\to \mathbb{R}$ is positive if $w(f,f)\geq 0$ for every $f \in C_0^\infty(M;\mathbb{R})$. 
For such distributions, we have the following theorem in which we use the $*-$algebras generated by $V^g=\int \cos(a\varphi(x))g(x)d\mu_x$ whose elements are contained in $\mathcal{F}^V$. The construction of these $*-$algebras is also discussed in Remark \ref{rem:how-use-it}. The expectation values of exponentials of $V^g$ appearing below need to be understood as limits of the corresponding partial sums. Later we shall prove that the partial sums converge when evaluated in ${\rm{ev}}_{0}(F)=F(0)$.

\begin{thm}\label{th:control}
Let us consider two positive, real valued,  symmetric, bi-distributions $w_1$ and $w_0$ in $\mathcal{D}'(M^2)$ such that $\cdot_{w_i}$ is well defined on $\mathcal{F}^V$. 
Suppose that 
\[
w_0= P - N + w_1 
\]
where $P$ and $N$ are also positive symmetric bi-distributions   and that $N$ can be restricted on the diagonal and there it is bounded, namely,
\[
\sup_x\{N(x,x)\}\leq K
\]
for some positive constant $K$.
Then for $V^g=\int \cos(a\varphi(x))g(x)d\mu_x$ as in \eqref{eq:interacting-potential}, assuming that the partial sums defining the right hand side converge, we have the following estimates: 
\[
{\rm{ev}}_{0}(\exp_{w_1}(\lambda V^g))\leq  {\rm{ev}}_{0}(\exp_{v_0}(\lambda V^g)) \leq 
2 {\rm{ev}}_{0}(\exp_{w_0}(2e^{a^2\frac{K}{2}} \lambda V^g)) 
\]
where $v_0$ is $w_0+N$ and where we   used the notation of \eqref{eq:estimate-S(V)}.
\end{thm}
\begin{proof}
 
The distribution $w_1$, $v_0$ and $w_0$ are real valued, hence, the terms of the series defininig 
${\rm{ev}}_{0}(\exp_{w_1}(\lambda V^g))$, 
${\rm{ev}}_{0}(\exp_{v_0}(\lambda V^g))$
and 
$
2 {\rm{ev}}_{0}(\exp_{w_0}(2e^{a^2\frac{K}{2}} \lambda V^g)) 
$ can be directly computed as in \eqref{eq:stima-Re} and they are positive.
Hence, since the partial sums defining
$
2 {\rm{ev}}_{0}(\exp_{w_0}(2e^{a^2\frac{K}{2}} \lambda V^g)) 
$
is convergent by hypothesis all the three series are absolutely convergent.
Since $v_0\geq w_1$,   in the sense that $v_0-w_1$ is a positive distribution, the first inequality  descends from conditioning (Proposition \ref{pr:conditioning} and the discussion given in Remark \ref{rem:how-use-it}).
In particular, to apply Proposition \ref{pr:conditioning} we observe that considering $f_{n,x}$ a sequence of smooth compactly supported function converging to the Dirac delta function centered in $x$ we can write  $V^g$ as a limit of polynomial regular functionals in the following way
\[
V^g = \frac{1}{2} (V_a^g+V_{-a}^g), 
\qquad
V^g_a = \lim_{n\to \infty}\lim_{N\to\infty}
\sum_{j=0}^N\frac{(\mathrm{i}a)^j}{j!} 
\int_M  \Phi(f_{n,x})^j d\mu_x
\]
and by direct inspection
${\rm{ev}}_{0}(V_g\cdot_w\dots \cdot_w V_g)$ is positive.

Furthermore, since $w_0 \leq v_0$  and $v_0(x,x)-w_0(x,x)\leq K$, the second inequality descends from inverse conditioning (Proposition \ref{pr:inverse-conditioning2}).
\end{proof}

\subsection{The massless Hadamard function.}
  Consider the function $\mathcal{H}_0^\mu$  on two-dimensional Minkowski spacetime $M$ which has been introduced in \eqref{eq:form-of-H0-Delta} and which takes the form
\[
\mathcal{H}^\mu_0(x) = -\frac{1}{4\pi}\log\left|\frac{x^2}{ 4\mu^2}\right|
\]
  where as before 
$\mu$ is a positive length scale.
Notice that $\mathcal{H}^\mu_0$ is locally integrable, hence we may associate a symmetric distribution (two-point function)  
namely an element of $\mathcal{D}'(M^2)$ in the following way
\begin{equation}\label{eq:distH}
H^\mu_0(f,g) := \int_{M^2} \mathcal{H}^\mu_0(x-y)f(x)g(y) d\mu_xd\mu_y.
\end{equation}
  With a little abuse of notation, in \eqref{eq:form-of-H0-Delta} we have denoted the integral Kernel of $H_0^{\mu}$ with the same symbol.  The distribution
$H^\mu_0$ is symmetric however it is not positive.
We would like to prove that its restriction on a diamond
\begin{equation}\label{eq:domainDmu}
D_\mu:=\{ x=(t,s)\in M |  -\mu <(t+s) <\mu,\;-\mu <(t-s) <\mu   \}
\end{equation}
is positive.
We have the following proposition.
\begin{prop}
The distribution $H^\mu_0$ introduced in  \eqref{eq:distH} is symmetric and positive on $\mathcal{D}(D_\mu\times D_\mu)$.
\end{prop}
\begin{proof}
Since $\mathcal{H}^\mu_0(x)=\mathcal{H}^\mu_0(-x)\in \mathbb{R}$, $H^\mu_0$ is symmetric.
We observe that if $f,g\in\mathcal{D}(D_\mu)$ we have
\[
H^\mu_0(f,g) =  \int_{M^2} \tilde{\mathcal{H}^\mu_0}(u-u',v-v')f(u,v)g(u',v') dudvdu'dv'
\]
where $v=t+s$ and $u=t-s$ are null coordinates and 
\[
\tilde{\mathcal{H}}(u,v) = L(u)+L(v) :=
-\frac{1}{16\pi}\left(\log\left|\frac{u}{2\mu}\right|\chi(u)  +\log\left|\frac{v}{2\mu}\right| \chi(v)\right)
\]
where $\chi$ is the characteristic function of the interval $[-2\mu,2\mu]$. 
Hence $\tilde{\mathcal{H}}$ coincides with $\mathcal{H}$ on $D_\mu \times D_\mu$.
$L$ is a function from $\mathbb{R}$ to $\mathbb{R}$, its Fourier transform is
\[
\hat{L}(p) = c  \frac{\text{Si}(2\mu p)}{p} 
\]
where $\text{Si}(z)$ is the sine integral function and $c$ an appropriate constant.
Hence we have that 
\begin{align*}
H^\mu_0(f,g) &= c' \int_{\mathbb{R}}  \hat{L}(p_u) \overline{\hat{f}}(p_u,0)\hat{g}(p_u,0) dp_u
+
 c' \int_{\mathbb{R}}  \hat{L}(p_v) \overline{\hat{f}}(0,p_v)\hat{g}(0,p_v) dp_v\\
 &= c' \int_{\mathbb{R}}  \hat{L}(p) \left( \overline{\hat{f}}(p,0)\hat{g}(p,0) +  \overline{\hat{f}}(0,p)\hat{g}(0,p) \right) dp
\end{align*}
where $\hat{f}$ and $\hat{g}$ are the Fourier transform of $f$ and $g$ in $M$. 
Since $\text{Si}(z)/z$ is positive, $\hat{L}$ is positive and we have proved that $H^\mu_0$ on $\mathcal{D}(D_\mu\times D_\mu)$ is positive.
\end{proof}

\subsection{The massive Hadamard function.}
Consider now the symmetric part of the massive two-point function $\Delta_m^+$ given in \eqref{eq:massive2pt}. Its integral kernel is the real part of \eqref{eq:massive2pt}, i.e.
\begin{equation}\label{eq:hm}
\mathcal{H}_m(x)= \frac{1}{2\pi} \text{Re}\left(K_0(m\sqrt{ x^2 })\right).
\end{equation}
The distribution
\begin{equation}\label{eq:Hm}
H_m(f,g) = \int_{M^2} \mathcal{H}_m(x-y) f(x)g(y) d\mu_xd\mu_y
\end{equation}
is the symmetric or real part of the two-point function of the massive vacuum and it is also equal to the real part of the Feynman propagator.

\begin{prop}
The distribution $H_m$ is symmetric and positive.
\end{prop}
\begin{proof}
By definition $H_m$ is symmetric. As is well-known, positivity follows from the fact that its Fourier transform 
is the Schwartz distribution
\[
{\hat{\mathcal{H}}}_m(p)=\frac{1}{2}\delta(p^2+m^2)
\]
which is positive.
\end{proof}

\subsection{Comparison}\label{se:comparison}

To apply Theorem \ref{th:control} we have to analyze $H_0^\mu-H_m$ on $\mathcal{D}(D_\mu^2)$,
 extract its positive part $P$ and its negative part $N$ and subsequently, find a bound for $N(x,x)$.
To do this in practice, we consider $H_0^\mu-H_m$ on $\mathcal{D}(M^2)$ and restrict the domain with a smooth compactly supported positive symmetric function $\Omega\in C^\infty_0(M)$ which is equal to $1$ on {\color{ blue} $D_{2\mu}\times D_{2\mu}$}
when composed with the map from $\imath:M^2\to M$, 
$\imath(x,y)=x-y$, 
namely we use as integral kernel $(\mathcal{H}_0^\mu-\mathcal{H}_m)\Omega$ at the place of $(\mathcal{H}_0^\mu-\mathcal{H}_m)$.
Let us denote this distribution by  
\[
W(f\otimes g) = \int_{M^2} (\mathcal{H}_0^\mu(x-y)-\mathcal{H}_m(x-y))\Omega(x-y)     f(x) g(y) d\mu_x d\mu_y.
\]
Clearly for $f, g\in \mathcal{D}(D_\mu)$, $W$ coincides with $H_0^\mu-H_m$. 
Furthermore, since $\mathcal{H}_0^\mu-\mathcal{H}_m$ is a locally 
integrable function we have that $(\mathcal{H}_0^\mu-\mathcal{H}_m)\Omega \in L^{1}(M)$. Hence, there is a constant $C$, such that
\[
|W(f\otimes g)| \leq C \|f\|_2\|g\|_2.
\] 
This implies that $W$ can be extended by continuity to a bounded quadratic form $\tilde{W}$ on $L^2(D_\mu)$ and then, by a corollary of the Riesz representation theorem, $\tilde{W}$ is represented by a bounded linear operator on $L^2(D_\mu)$.
This operator, associated to $W$, is multiplicative in the Fourier domain and it corresponds to 
\[
W(f\otimes g) = \int_{\mathbb{R}^2} \hat{\mathcal{W}}(k) \hat{f}(k) \hat{g}(k)  dk
\]
for a suitable function $\hat{\mathcal{W}}$. Actually, 
by the Riemann Lebesgue lemma, $\hat{\mathcal{W}}$ is a continuous function vanishing at infinity. Furthermore, both $\mathcal{H}_0^\mu-\mathcal{H}_m$ and $\Omega$ are symmetric real functions, hence 
$\hat{\mathcal W}$ is real. 
To decompose $W$ in its positive and negative part we just need to decompose the real continuous function  $\hat{\mathcal{W}}$ in its positive and negative part
\[
\hat{\mathcal{W}} = \hat{\mathcal{P}} -\hat{\mathcal{N}}
\]
where now both $\hat{\mathcal{P}}$ and $\hat{\mathcal{N}}$ are positive. 
Both $N$ and $P$ are convolution operators
\begin{align*}
N(f,g) &:= \int_{M^2} \mathcal{N}(x-y) f(x)g(y) d\mu_xd\mu_y \\
P(f,g) &:= \int_{M^2} \mathcal{P}(x-y) f(x)g(y) d\mu_xd\mu_y
\end{align*}
and 
\[
N(x,x) = \mathcal{N}(0) = \int_{\mathbb{R}^2} \hat{\mathcal{N}}(k) d^2k.
\]
  Furthermore, since $N$ is a positive distribution, 
$\hat{\mathcal N}\geq 0$, and thus  $\mathcal{N}(0)= \int_{\mathbb{R}^2}  \hat{\mathcal{N}}(k) d^2k
=\|\hat{\mathcal{N}}\|_1$.
The  following proposition shows that $\hat{\mathcal W}$ is integrable and hence, so is its negative part $\hat{\mathcal N}$.

\bigskip
\begin{prop}\label{prop:comparison}
The Fourier transform $\hat{\mathcal{W}}$ of 
\[
\mathcal{W}=(\mathcal{H}_0^\mu-\mathcal{H}_m)\Omega,
\]
is in $L^1(\mathbb{R}^2)$. The intersection of the supports of its positive and negative part in the Fourier domain is a null set.
It follows that also $\hat{\mathcal N}$ is integrable, so there is a constant $K$ such that 
\[
|N(x,x)| = \| \hat{\mathcal{N}} \|_1 \leq   \frac{K}{\hbar}.
\]
The constant $K$ can be chosen to be equal to 
$\hbar \| \hat{\mathcal{W}} \|_1$.
\end{prop}
\begin{proof}
Notice that 
\begin{equation}
\mathcal{H}_0^\mu(x) - \mathcal{H}_m(x) = -\frac{1}{4\pi}   \log\left( {\frac{|x^2|}{4 \mu^2}}\right)(1-I_0(m\sqrt{x^2})) + r(x)
\end{equation}
where $I_0$ is the modified Bessel function of the first kind. Furthermore, $I_0(m\sqrt{x^2})$ and $r(x)$ are smooth functions on $M$. 
We also have that 
\[
-\frac{1}{4\pi}  \log \left({\frac{|x^2|}{4 \mu^2}}\right)(1-I_0(m\sqrt{x^2})) = s_0(x)
+ \log\left({\frac{|x^2|}{4 \mu^2}}\right)r'(x) 
\]
where $r'(x)$ is a smooth function and it is  of order $O((x^2)^2)$ near $x^2=0$ and where 
\begin{align}
s_0(x) &= - \frac{m^2}{16 \pi}\log\left(\frac{|x^2|}{4 \mu^2}\right)x^2.
\end{align}
Notice that $r(x)\Omega(x)$ is a compactly supported smooth function. Its Fourier transform is a Schwartz function, hence it is integrable. 
We now observe that, the Fourier transform of $l_1(x) = r'(x)\log|x^2|\Omega(x)$ is also in $L^1$. Actually,  
using null coordinates $u=t-s$ and $v=t+s$ to parametrize $x=(t,s)$, we have that 
\[
l_1(u,v) = \log|v|r'(u,v)\Omega(u,v)+\log|u|r'(u,v)\Omega(u,v).
\]
Hence,  $l_1$ is an integrable function because it is of compact support, both $\log|v|$ and $\log|u|$ are locally integrable near $v=0$ or $u=0$ and $\Omega$ and $r'$ are smooth compactly supported functions.
Hence, the Fourier transform of $l_1\in L^1(M)$ exists and it is a bounded continuous function, namely by $|\hat{l}_1|\leq \|l_1\|_1$.   We can now improve that bound to get that $\hat{l}_1$ is integrable too, actually
near $v=0$ $\log|v|r'(u,v)\Omega(u,v) = c\log|v|(v^2u^2+ O(v^3))$ where $c$ is a suitable constant and similarly, near $u=0$ $\log|u|r'(u,v)\Omega(u,v)= c\log|u|(v^2u^2+ O(u^3)$ and hence we have that 
\[
\partial_u^2 l_1(u,v), \qquad 
\partial_v^2 l_1(u,v),
\qquad 
\partial_u^2\partial_v^2 l_1(u,v),
\]
are integrable functions because they are locally integrable near $v=0$ or $u=0$, smooth if $u\neq 0$ and $v\neq 0$ and of compact support.
This implies that 
\[
|\hat{l}_1(k_v,k_u)(1+k_v^2)(1+k_u^2)|\leq C =\|l_1\|_1+ \|\partial_v^2 l_1\|_1+\|\partial_u^2 l_1\|_1+\|\partial_v^2\partial_u^2 l_1\|_1
\]
and hence $\hat{l}_1$ is also integrable.

To conclude the proof, we need to be sure that also $l_0=s_0\Omega$ is such that $\hat{l}_0$ is an element of $L^1$. Using null coordinates we have that 
\[
s_0 = - \frac{m^2}{16 \pi}\log\left(\frac{|v|}{2\mu}\right)vu 
- \frac{m^2}{16 \pi}
\log\left(\frac{|u|}{2\mu}\right)vu.
\]
Hence,
\begin{equation}\label{eq:step-l0}
    l_0(u,v)=s_0(u,v)\Omega(u,v) =  f(v) \frac{u\Omega(u,v)}{16\pi} +
f(u) \frac{v\Omega(u,v)}{16\pi}
-\frac{m^2}{8\pi} \log\left(\frac{b}{2\mu}\right)uv\Omega(u,v)
\end{equation}
where $f(v):=-m^2 v\log|v/b|$ for $v$ on the interval $[-b,b]$ and $0$ outside, with $b$ chosen sufficiently large to have that $\chi_{[-b,b]}(v)\Omega(u,v) = \Omega(u,v)$ and $\chi_{[-b,b]}(u)\Omega(u,v) = \Omega(u,v)$, where $\chi_{[-b,b]}$ is the characteristic function of the interval $[-b,b]$.

To prove that $\hat{l}_0$ exists as an element of $L^1$ we analyze the Fourier transform of the three contributions $l_{01}$, $l_{02}$ and $l_{03}$ at the right hand side of \eqref{eq:step-l0} separately. 
$l_{03}(u,v):= -\frac{m^2}{8\pi}\log\left(\frac{b}{2\mu}\right) uv\Omega(u,v)$ is smooth and compactly supported, hence, its Fourier transform is a Schwartz function and thus also integrable.

Let us study the Fourier transform of
\[
l_{01}(u,v):=f(v)g(u,v) := f(v)\frac{u\Omega(u,v)}{16\pi}
\]
$g$ is smooth and compactly supported hence, its Fourier transform $\hat{g}(k_u,k_v)$ is a Schwartz function.  Furthermore, by the convolution theorem in one dimension 
\[
\hat{l}_{01}(k_u,k_v) = \frac{1}{\sqrt{2\pi}} \int_{-\infty}^\infty dp\;    
\hat{g}(k_u,p) \hat{f}(k_v-p).
\]
Hence in order to prove that $\hat{l}_{01}$ is integrable it is sufficient to prove that the Fourier transform of $f$ is in $L^1$. 
By direct inspection, we have that
\[ 
\hat{f}(k) =  \frac{m^2}{\sqrt{2\pi}} \int_{-b}^b  v\log\left|\frac{v}{b}\right|      e^{\mathrm{i} k v} dv  =
i\sqrt{\frac{2}{\pi}}\frac{(\sin(b k)-\text{Si}(b k))}{k^2}
\]
it vanishes in $k=0$ and it decays as $1/k^2$ for large $k$ and it is thus an integrable function.
Similarly we can prove that  $\hat{l}_{02}$, the Fourier transform of $l_{02} := f(u)\frac{v\Omega(u,v)}{16\pi}$, is an in $L^1$. 
Hence, $\hat{l}_0=\hat{l}_{01}+\hat{l}_{02}+\hat{l}_{03}$ is a sum of three integrable functions hence it is integrable.



This discussion implies that $\hat{\mathcal W}$ is in $L^1$ and hence, so is its negative part $\hat{\mathcal N}$, furthermore, since $\supp \hat{\mathcal N}$ intersect $\supp \hat{\mathcal P}$ in a null set we have that $\|\hat{\mathcal{N}}\|_1 \leq \|\hat{\mathcal{N}}\|_1+\|\hat{\mathcal{P}}\|_1=\| \hat{\mathcal{W}}\|_1$, so  $K=\hbar \|\hat{\mathcal{W}}\|_1$ is a finite number which bounds $N(x,x)$.
\end{proof}

\subsection{Convergence of the S-matrix on a sufficiently small local region}
In this section we prove the analogue of Proposition~6 in \cite{BR16} where the convergence of the $S$-matrix of the massive Lorentzian Sine-Gordon theory was proved. Beyond what was proved there, we gain a clearer understanding of the estimate's dependence of the support of the cutoff function (adiabatic switching function) $g$ that cuts off the interaction $V^g$.
We shall work in $\mathcal{A}^V_{\hbar\Delta^+_m}=(\mathcal{F}^V,\star_{\hbar\Delta^+_m})$, hence the smeared interaction Lagrangian $V^g$ seen as an element of $\mathcal{A}^V_{\hbar\Delta^+_m}$ corresponds to the interaction Lagrangian normal ordered with respect to the massive vacuum.   The state we shall consider is the Minkowski vacuum and it is represented by the positive linear functional ${\rm ev}_0$ on $\mathcal{A}^V_{\hbar\Delta^+_m}$.
To simplify the notation, unless strictly necessary, from now on, we shall omit the symbol $\star$ and $\star_w$ to denote the product in the used algebra.

\begin{thm}\label{thm:convergence}
Let $g\geq0$ be a cutoff function whose support is contained in $D_\mu$ given in \eqref{eq:domainDmu} and which is equal to $1$ on $D_{\mu'}$ for some $\mu'<\mu$.
Let $S_n^m(V^g)$ be the $n-$th perturbative order of the $S$ matrix of $V^g=\int\cos(a\varphi(x))g(x)d\mu_x$ constructed with respect to the linear theory with mass $m$. Where the constant $a$ is positive and $\alpha=a^2\hbar/(4\pi)<1$.
Then  for any real number $p$ in $[1,\alpha^{-1})$, there are constants $K$ and $C$, such that
\[
\left| S_n^m(V^g)(\varphi) \right|  \leq 
\frac{2(2\mu)^{n\alpha}}{(n!)^{1-1/p}}\left(\frac{2\lambda e^{\frac{a^2}{2} K}}{\hbar} \right)^n  
\|g\|_q^{n}
\left( C^{n} \right)^{1/p}
\]
with $1/q + 1/p=1$.
Choosing $p>1$, it follows that  the sum $S(V^g)=\sum_{n\geq 0}S_n^m(V^g)$ 
is absolutely convergent for every $\varphi\in\mathcal{C}$.

Observe that the constant $K$ is the one obtained in Proposition \ref{prop:comparison} and in fact does not depend on $p$.
\end{thm} 
\begin{proof}
In \eqref{eq:S-matrix} we have seen that $S(V^g)$ can be expressed as a power series whose coefficients are elements of $(\mathcal{F},\star_{\Delta^+_m})$.
We want to prove that, for every $\varphi\in\mathcal{C}$, the sum is absolutely convergent.  
We have seen in \eqref{eq:estimate-S(V)} that if the series defining ${\rm{ev}}_0\left( \exp_{\hbar H_m}\left(\frac{\lambda}{\hbar}V^g\right)\right)$
converges, the series defining $S(V^g)$   is absolutely convergent and it can be further bounded by
\[
|S(V^g)| \leq  \sum_n |S^m_n(V^g)|  \leq {\rm{ev}}_0\left( \exp_{\hbar H_m}\left(\frac{\lambda}{\hbar}V^g\right)\right)
\]
where $\exp_{\hbar H_m}$ is the time ordered exponential constructed with the time ordering map $\mathcal{T}^{\hbar H_m}$.
We may now use Theorem \ref{th:control} to control the convergence of 
${\rm{ev}}_0\left( \exp_{\hbar H_m}\left(\frac{\lambda}{\hbar}V^g\right)\right)$ by means of 
${\rm{ev}}_0\left( \exp_{\hbar{H_0^\mu}}\left(\frac{\lambda}{\hbar}V^g\right)\right)$ provided the series defining the latter converges.
  Furthermore, notice that $H_0^\mu$ is positive on $D_\mu\times D_\mu$, and since the support of $g$ is in $D_\mu$ the fact that $H^\mu_0$ is not positive in $M$ does not prevent us to apply Theorem  \ref{th:control}.
To be more precise, we can multiplied $H^0_\mu$ and $H_m$ with the cutoff function $\Omega\circ\imath$ introduced in section \ref{se:comparison} without altering finial result.

Actually, thanks to the discussion presented in section \ref{se:comparison} and thanks to the results of Proposition \ref{prop:comparison},  we may extract the positive and negative part of 
\[
(H_{0}^\mu-H_m)\Omega\circ\imath = P - N
\]
furthermore, the integral kernel of the negative part $N$ is such that $N(x,x)\leq \frac{K}{\hbar} $, hence, we 
may now use Theorem \ref{th:control} to obtain
\begin{equation}\label{eq:step-local}
{\rm{ev}}_0\left( \exp_{\hbar H_m}\left(\frac{\lambda}{\hbar}V^g\right)\right) \leq 
2 {\rm{ev}}_0\left( \exp_{\hbar{H_0^\mu}}\left(\frac{2 \lambda e^{\frac{a^2}{2}K}}{\hbar}V^g\right)\right). 
\end{equation}
where $K$ is the constant given in Proposition \ref{prop:comparison}.
To prove that the right hand side of the inequality \eqref{eq:step-local} is the result of an absolutely convergent power series in $\lambda$ we proceed as in the proof of Proposition 6 in \cite{BR16}. 
  Let $w$ be any positive distribution of Hadamard type,  then he functional ${\rm{ev}}_0(F):=F(0)$ is a linear normalized positive functional on the commutative $*$-algebra 
$(\mathcal{F}^V,\cdot_w)$
generated by $V^g_a$ and ${V^g_a}^*$ with the $\cdot_w$ product, see Remark \ref{rem:how-use-it} for further details, and hence  the  Cauchy-Schwarz inequality holds
\[
|{\rm{ev}}_0(A \cdot_w B)| \leq  \sqrt{|{\rm{ev}}_0(A^* \cdot_w A)|}\sqrt{|{\rm{ev}}_0(B^* \cdot_w B)|}.
\]
Choosing $B=1$, we obtain
\begin{equation}\label{eq:C*inequality-1}
|{\rm{ev}}_0(A)|^2 \leq |{\rm{ev}}_0(A^* \cdot_w A)|.
\end{equation}
  With this observation used with $w=\hbar H_0^\mu$, to estimate the convergence of the sum present in the last term on the right hand side of \eqref{eq:step-local}
we notice that  
\begin{align*}
{\rm{ev}}_0(\underbrace{V^g \cdot_{\hbar{H_0^\mu}}   \dots \cdot_{{\hbar{H_0^\mu}}} V^g}_{n})  
&=\frac{1}{2^n}\sum_{k=0}^n  \begin{pmatrix} n \\ k \end{pmatrix} 
{\rm{ev}}_0\left(
\underbrace{V^g_{a} \cdot_{{\hbar{H_0^\mu}}}   \dots \cdot_{{\hbar{H_0^\mu}}} V^g_{a}}_{k}\cdot_{{\hbar{H_0^\mu}}}
\underbrace{V^g_{-a} \cdot_{{\hbar{H_0^\mu}}}   \dots \cdot_{{\hbar{H_0^\mu}}} V^g_{-a}}_{n-k}
\right)
\\
&\leq \frac{1}{2^n}\sum_{k=0}^n  \begin{pmatrix} n \\ k \end{pmatrix} 
\left({\rm{ev}}_0\left(
\underbrace{V^g_{a} \cdot_{{\hbar{H_0^\mu}}}   \dots \cdot_{{\hbar{H_0^\mu}}} V^g_{a}}_{n}\cdot_{{\hbar{H_0^\mu}}}
\underbrace{V^g_{-a} \cdot_{{\hbar{H_0^\mu}}}   \dots \cdot_{{\hbar{H_0^\mu}}} V^g_{-a}}_{n}
\right)
\right)^{\frac{1}{2}}
\\
&=  
\left({\rm{ev}}_0\left(
\underbrace{V^g_{a} \cdot_{{\hbar{H_0^\mu}}}   \dots \cdot_{{\hbar{H_0^\mu}}} V^g_{a}}_{n}\cdot_{{\hbar{H_0^\mu}}}
\underbrace{V^g_{-a}\cdot_{{\hbar{H_0^\mu}}}   \dots \cdot_{{\hbar{H_0^\mu}}} V^g_{-a}}_{n}
\right)
\right)^{\frac{1}{2}}
\end{align*}
where we used \eqref{eq:C*inequality-1} in the first inequality and that ${V^g}_a^*=V^g_{-a}$.
Now, since $g\in \mathcal{D}(D_\mu)$, 
by Lemma \ref{le:Cauchy-Determinant} we have that
\[
{\rm{ev}}_0(\underbrace{V^g \cdot_{\hbar{H_0^\mu}}   \dots \cdot_{{\hbar{H_0^\mu}}} V^g}_{n})  
\leq
(2\mu)^{n\alpha}
\|g\|_q^{n}  \left( C^{n}  n!\right)^{1/p}.
\]
Together with the discussion presented above and in particular \eqref{eq:step-local}, this implies that
\begin{equation}\label{eq:estimate-Cauchy}
|S(V^g)| \leq  \sum_n |S^m_n(V^g)|
\leq \sum_n   
\frac{2(2\mu)^{n\alpha}}{n!}\left(\frac{2\lambda e^{\frac{a^2}{2}K}}{\hbar} \right)^n  
\|g\|_q^{n}
\left( C^{n}  n!\right)^{1/p}
\end{equation}
where $K$ is given in \ref{prop:comparison},
and thus the sum defining $S(V^g)$ for a generic $\varphi$ converges absolutely because $\|g\|_q$ is finite and the right hand side of \eqref{eq:estimate-Cauchy} does note depend on $\varphi$. 
Hence, 
considering a suitable number of derivatives in $\lambda$ and evaluating at $\lambda = 0$,
the last inequality in \eqref{eq:estimate-Cauchy} improves \eqref{eq:step-local-g} to give
\begin{align*}
\left| S_n^{m}(V^g)(\varphi) \right| 
&\leq\frac{2(2\mu)^{n\alpha}}{n!}\left(\frac{2\lambda e^{\frac{a^2}{2}K}}{\hbar} \right)^n  
\|g\|_q^{n}  \left( C^{n}  n!\right)^{1/p}
\end{align*}
which concludes the proof.
 
\end{proof}


In the proof of Theorem \ref{thm:convergence} given above we used a known result originally obtained by Cauchy, see also \cite{Froe76, BR16,BFR17},  to estimate the expectation values in ${\rm{ev}}_0$ of the $n-th$ power of $V^g$. We collect this result in the following Lemma where we highlighted the dependence on $\mu$ in the obtained estimates 
and where we adapted the result known in the Euclidean case to the more extended singularities present in $\log|\frac{(x^0)^2-(x^1)^2}{4\mu^2}|$
with respect to 
$\log|\frac{(x^0)^2+(x^1)^2}{4\mu^2}|$.

\begin{lemma}(Cauchy determinant)\label{le:Cauchy-Determinant}
Let $\mu$ be a fixed number, and let $g_0\in C^\infty_0(D_\mu)$. 
Consider the following expression
\begin{gather*}
O={\rm{ev}}_0\left(
\underbrace{V^{g_0}_{a} \cdot_{\hbar{H_0^\mu}}   \dots \cdot_{\hbar{H_0^\mu}} V^{g_0}_{a}}_{n}
\cdot_{\hbar{H_0^\mu}}
\underbrace{V^{g_0}_{-a} \cdot_{\hbar{H_0^\mu}}   \dots \cdot_{\hbar{H_0^\mu}} V^{g_0}_{-a}}_{n}
\right)
=
\\
=\int  e^{-\sum_{1\leq i<j\leq 2n} a_ia_j\hbar H_0^\mu(x_i,x_j)}g_0(x_1)\dots g_0(x_{2n}) 
d\mu_{x_1}\dots d\mu_{x_{2n}}\\
=(4\mu^2)^{n\alpha}\int    
\left( 
\frac{ \prod_{1\leq i<j\leq n}   |(x_i-x_j)^2| |(y_i-y_j)^2|}{\prod_{i=1}^n \prod_{j=1}^n  |(x_i-y_j)^2|} 
\right)^{\alpha} 
g_0(x_1)\dots g_0(x_n)g_0(y_1)\dots g_0(y_n) 
d\mu_{x_1}\dots d\mu_{x_n} d\mu_{y_1}\dots d\mu_{y_n} 
\end{gather*}
where $a_i=a$ if $i\in \{1,\dots, n\}$ and $a_i=-a$
if $i\in\{n+1,\dots, 2n\}$
and where
$\alpha=a^2\hbar/(4\pi)$ is assumed to be smaller than $1$.
For any $p \in [1,\alpha^{-1})$ there exists a constant $C$ such that 
\[
|O|\leq (4\mu^2)^{n\alpha} \|g_0\|_q^{2n} \left( C^{2n}  (n!)^2\right)^{1/p}, 
\]
where $q = p/(p - 1)$.
\end{lemma}
\begin{proof}
Consider an element $x=(x^0,x^1) \in M$, in terms of the ordinary null coordinates ($x^v = x^0+x^1$ and $x^u=x^0-x^2$ ) we have that the Lorentzian squares factorizes $|x^2| = |x^v||x^u|$.
Notice that by a lemma of Cauchy which holds also when the squares are the Lorentzian squares 
\[
\frac{ \prod_{1\leq i<j\leq n}   |(x_i-x_j)^2| |(y_i-y_j)^2|}{\prod_{i=1}^n \prod_{j=1}^n  |(x_i-y_j)^2|} = |\det \mathcal{D}^v||\det \mathcal{D}^u|
\]
where $\mathcal{D}^v$ is an $n\times n$ matrix whose entries are $\mathcal{D}^v_{ij}=1/(x^v_{i}-y^v_{j})$ and
$\mathcal{D}^u$ is an $n\times n$ matrix whose entries are $\mathcal{D}^u_{ij}=1/(x^u_{i}-y^u_{j})$. Hence we have 
\begin{gather*}
O
=
(4\mu^2)^{n\alpha} \int  |\det \mathcal{D}^v|^\alpha|\det \mathcal{D}^u|^\alpha G(X,Y) dXdY 
\end{gather*}
where $X=(x_1,\dots, x_n)$, $Y=(y_1,\dots, y_n)$,  
$dX=d\mu_{x_1}\dots d\mu_{x_n}$, $dY =  d\mu_{y_1}\dots d\mu_{y_n}$ and 
$G(X,Y) = \prod_i g_0(x_i)g_0(y_i)$. 
We observe that 
\[
\det \mathcal{D}^b = \sum_{\pi} \prod_{i=1}^n \frac{(-1)^{|\pi|}}{(x^b_i-y^b_{\pi(i)})}, \qquad b\in\{u,v\}
\]
where the sum is taken over all possible permutations of $(1,\dots, n)$.

Following Theorem 3.4 in \cite{Froe76} we have that 
\[
|O|\leq (4\mu^2)^{n\alpha} \|G\|_q \||\det \mathcal{D}^v|^\alpha |\det \mathcal{D}^u|^\alpha \|_p 
\]
where the norms are computed over ${D}_\mu^{2n}$  and $1/p+1/q=1$ hence both $p,q\geq 1$.

Since we are in the regime where $\alpha = \frac{a^2}{4\pi}\hbar < 1$, for any $p$ such that 
\[
1\leq p < \frac{1}{\alpha}
\]
we have that $\alpha p < 1 $. This implies that  
$\prod_{i,j} 1/|x^b_i-y^b_j|^{\alpha p}$ for $b\in \{u,v\}$ are locally integrable and in that case 
$\||\det \mathcal{D}^v|^\alpha|\det \mathcal{D}^u|^\alpha \|_p$ is finite because the norms are computed on the compact set ${D}_\mu^{2n}$.
More precisely
\[
|O|\leq  (4\mu^2)^{n\alpha} \|G\|_q  \left( \sum_{\pi}  \sum_{\pi'}  \int dXdY \prod_{i=1}^n \frac{1}{|x^v_i-y^v_{\pi(i)}|^{\alpha p}}
\prod_{j=1}^n \frac{1}{|x^u_j-y^u_{\pi'(j)}|^{\alpha p}}   \right)^{1/p} 
\]
where we used the fact that $(\sum_i|a_i|)^{\alpha p} \leq \sum_{i}|a_i|^{\alpha p}$ when $\alpha p<1$ and
where both sums are taken over all possible permutations of $(1,\dots, n)$. There exists a suitable constant $C$ such that for every $\pi$ 
\[
\int dXdY \prod_{i=1}^n   \frac{1}{|x^v_i-y^v_{\pi(i)}|^{\alpha p}}\prod_{j=1}^n \frac{1}{|x^u_j-y^u_{\pi'(j)}|^{\alpha p}}  \leq  C^{2n}.
\]
Since
\[
\|G\|_q = \|g_0\|_q^{2n}
\]
and since the number of possible partitions $\pi$ is $n!$  we thus conclude the proof.
\end{proof}
We observe that the absolute convergence of the series defining $S(V)(\varphi)$ for every $\varphi\in\mathcal{C}$ can be obtained even if $g$ is in $L^q$ for every $q$, actually, Lemma \ref{le:Cauchy-Determinant} can be proved also for those $g$.
We shall use this extension for $g$ which are characteristic functions of finite regions of Minkowski space.

\subsection{Convergence of the S-matrix on every local region}

In Theorem \ref{thm:convergence} we established the convergence of the $S-$matrix when $V^g$ is the interaction Lagrangian smeared with $g$ with a support contained in $D_\mu$ for a fixed $\mu$. In this section we want to prove absolute convergence when $g$ has values in $[0,1]$ and has generic compact support independent on $\mu$. 
To this end recall the form of the series defining $S(V)$ in \eqref{eq:S-matrix} and we start again from the estimate \eqref{eq:estimate-S(V)} which we recall here, 
\[
|S(V^g)| \leq {\rm{ev}}_0\left( \exp_{\hbar H_m}\left(\frac{\lambda}{\hbar}V^g\right)\right).
\]
We can improve the bound using conditioning and inverse conditioning.
In particular, consider the following set of squares in the Minkowski space $M$ constructed with respect to a standard coordinate system
\[
\mathcal{Q} := \{[l,l+1]\times [n,n+1] \subset M \mid l,n\in\mathbb{Z}\}
\]
we denote by $\mathcal{I}_g$ the subset of $\mathcal{Q}$ formed by squares which intersect the support of $g$, namely
\begin{equation}\label{eq:Ig}
\mathcal{I}_g:=\{Q\in\mathcal{Q} \mid Q\cap \supp g \neq \emptyset\},
\end{equation}
Let $\mathcal{Z}_g=\cup_{Q\in\mathcal{I}_g}Q$ be the union of the set of squares in $\mathcal{I}_g$, we denote by $\mathcal{V}_g$ the volume of $\mathcal{Z}_g$ and by $\tilde{g}$ the characteristic function of $\mathcal{Z}_g$.
Hence the cardinality of $\mathcal{I}_g$ coincides with the volume $\mathcal{V}_g$ of $\mathcal{Z}_g$.
With this notation, we have the following theorem.

\begin{thm}\label{thm:convergence-V}
Consider $\mu>1$ a fixed parameter. Let $g$ be a compactly supported smooth function with values in $[0,1]$ and which is $1$ on a large region of the spacetime. Let $g_0$ be the characteristic function of $B = [0,1]^2 \subset M$. 
Let the $\mathcal{V}_g$ be the volume of the union of elements of $\mathcal{I}_g$ given in \eqref{eq:Ig}. 
Let $S_n^m(V^g)$ be the $n-$th perturbative order of the $S$ matrix of $V^g=\int\cos(a\varphi(x))g(x)d\mu_x$ constructed with respect to the linear theory with mass $m$, where the constant $a$ is positive and $a^2\hbar/(4\pi)<1$.
Then, for every $\varphi\in\mathcal{C}$, and 
choosing the parameters $p,C,K$ as in Theorem \ref{thm:convergence}, it holds that 
\[
\left| S_n^m(V^g)(\varphi) \right|  \leq 
\frac{2(2\mu)^{n\alpha}}{(n!)^{1-1/p}}\left(\frac{2\lambda e^{\frac{a^2}{2}K}}{\hbar} \right)^n  
\|g_0\|_q^{n}
\left( C^{n} \right)^{1/p} \mathcal{V}_g^n.
\]
Hence, for every $\varphi\in\mathcal{C}$ the sum $S(V^g)=\sum_{n\geq 0}S_n^m(V^g)$ is absolutely convergent.
\end{thm}

\begin{proof}
We start considering $\mathcal{I}_g$  the set of squares of edge $1$ which are intersecting the support of $g$ given in \eqref{eq:Ig} and their   union $\mathcal{Z}_g$.
We denote by $\tilde{g}$ the characteristic function of $\mathcal{Z}_g$. We have by construction that $0\leq g\leq \tilde{g}$.
Recalling the steps to obtain \eqref{eq:stima-Re}, and observing that $H_m$ is real valued, we have, by direct inspection, that the inequality  \eqref{eq:estimate-S(V)} can be further bounded
\[
|S(V^g)| \leq {\rm{ev}}_0\left( \exp_{\hbar H_m}\left(\frac{\lambda}{\hbar}V^g\right)\right) \leq 
 {\rm{ev}}_0\left( \exp_{\hbar H_m}\left(\frac{\lambda}{\hbar}V^{\tilde{g}}\right)\right).
\]
We now decompose $V^{\tilde{g}} = \sum_{Q\in\mathcal{I}_g} V^{g_Q}$ where $g_Q$ is the characteristic function of the square $Q$.
 
We shall then use Proposition \ref{prop:multi-holder} to further improve this bound, to match the hypothesis of that proposition,
$V^g$ needs to written with positive functionals. 
Let us start observing that $V^{\tilde{g}}$ can be approximated by a sequence of positive regular functionals in $(\mathcal{F}^V,\cdot_{H_m})$ plus  constant functionals. 
Let $f_{x;n}\geq 0$ with $x\in M$ and $n\in \mathbb{N}$ a sequence of compactly supported smooth functions which converges to the Dirac delta function centered in $x$ for large $n$. 
Observe that 
\[
V^{\tilde{g}}  = \lim_{n\to\infty} 
\int  \cos(a \Phi(f_{x;n}))   \tilde{g}(x) d\mu_x
\]
furthermore, 
\[
\cos(a\Phi(f)) = 2 e^{\frac{a^2}{4}\hbar H_m(f,f)} \cos(\frac{a}{2}\Phi(f))\cdot_{\hbar H_m}
\cos(\frac{a}{2}\Phi(f)) - 
e^{\frac{a^2}{2}\hbar H_m(f,f)}.
\]
Hence
\begin{equation}\label{eq:decomposition-of-V}
V^{\tilde{g}}  = \lim_{n\to\infty}
P^{\tilde{g}}_{n} + c^{\tilde{g}}_n
\end{equation}
where $P^{\tilde{g}}_{n}$ is a positive regular functional and $c^{\tilde{g}}_n$ is a constant functional
 hence
\[
 {\rm{ev}}_0\left( \exp_{\hbar H_m}\left(\frac{\lambda}{\hbar}V^{\tilde{g}}\right)\right)
  =
  \lim_{n\to\infty}
  {\rm{ev}}_0
  \left( \exp_{\hbar H_m}\left(
   \frac{\lambda}{\hbar}P^{\tilde{g}}_n + \frac{\lambda}{\hbar}c^{\tilde{g}}_n
  \right)\right)
  = 
  \lim_{n\to\infty}
  e^{\frac{\lambda}{\hbar}c^{\tilde{g}}_n}
  {\rm{ev}}_0
  \left( \exp_{\hbar H_m}\left(
   \frac{\lambda}{\hbar}P^{\tilde{g}}_n 
  \right)\right)
\]
where in the last equality we used the fact that $c_n^{\tilde{g}}$ is a constant functional.

To estimate the 
 $N-$th contribution to the exponential $\left( \exp_{\hbar H_m}\left(
   \frac{\lambda}{\hbar}P^{\tilde{g}}_n 
  \right)\right)$
  we proceed as follows
\[
{\rm{ev}}_0(\underbrace{P_n^{\tilde{g}} \cdot_{\hbar H_m} \dots \cdot_{\hbar H_m} P_n^{\tilde{g}}}_{N}) 
=
{\rm{ev}}_0(\underbrace{\sum_{Q\in\mathcal{I}_g} P_n^{g_Q}  \cdot_{\hbar H_m} \dots \cdot_{\hbar H_m} \sum_{Q\in\mathcal{I}_g} P_n^{g_Q} }_{N}) 
=
\sum_{(Q_1,\dots, Q_N) \in \mathcal{I}_g^N}{\rm{ev}}_0(
\prod_{j=1}^N P_n^{g_{Q_j}}),
\]
where the product in the last term of the chain of equalities is computed with respect of the commutative product $\cdot_{\hbar H_m}$.
We may now apply
Proposition \ref{prop:multi-holder} and more precisely equation \eqref{eq:multi-holder-ev} to obtain

\[
{\rm{ev}}_0(\underbrace{P_n^{\tilde{g}} \cdot_{\hbar H_m} \dots \cdot_{\hbar H_m} P_n^{\tilde{g}}}_{N})
\leq
\sum_{(Q_1,\dots, Q_N) \in \mathcal{I}_g^N}
\prod_{j=1}^N
{\rm{ev}}_0(\underbrace{P_n^{g_{Q_j}}
\cdot_{\hbar H_m} \dots \cdot_{\hbar H_m} P_n^{g_{Q_j}}}_{N})^{\frac{1}{N}}.
\]
By translation invariance, we have that for every $j\in \{1,\dots, N\}$ 
\[
{\rm{ev}}_0(\underbrace{P_n^{g_{Q_j}}
\cdot_{\hbar H_m} \dots \cdot_{\hbar H_m} P_n^{g_{Q_j}}}_{N})
=
{\rm{ev}}_0(\underbrace{P_n^{g_{0}}
\cdot_{\hbar H_m} \dots \cdot_{\hbar H_m} P_n^{g_{0}}}_{N})
\]
hence
\[
{\rm{ev}}_0(\underbrace{P_n^{\tilde{g}} \cdot_{\hbar H_m} \dots \cdot_{\hbar H_m} P_n^{\tilde{g}}}_{N})
\leq
\sum_{(Q_1,\dots, Q_N) \in \mathcal{I}_g^N}
{\rm{ev}}_0(\underbrace{P_n^{g_{0}}
\cdot_{\hbar H_m} \dots \cdot_{\hbar H_m} P_n^{g_{0}}}_{N})
=
{{\mathcal{V}_g}}^N {\rm{ev}}_0(\underbrace{P_n^{g_0} \cdot_{\hbar H_m} \dots \cdot_{\hbar H_m} P_n^{g_0}}_{N})
\]
where the last equality holds because the 
number of squares in $\mathcal{I}_g$ coincides with the volume $\mathcal{V}_g$.
By translation invariance it holds that  
$c^{\tilde{g}}_n  = \mathcal{V}_g c^{g_0}_n$
hence
\[
  {\rm{ev}}_0
  \left( \exp_{\hbar H_m}\left(
   \frac{\lambda}{\hbar}P^{\tilde{g}}_n + \frac{\lambda}{\hbar}c^{\tilde{g}}_n
  \right)\right)
  \leq 
    {\rm{ev}}_0
  \left( \exp_{\hbar H_m}\left(
   \mathcal{V}_g\frac{\lambda}{\hbar}P^{g_0}_n + \mathcal{V}_g\frac{\lambda}{\hbar}c^{g_0}_n
  \right)\right)
\]
finally taking the limit $n$ to infinity on both side of the inequality we obtain 
\[
 {\rm{ev}}_0\left( \exp_{\hbar H_m}\left(\frac{\lambda}{\hbar}V^{\tilde{g}}\right)\right)
 \leq
  {\rm{ev}}_0\left( \exp_{\hbar H_m}\left(\mathcal{V}_g\frac{\lambda}{\hbar}V^{g_0}\right)\right).
 \]
Hence we have that
\begin{align*}
|S(V^g)| &\leq {\rm{ev}}_0\left( \exp_{\hbar H_m}\left(\mathcal{V}_g\frac{\lambda}{\hbar}V^{g_0}\right)\right).
\end{align*}
Applying Theorem \ref{thm:convergence} with $\lambda \mathcal{V}_g$ in the place of $\lambda$ and with $\mu=1$ we have proved the claim.
\end{proof}

\begin{rem}\label{rem:strong-op-top-conv}
With the estimates used to prove Theorem \ref{thm:convergence} and Theorem \ref{thm:convergence-V}, we can repeat an argument presented in \cite{BFR17} (see Proposition 7) to prove that the sum which defines $S(V)=\sum_n S_n$ is actually convergent in the strong operator topology of the GNS representation of $\mathfrak{A}$ the Weyl algebra of the free field in the vacuum state $\omega_m$.
To present this argument in some details, consider the GNS tripe $(\mathfrak{H},\pi,\Omega_0)$ of the free Weyl algebra of massive observable in the vacuum state $\omega_m$. 
Consider $D$ the dense set in  $\mathfrak{H}$ spanned by vectors $\pi(e^{\mathrm{i}\phi(f)})\Omega_0$ with any $f\in C^\infty_0(M;\mathbb{R})$, where $e^{\mathrm{i}\phi(f)}$ are the Weyl generators of the algebra   which are represented by $e^{\mathrm{i} \Phi(f)}$ in $(\mathcal{F}_{\text{reg}},\star)$ and by $e^{-\frac{1}{2} w(f,f)}e^{\mathrm{i} \Phi(f)}$ in $(\mathcal{F}_{\text{reg}},\star_w)$.
The estimates used to prove Theorem \ref{thm:convergence} and Theorem \ref{thm:convergence-V} are sufficient to prove that 
the double sum 
\[
\pi(S( V^g))=\sum_n \frac{1}{n!} \frac{{\mathrm{i}^n}\lambda^n}{2^n  \hbar^n}    \sum_{a_i\in\pm a} \pi(\mathcal{T}(V^g_{a_1} \otimes \dots \otimes  V^g_{a_n}))
\]
 converges in the strong operator norm on the dense domain $D$. Hence $\pi(S(\lambda V))$ is an operator on $D$ and with values in $\mathfrak{H}$. 
  Actually, by direct inspection,  since $f$ is real, we get that  
\[
\| \pi(\mathcal{T}(V^g_{a_1} \otimes \dots \otimes V^g_{a_n})) \pi(e^{\mathrm{i}\phi(f)})\Omega_0\|^2 \leq {\rm{ev}}_0(\underbrace{V^g_{a} \cdot_{\hbar{H_m}}   \dots \cdot_{\hbar{H_m}} V^g_{a}}_{n}\cdot_{{H_m}}
\underbrace{V^g_{-a}\cdot_{\hbar{H_m}}   \dots \cdot_{\hbar{H_m}} V^g_{-a}}_{n}
)
\]
and the right hand side does not depend on $f$.  Hence, estimates similar to those used in 
Theorem \ref{thm:convergence} or in Theorem \ref{thm:convergence-V}, permits to obtain the desired convergence of the sum in the strong operator topology.  
Finally, following Proposition 8 in \cite{BFR17}, we also obtain that as an operator on $\mathfrak{H}$, $\pi(S(V^g))$ preserves the norm, hence, since $S( V)$ is formally unitary, it can be extended to an unitary operator on $\mathfrak{H}$.  
\end{rem}

\section{Thermal state on local region for the generating functional} \label{se:thermal}

\subsection{Causality and the algebraic adiabatic limit}\label{se:strategy}

In this section we shall discuss the construction of equilibrium states for the interacting Sine-Gordon theory for a field with finite mass.
The observables of the interacting theory are represented as 
power series in the coupling constants whose coefficients are
elements of the massive Klein-Gordon theory by means of the Bogoliubov map which is given in terms of the $S-$matrix.
Once a cutoff function $g$ with compact support is used to smear the interaction Lagrangian,
\[
V^g = \int    g \cos (a \varphi)  d\mu, 
\]
and the algebra of interacting fields $\mathcal{A}^V_I(M)$ is  constructed
 as the smaller subalgebra of $\mathcal{A}^V[[\lambda]]$
 which contains $\{R_{V^g}(F) \mid F\in \mathcal{F}^V\cap \mathcal{F}_{\text{loc}}(M)\}$. 
As discussed in section \ref{se:BogoliubovMap} the Bogoliubov map $R_{V^{g}}(F)$ is proportional to the derivative of  $S(V^g)^{-1}S(V^g+b F)$ with respect to $b$ evaluated at vanishing $b$. 
Hence,
in the two-dimensional Sine-Gordon case, the power series in $\lambda$ describing $R_{V^g}(F)(\varphi)$ evaluated in a generic field configuration $\varphi$ is actually an absolutely convergent series for any finite value of $\lambda$, because, as shown in Theorem \ref{thm:convergence-V},
this holds for $S(V^g)$.

Any quantum state $\omega$ of $\mathcal{A}^V(M)$ is thus a state for the algebra of interacting fields $\mathcal{A}^V_I(M)$. Hence, the expectation values of the interacting local fields $F\in \mathcal{A}^V_I(M)$ is then obtained as
\[
\omega^I(F) = \omega(R_{V^{g}}(F)).
\]
To obtain $\omega^I$ we thus construct $\omega$ on $\mathcal{A}^V(M)$. Its pullback to $\mathcal{A}_I^V(M)$ is then strightforward. 
  Furthermore, in many cases, like for ${\rm{ev}}_0$, the pointwise convergence of the power series describing $R_{V^g}(F)$ on a generic field configuration $\varphi$, implies the convergence of the power series of expectation values $\omega^I(F)$. 
We shall show that this is the case also for the thermal state we are aiming to construct also in the limit $g\to1$, even if in that case $\omega$ will be described by a power series in $\lambda$.
In the case of states represented by ${\rm{ev}}_0$ in 
$(\mathcal{F}^V,\star_{w_m^\infty})$, as discussed in \cite{BFR17} and in Remark \ref{rem:strong-op-top-conv}, the pointwise convergence implies convergence in the strong operator topology on the Hilbert space of the GNS representation of the studied state.

In view of the time slice axiom \cite{ChF08}, the state $\omega$, is completely fixed by its restriction on  $\mathcal{A}^V(\Sigma_\epsilon)$ (the smaller subalgebra of $\mathcal{A}^V$ which contains the local fields whose support is contained in $\Sigma_\epsilon$) where $\Sigma_\epsilon$ is a time slab namely a small neighbourhood of the Cauchy surface with vanishing Minkowski time, 
\[
\Sigma_\epsilon=\{(x^0,x^1)\in M \mid |x^0|<\epsilon\}. 
\] 
Furthermore, in view of causal properties of the $S$ matrix, for observables supported on $\Sigma_\epsilon$ we may modify the form of the interaction Lagrangian in the complement of the past of $\Sigma_\epsilon$ without altering the action of $R_{V^g}$, namely
\[
R_{V^{g}}(F) = R_{V^{g'}}(F)
\]
if $supp(g-g')\subset (M\setminus J^-(\Sigma_\epsilon))$.
Similarly, if we modify $g$ in the complement of  $J^+(\Sigma_\epsilon)$, we obtain an equivalent representation of $\mathcal{A}_I^V$. This equivalence is described by the adjoint action of a suitable unitary operator, namely
\[
R_{V^{g}}(F) = \mathcal{U} R_{V^{g'}}(F)\mathcal{U}^{-1}
\]
if $supp(g-g')\subset (M\setminus J^+(\Sigma_\epsilon))$. Here $\mathcal{U}= S(V^g)^{-1}S(V^{g'})$ and it is formally a unitary element of the algebra which does not depend on $F$.
Hence, in order to construct an equilibrium state for the interacting theory, we restrict our attention to observables supported on $\Sigma_\epsilon$ and we choose the cutoff function of the following form
\begin{equation}\label{eq:cut-off}
g(x^0,x^1) = \chi(x^0)h(x^1)
\end{equation}
where $\chi$ is a smooth compactly supported function on $\mathbb{R}$ with values in $[0,1]$ which is equal to $1$ on $[-\epsilon,\epsilon]$ and whose support is contained on $[-2\epsilon,2\epsilon]$. Furthermore, $h$ is a space cutoff, namely it is a smooth compactly supported function on $\mathbb{R}$ which is $1$ on large regions of the space.

We shall thus construct a state with fixed $h$ and $\chi$ and later, in order to discuss the adiabatic limit, we analyze a suitable limit where $h\to1$ as the limit of $\{h_n\}$ for $n\to \infty$ where $h_n$ is equal to $1$ on $[-n,n]$.

Even if the Bogoliubov map $R_V$ depends on the cutoff functions, 
the causal factorisation property implies that
 for $F\in\mathcal{F}_{\text{loc}}(O)$ (the set of local fields supported in $O$), with $O\subset \Sigma_{\epsilon}$ there exists an $N$, which could depend on $F$, such that for every $l,n \geq N\in \mathbb{N}$
\[
R_{V^{\chi h_n}}(F) = R_{V^{\chi h_l}}(F).  
\]
With this observation, we have that the limit $n\to\infty$ can easily be taken at algebraic level.  
As we shall see in short, in order to obtain an equilibrium state for the interacting theory, the state itself will depend on the cutoff function and in particular on $h$. We have thus to analyze the adiabatic limit at this level as well.

\subsection{Thermal states for the free theory}\label{se:thermal-state-free-theory}
The starting point for the construction of an interacting equilibrium state is the equilibrium state for the free theory.
On infinite systems, the condition of being an equilibrium state at a fixed temperature is captured by the Kubo-Martin-Schwinger (KMS) condition \cite{KMS}. 
We recall that a state $\omega$ on a $C^*-$algebra $\mathfrak{A}$ satisfies the KMS condition at inverse temperature $\beta$  with respect to the time evolution $\tau_t:\mathfrak{A} \to \mathfrak{A}$, which is a one parameter group of $*-$automorphisms if, for every $A,B\in \mathfrak{A}$ the function 
\[
f_{A,B}:t\mapsto \omega(A\tau_{t}B)
\] 
is analytic in the strip $\text{Im} t \in (0,\beta)$ and continuous at the boundaries and if 
\[
f_{A,B}(i\beta) = \omega(BA). 
\]
If the $C^*-$algebra $\mathfrak{A}$ is generated by a field, similar relations holds for the $n-$point functions computed in the state $\omega$. 
In the case of $*-$algebras, the KMS condition given above 
  can be given in terms of suitable relations satisfied by the $n-$point functions. 
Further details can be found e.g. in Definition 1 in \cite{FL} to which we refer for a complete definition in the $*-$algebraic case.
For free fields, the free time evolution is described by the $*-$automorphisms which implements the time translation and thus it is constructed in the following way
\[
\tau_t(F)=F_t, \qquad F_t(\varphi) = F(\varphi_t), \qquad \varphi_t(x)=\varphi(x+te_0). 
\] 
In the case of strictly positive mass $m$, there exists an unique quasifree KMS state at inverse temperature $\beta$ with respect to $\tau_t$ which we denote by $\omega^\beta$. The integral Kernel of the two-point function of this state has the form $w^\beta_m(x,y)=w^\beta(x,y)= \hbar \Delta^\beta(x-y)$ where 
\begin{align}
\Delta^\beta(t-\mathrm{i}u ,x) 
&:= 
\frac{1}{2\pi}\int   d p dp_0 e^{- \mathrm{i}p_0 t }    e^{- p_0 u}    e^{\mathrm{i} p x}    \frac{1}{1-e^{-\beta p_0}} \delta(p_0^2-p^2-m^2) \text{sign}(p_0)
\notag
\\
&= 
\frac{1}{2\pi}\int   d p    \frac{1}{2{\rm w}_p (1-e^{-\beta {\rm w}_p}) } \left(   e^{- \mathrm{i} {\rm w}_p t }    e^{- {\rm w}_p  u}   + 
e^{ \mathrm{i} {\rm w}_p t }    e^{-(\beta-u) {\rm w}_p }   \right) e^{\mathrm{i} p x}  
\label{eq:2pt-beta}
\end{align}
where ${\rm w}_p=\sqrt{p^2+m^2}$. $\Delta^\beta$ is an analytic function for $u\in(0,\beta)$. Since the state is quasifree, the $n-$point functions with $n\neq 2$ are given in terms of the two-point function. 
Notice that $\lim_{\beta\to\infty}\Delta^\beta = \Delta^\infty = \Delta^+_m$ given in \eqref{eq:massive2pt}.

In order to compute expectation values in the state $\omega^\beta$ it is easier to represent observables in $(\mathcal{F}^V,\star_{w^\beta})$. To this end, we observe in particular that
$s(x,y)=\hbar(\Delta^\beta-\Delta^\infty)(x-y)$ is a smooth function, 
hence, 
proceeding as in Proposition \ref{prop:algebra-well-defined}
we get that the action of $\star_{w^\beta}$ is closed in $\mathcal{F}^V$. Furthermore,
we can obtain elements of $(\mathcal{F}^V,\star_{w^\beta})$ simply applying $\alpha_{s}$ to $(\mathcal{F}^V,\star_{\omega^\infty})$.
The state $\omega^\beta$ is then simply represented by ${\rm{ev}}_0$ on $(\mathcal{F}^V,\star_{w^\beta})$.

The normal ordered interaction Lagrangian 
$\no{V^g}=\frac{1}{2}\no{V^g_a+V^{g}_{-a}}$ acquires an extra factor when represented in   $(\mathcal{F}^V,\star_{w^\beta})$.
In particular, we have
\[
\alpha_{\Delta^\beta}\no{V^g} = 
  \alpha_{s} V^g  =
e^{-\hbar \frac{a^2}{2} (\Delta^\beta-\Delta^\infty)(0)}V^g\in  (\mathcal{F}^V,\star_{w^\beta}).
\]
At the same time, we have that 
\begin{equation}\label{eq:omegabeta-omegainfinity}
w^{\beta}(\overline{f},f) \geq w^{\infty}(\overline{f},f),
\end{equation}
hence, by an application of conditioning and inverse conditioning, namely by an application of Theorem \ref{th:control}, we can control the theory at finite temperature with the one at zero temperature.  
Furthermore, when $S(e^{-\hbar \frac{a^2}{2} (\Delta^\beta-\Delta^\infty)(0)}V)$ is constructed in  $(\mathcal{F}^V,\star_{w^\beta})$, we have that an analogue of Theorem \ref{thm:convergence} and of Theorem \ref{thm:convergence-V}
holds with $V$ multiplied by the extra factor $e^{-\hbar \frac{a^2}{2} (\Delta^\beta-\Delta^\infty)(0)}$ and with $K$ constructed in such a way that \eqref{eq:omegabeta-omegainfinity} is taken into account.

\subsection{Thermal states for interacting theories and their properties} \label{se:thermal-state-interacting-theory}
To obtain a
KMS state at finite inverse temperature for the interacting theory, we follow the approach of \cite{FL}.
The starting point is the automorphisms of the interacting dynamics which is represented in the free algebra as
  \[
  \tau_t^V(R_V(F)) =  R_V(F_t).
  \]
We notice that $\tau_t^V$ and $\tau_t$ are intertwined by a cocycle. For $t\geq 0$, and for $F$ with $\supp F \subset J^+(\Sigma_0)$ we have
  \[
  \tau_t^V(\cdot ) =  U(t) \tau_t(\cdot) U(t)^*, \qquad U(t) = S(V)^{-1}S(V_t).
  \]
  The cocycle is such that
  \begin{equation}\label{eq:cocycle-prop}
  U(t+s)=U(t)\tau_t(U(s))=U(s)\tau_s(U(t)).
  \end{equation}
 The interacting KMS state is (see Fredenhagen Lindner  \cite{FL} and Araki \cite{Araki73})
\begin{equation}\label{eq:interacting-state}
\omega^{\beta,V}(F) = \frac{\omega^\beta(F U(\mathrm{i}\beta))}{\omega^\beta(U(\mathrm{i}\beta))},\qquad F\in \mathcal{A}^V(J^{+}(\Sigma_0)).
\end{equation}
Notice that $U(s)$ is constructed as a power series in the coupling constant 
whose coefficients are elements of 
$\mathcal{A}^V$. 
Hence, $\omega^{\beta,V}(F)$ is also as a power series in $\lambda$ too, we shall later prove that for certain choices of $F$ that power series converges also in the adiabatic limit.
Furthermore, the wellposedness of the analytic continuation of 
$s\mapsto \omega^\beta(FU(s))$ necessary to get $\omega^{\beta, V}$ can be shown to hold order by order proceeding as in \cite{FL}. 
The state $\omega^{\beta, V}$ constructed in this way
 
is invariant under time translation, actually
\begin{align*}
\omega^{\beta,V}(\tau_t^V F) &= \frac{\omega^\beta( U(t) \tau_t(F) U^{-1}(t) U(\mathrm{i}\beta))}{\omega^\beta(U(\mathrm{i}\beta))}
= \frac{\omega^\beta( \tau_t(F) U^{-1}(t) U(\mathrm{i}\beta) \tau_{\mathrm{i}\beta}(U(t)) )}{\omega^\beta(U(\mathrm{i}\beta))}\\
&= \frac{\omega^\beta( \tau_t(F) \tau_t (U(\mathrm{i}\beta))  )}{\omega^\beta(U(\mathrm{i}\beta))}= \frac{\omega^\beta( F U(\mathrm{i}\beta)  )}{\omega^\beta(U(\mathrm{i}\beta))},
\end{align*}
where we used the property \eqref{eq:cocycle-prop} in the   third equality and the
  KMS condition of $\omega^\beta$
in the last one.
$\omega^{\beta,V}$ satisfies the KMS condition with respect to the interacting time evolution $\tau^V$. For a proof of the analytic properties of $t\mapsto \omega^{\beta,V}(F\tau^V_t (G))$ we refer to the work of Araki \cite{Araki73}  or Fredenhagen Lindner \cite{FL}. Here we check that, thanks to the cocycle condition,
\begin{align*}
\omega^{\beta,V}(F\tau^V_{\mathrm{i}\beta}(G)) 
&= \frac{\omega^\beta( F U(\mathrm{i}\beta) \tau_{\mathrm{i}\beta}(G) U(\mathrm{i}\beta)^{-1}U(\mathrm{i}\beta))}{\omega^\beta(U(\mathrm{i}\beta))}
= \frac{\omega^\beta( G F U(\mathrm{i}\beta) )}{\omega^\beta(U(\mathrm{i}\beta))}
= \omega^{\beta,V}(GF).
\end{align*}
Hence, using the KMS condition of $\omega^\beta$, we may change $\chi$ to $\chi'$ without changing the state, actually,
assuming 
$\supp \dot{\chi}  \subset M\setminus J^{+}(\supp F) $ and $\supp \dot{\chi'}  \subset M\setminus J^{+}(\supp F) $.
we notice that
\[
S(V^\chi+F) = S(V^{\chi'} + F + V^{\chi-\chi'}) =S(V^{\chi'} + F) S(V^{\chi'})^{-1} S(V^{\chi}) = S(V^{\chi'} + F) W_{\chi'\chi}
\]
where, 
thanks to the fact that $F\gtrsim V^{\chi-\chi'}$, 
we used the casual factorisation property   of the $S$ matrix which descends from the causal factorisation property of the time ordered product, see \cite{Book} for further details.
Furthermore
\[
S_{V^\chi}(F) = W_{\chi'\chi}^{-1} S_{V^{\chi'}}(F) W_{\chi'\chi}
\]
and
\[
U_\chi(\mathrm{i}\beta) = S(V^\chi)^{-1} \tau_{\mathrm{i}\beta}(S(V^\chi)) = W_{\chi'\chi}^{-1} U_{\chi'}(\mathrm{i}\beta) \tau_{\mathrm{i}\beta}(W_{\chi'\chi}).
\]
Hence
\[
\omega^\beta(S_{V^\chi}(F) U_\chi(\mathrm{i}\beta))  = \omega^\beta( W_{\chi'\chi}^{-1} S_{V^{\chi'}}(F) U_{\chi'}(\mathrm{i}\beta) \tau_{\mathrm{i}\beta}(W_{\chi'\chi}))  = \omega^\beta( S_{V^{\chi'}}(F) U_{\chi'}(\mathrm{i}\beta))
\]
where in the last equality we used the KMS condition of the free state.
This proves that, if $\supp{\dot\chi}\theta^- \subset M\setminus J^{+}(\Sigma_0)$   where $\theta^{-}(x^0)=\theta(-x^0)$,
\[
\omega^{\beta,V^\chi}(S_{V^\chi}(F)) = \omega^{\beta,V^{\chi'}}(S_{V^{\chi'}}(F)), \qquad \supp F\subset    \Sigma_\epsilon
\]
namely the state restricted to observables supported in the future of $\Sigma_\epsilon$ does not depend on $\chi$.

\subsection{Choice of the time cutoff function and the adiabatic limits}\label{se:form-og-generating-functionals}
To compute the correlation functions of the KMS state at inverse temperature $\beta$ for the interacting theory we consider the following generating functional
\begin{equation}\label{eq:generating-functional}
{G}(f) := \omega^{\beta,V}({W}(f)), \qquad f\in  \mathcal{D}(M),
\end{equation}
where ${W}(f)$ is the exponential of $\Phi(f)$  which is namely
\[
{W}(f) = e^{\mathrm{i}\Phi(f)}
\]
and, up to the factor $e^{-\frac{1}{2} w^\beta(f,f)}$ it coincides with the exponential of $\mathrm{i}\Phi(f)$ computed as the limit of partial sums contained in the algebra $(\mathcal{F}_{\text{reg}},\star_{w^\beta})   \subset \mathcal{A}^V$
  which are thus contained in the domain of $\omega^{\beta,V}$.  Furthermore, the state $\omega^{\beta,V}$ is given in terms of the KMS state of the free theory in \eqref{eq:interacting-state}.
The correlation functions of $\omega^{\beta,V}$ as a state over the free theory $\mathcal{A}^V$ are then obtained by means of a functional derivative of ${G}(f)$ with respect to $f$.
Following Fredenhagen and Lindner \cite{FL}, in terms of the KMS state of the free theory $\omega^\beta$, the generating functional $G$ given in \eqref{eq:generating-functional} can be computed as the ratio in \eqref{eq:interacting-state} once we have the following explicit expressions
\begin{equation}\label{eq:FandZ}
\begin{aligned}
\omega^\beta(W(f)U(\mathrm{i}\beta)) &= \sum_n  \int_{\beta  \mathcal{S}_n}  dU \omega^\beta\left(W(f) \star \tau_{\mathrm{i}u_1}(R_{V^\chi}(K^\chi)) \cdot_{T}   \dots \cdot_{T} \tau_{\mathrm{i}u_n}(R_{V^\chi}(K^\chi))\right) 
\\
\omega^\beta(U(\mathrm{i}\beta)) &= \sum_n  \int_{\beta  \mathcal{S}_n}  dU \omega^\beta\left(\tau_{\mathrm{i}u_1}(R_{V^\chi}(K^\chi)) \cdot_{T}   \dots \cdot_{T} \tau_{\mathrm{i}u_n}(R_{V^\chi}(K^\chi))\right) 
\end{aligned}
\end{equation}
where $U=(u_1,\dots, u_n)$ and $\beta \mathcal{S}_n=\{U\in \mathbb R^{n}\mid 0\leq u_1\leq \dots \leq u_n\leq \beta\}$ and it is the $n-$dimensional simplex of edge $\beta$.
Furthermore, as shown in \cite{FL}, the interacting Hamiltonian has the form
\[
K^\chi = \int \dot{\chi}(x^0) \theta(-x^0)    h(x)   \cos(a\varphi(x)) d\mu_x
\]
where $\theta(-x^0)$ is the Heaviside step function.

We want to prove that the sums in \eqref{eq:FandZ}  exist and that the limit $h\to 1$ of their ratio can be taken. 
To this end, we use the freedom in the choice of $\chi$ discussed in the previous section and we choose a time cutoff function 
which is the Heaviside step function. Contrary to the 4d case this does not give origin to further divergences.
Actually, the estimates in the proofs of the Theorems given in the previous sections involve only integrals of the coupling constants and not of its derivatives.   
This implies that the limit where $\chi$ tends to the Heaviside step function can be taken. Furthermore, under this limit 
$\dot\chi \to {\delta}$ and, formally,
 \begin{equation}\label{eq:interacting-hamiltonian}
 R_{V^\chi}(K^\chi) \to \dot{V} =   \int  d\mathbf{x} \cos(a \varphi(0,\mathbf{x})) h(\mathbf{x}).
 \end{equation}
To be more precise the $n-$th order contribution in $\lambda$ to $R_{V^\chi}(K^\chi)$ is
\begin{gather*}
R_n = \sum_{(a_1,\dots,a_{n+1})\in \{-a,a\}^{n+1}}
\sum_{0\leq l\leq n}
\frac{1}{2^{n+1}}
\frac{(-\mathrm{i})^l(\mathrm{i})^{n-l}}{l! (n-l)!}
\int_{M^{n+1}}    
e^{\mathrm{i} \sum_i a_i \varphi(x_i)}
e^{-\sum_{1\leq i< j \leq l} a_ia_j\overline{w^{\beta,{\rm F}}}(x_i,x_j) } 
\\
e^{- \sum_{l+1\leq i< j \leq {n+1}}a_ia_j{w^{\beta,{\rm F}}}(x_i,x_j) }  
e^{- \sum_{1\leq i\leq l} \sum_{l+1< j<n+1} a_ia_j{w^\beta}(x_i,x_j) }  
g(x_1)\dots g(x_n)  \dot{g}(x_{n+1}) 
d\mu_{x_1} \dots d\mu_{x_{n+1}}
\end{gather*}
where $w^{\beta,{\rm F}}= w^{\beta}+\mathrm{i} \hbar \Delta^A$ is the Feynman propagator associated to $w^{\beta}$. 
Notice that the real part of  $w^{\beta,{\rm F}}$, $\overline{w^{\beta,{\rm F}}}$ and of $w^\beta$ coincides. We denote it by $\hbar H_m^\beta = {\rm Re} w^\beta$ and $H_m^\beta$ differs from $H_m$ in \eqref{eq:Hm} by a smooth function. 
Furthermore, by the causal properties of the Bogoliubov map the domain of integration can be (non-optimally) restricted to
$J^{-} = \{(x_1,\dots, x_n, x_{n+1}) | x_{n+1}\in M, x_j\in J^{-}(x_{n+1}), x_j^0 \geq 0 , j\in \{1,\dots, n\}  \} $,
hence 
\begin{gather*}
|R_n| \leq
\sum_{(a_1,\dots,a_{n+1})\in \{-a,a\}^{n+1}}
\frac{1}{2}
\int_{J^{-}}    
e^{-\sum_{1\leq i< j \leq n+1} a_ia_j \hbar H^\beta_m(x_i,x_j) } 
|g(x_1)\dots g(x_n)  \dot{g}(x_{n+1}) \theta(-x_{n+1}^0)| 
d\mu_{x_1} \dots d\mu_{x_{n+1}}.
\end{gather*}
With an estimate similar to the one used to prove that $S(V)$ is an absolutely convergent series for every $\varphi$ we can thus prove that the series defining $R_V(K)$ is an absolute convergence series for every $\varphi$.  

Notice that $e^{-\sum_{1\leq i< j \leq n+1} a_ia_j {\hbar H^\beta_m}(x_i,x_j) }$ is a locally integrable function because $H^\beta_m$ differs from $H_m$ by a smooth function and, as discussed above, $e^{-\sum_{1\leq i< j \leq n+1} a_ia_j \hbar H_m(x_i,x_j) }$ is locally integrable.  Furthermore, in the limit where $\chi \to \theta$,  $J^-$ becomes a zero measure set. 
Since $g$ is uniformly bounded by $1$ also in the limit, this implies that all the integrals over $x_1, \dots, x_n$ in $R_n$ are integrals of integrable functions over zero measure sets hence they give zero abd thus
$|R_n| \to 0$ uniformly in $\varphi$ for $n\geq 1$ while $R_0\to \dot{V}$.

 Notice that $\dot{V}$ is outside $\mathcal{A}^V$, however, $G(f)$ in \eqref{eq:generating-functional} and the correlation functions computed directly with $\chi=\theta$ are still finite.  
 In order to further simplify \eqref{eq:FandZ} we may use the time slice axiom \cite{ChF08} of the free theory 
 to restrict the observables (generating functional) to $\Sigma_0$. In this way the problem   of taking the adiabatic limit is only one dimensional.
 
Actually, for $f\in C^{\infty}_0(J^+(\Sigma_0);\mathbb{R})$ we have, up to solutions of the free equation of motion, 
\begin{equation}\label{eq:time-slice-free}
e^{\mathrm{i}\Phi(f)}
=
e^{\mathrm{i} \sigma(\varphi,\psi_f)}e^{\mathrm{i} \Phi(P (\chi_{2}\Delta^A f))}
\end{equation}
where $\psi_f= \Delta(f)$, $\Delta$ being the causal propagator of the free theory, $\Delta^A$ is the 
advanced propagator $P=\square - m^2$ is the differential operator which   represents the free equation of motion,
$\chi_2$ is a smooth function invariant under spatial translation which vanishes on   $(-\infty, -\epsilon)$ and which is $1$ on $J^+(\Sigma_0)$.
Furthermore, 
\[
\sigma(\varphi,\psi_f) = \int_{M} \left(\varphi \partial_t \psi_f -  \partial_t\varphi \psi_f \right) \partial_t \chi_2.
\]
Notice that the integrand appearing at the right hand side is supported on $(-\epsilon,0)\times \mathbb{R}\subset M$
and in the limit $\epsilon\to 0$, namely when $\chi_2$ tends to the heaviside step function $\sigma(\varphi,\psi_f)$ converges to the standard symplectic form of the free theory. 
Since $\Phi(P (\chi_{2}\Delta^A f))$ vanishes when evaluated on a state, to obtain the generating functional $G(f)$, we just need to consider expectation values of $e^{\mathrm{i} \sigma(\varphi,\psi_f)}$, namely expectation values of
\begin{equation}\label{eq:A}
A(\psi,\psi') := e^{\mathrm{i}\Phi(\psi \partial_t \chi_2)}  e^{\mathrm{i}\partial_t \Phi(\psi'\partial_t \chi_2)}
\end{equation}
with $\psi$ and $\psi'$ compactly supported real functions whose support is near $\Sigma_0$   selected in such a way that $\psi \partial_t \chi_2= \partial_t\psi_f\partial_t{\chi_2}$ and
$\psi' \partial_t \chi_2= -\psi_f\partial_t{\chi_2}$. 
Summarizing this discussion, 
assuming to have $A$ as in \eqref{eq:A} with $\psi$ and $\psi'$ supported near $\Sigma_0$
and taking the limit $\chi\to\theta$ and $\chi_2\to \theta$ we have
\begin{align*}
\omega^\beta(W(f)U(\mathrm{i}\beta)) &= \omega^\beta(AU(\mathrm{i}\beta)) = \sum_n  \frac{\lambda^n}{\hbar^n} \int_{\beta  \mathcal{S}_n}  dU \omega^\beta\left(A(\partial_t \psi_f,-\psi_f) \star \tau_{\mathrm{i}u_1}(\dot{V}) \cdot_{T}   \dots \cdot_{T} \tau_{\mathrm{i}u_n}(\dot{V})\right)
\end{align*} 
where both $\star$ and $\cdot_T$ are defined with respect to $\omega^\beta$.
More explicitly, denoting by $\psi(\mathbf{x}):=\psi(0,\mathbf{x})$ we have
\begin{equation}\label{eq:Fexp-Zexp}
 \begin{aligned}
\omega^\beta(A(\psi,\psi') U(\mathrm{i}\beta))
&=
\sum_n  \frac{1}{2^n} \frac{\lambda^n}{\hbar^n} \int_{\beta  \mathcal{S}_n}  dU  
\int dX h^{\otimes n}(X)
\\
&\sum_{a_i = \pm a}
e^{-\sum_{i} a_i\hbar \int d\mathbf{x}_0 \psi(\mathbf{x}_0) C^\beta_m(u_i,\mathbf{x}_0-\mathbf{x}_i)}   
e^{+{\mathrm{i}}\sum_{i} a_i\hbar \int d\mathbf{x}_0 \psi'(\mathbf{x}_0) \partial_u C^\beta_m(u_i,\mathbf{x}_0-\mathbf{x}_i)}  
\\
&
e^{-\sum_{1\leq i<j\leq n} a_ia_j\hbar C^\beta_m(u_i-u_j,x_i-x_j)} 
\end{aligned}
\end{equation}
where $h^{\otimes n}(X)=h(x_1)\dots h(x_n)$ and where \textbf{the covariance} $C^\beta_m(u,x) := \Delta^\beta(-\mathrm{i}u ,x)$ with $\Delta^\beta$ given in \eqref{eq:2pt-beta}.
We preliminary study some properties of the quantities appearing in \eqref{eq:Fexp-Zexp}. 
In view of the form of \eqref{eq:2pt-beta} we have that 
\begin{equation}\label{eq:defC}
\begin{aligned}
C^\beta_m(u,x) = \frac{1}{\hbar} \Delta^\beta(-\mathrm{i}u ,x) 
&= 
\frac{1}{4\pi}\int_{-\infty}^\infty   d p    \frac{1}{{\rm w}_p (1-e^{-\beta {\rm w}_p}) } \left( e^{- u {\rm w}_p  }   +  e^{-(\beta-u) {\rm w}_p }   \right) e^{i xp }  
\\
&= 
\frac{1}{2\pi}\int_{0}^\infty   d p     \frac{1}{{\rm w}_p}  \frac{\cosh ( (\frac{\beta}{2}-u) {\rm w}_p )}{\sinh(\frac{\beta}{2} {\rm w}_p)}   \cos( xp )  
\end{aligned}
\end{equation}
where ${\rm w}_p=\sqrt{p^2+m^2}$.
Hence we observe that  $C^\beta_m(u,x)$ on $(0,\beta)\times \mathbb{R}$ is real.
Furthermore, it possesses the following symmetries
\begin{equation}\label{eq:C-symmetry}
C^\beta_m(u,x)=C^\beta_m(u,-x), \qquad C^\beta_m(u,x)=C^\beta_m(\beta -u,-x).
\end{equation}
Notice that 
\begin{lemma}\label{le:positiveC}
$C^\beta_m(u,x)$ is positive on $(0,\beta)\times\mathbb{R}$.
\end{lemma}
\begin{proof}
 
In view of the symmetry \eqref{eq:C-symmetry}, we just need to check the claim for $u\in(0,\frac{\beta}{2})$ and for $x \geq 0$.
The case $x=0$ descends for the last equality in \eqref{eq:defC}.
To treat the case $x>0$ we preliminary observe that for $u>0$ and $x>0$ it holds that the modified Bessel function of the second kind $K_0$, also called Mac Donald function satisfies the following integral representations
\begin{equation}\label{eq:Kint}
K_0(m\sqrt{x^2+u^2}) = \frac{1}{2}\int_{-\infty}^\infty  dp \frac{e^{- u {\rm w}_p}}{{\rm w}_p} \cos(p x)
=
\int_{m}^\infty  dw \frac{e^{- m\sqrt{x^2+u^2} w}}{\sqrt{w^2-m^2}}.
\end{equation}
These equalities hold by 
an application of Jordan Lemma, used to prove that for $x>0$ 
\[
\mathcal{J}_\alpha = \frac{1}{2}\int_{\Gamma_\alpha} \frac{e^{i z x}}{\sqrt{z^2+1}} dz
\]
with $\Gamma_\alpha=\{z(t)=\cos\alpha \sinh{t} + i \sin\alpha \cosh{t}| t\in \mathbb{R}\}$, does not depend on $\alpha\in [0,\frac{\pi}{2}]$.
Furthermore, one of the integral representations the modified Bessel function of the second kind $K_0$ given e.g. in 10.32.6 in \cite{NIST:DLMF}
gives that $\mathcal{J}_0(x) = K_0(x)$. 

By the last equality in \eqref{eq:Kint}, we thus have that $K_0(mx)$ is positive. Moreover, we can write $C^\beta_m(u,x)$ as a linear combination of modified Bessel functions of the second kind, actually, from \eqref{eq:defC} we have
\[
C^\beta_m(u,x) = \frac{1}{4\pi} \int_{-\infty}^\infty dp \frac{1}{{\rm w}_p} \sum_{n=0}^\infty e^{-\beta n {\rm w}_p} \left(e^{-u{\rm w}_p}+e^{-(\beta-u){\rm w}_p} \right) \cos(xp),
\]
hence for $x>0$
\[
C^\beta_m(u,x) = \frac{1}{2\pi} \sum_{n=-\infty}^\infty  K_0 (m \sqrt{x^2 + (\beta n+ u)^2}) 
\]
and thus $C^\beta_m(u,x)$ is positive because it is a sum of positive functions.
\end{proof}
Notice that in view of \eqref{eq:2pt-beta} $\Delta^\beta(-\mathrm{i}u ,x) =  C^{\beta}_m(u,x)  $, hence
the imaginary part of $\Delta^\beta(-\mathrm{i}u_i,\mathbf{x}_0-\mathbf{x}_i)$ is 0 and 
the real part of $\Delta^\beta(-\mathrm{i}u_i,\mathbf{x}_0-\mathbf{x}_i)$ is  equal to $ C^\beta_m(u_i,\mathbf{x}_0-\mathbf{x}_i)$.
We furthermore observe that 
\[
\partial_t \Delta^\beta(-\mathrm{i}u,{x})  = \mathrm{i} \partial_u C^\beta_m(u,x)
\]
and thus $\partial_t \Delta^\beta_m(-\mathrm{i}u,{x}) $ is purely imaginary because  $C^\beta_m$ is a real function.
We thus have that if $f$ and thus $\psi'_f$ are real, the contribution due to 
\[
e^{-\sum_{i} a_i \hbar\int d\mathbf{x}_0 \psi'(\mathbf{x}_0) \partial_t \Delta^\beta(-\mathrm{i}u_i,\mathbf{x}_0-\mathbf{x}_i)}  
= e^{-\mathrm{i}\sum_{i} a_i\hbar \int d\mathbf{x}_0 \psi'(\mathbf{x}_0) \partial_uC(u_i,\mathbf{x}_0-\mathbf{x}_i)}  
\]
is a pure phase.

The limit $h\to 1$ of  $\frac{\omega^\beta({W}(f) U(\mathrm{i}\beta))}{\omega^\beta(U(\mathrm{i}\beta))}$ can be taken adapting the estimates of the cluster  expansion discussed in section III of the work of Glimm Jaffe Spencer \cite{GJS}, see also Fr\"ohlich and Seiler \cite{FS}. 
Notice that they discuss the cluster expansion of a two-dimensional theory. Here we are in a simpler case, the limit we have to do is $h\to 1$ and $h$ is supported on some interval of $\mathbb{R}$.

\section{Covariance and boundary conditions}\label{se:covariance-bc}

Let us recall $C^\beta_m$ given in \eqref{eq:defC} which is the covariance without boundary conditions and possess the following symmetris
\[
C^\beta_m(u,x)=C^\beta_m(u,-x), \qquad C^\beta_m(u,x)=C^\beta_m(\beta -u,-x) 
\]
on $(0,\beta)\times \mathbb{R}$.
Notice that, having a translation invariant covariance $C^\beta_m(u,x_1-x_2)$, 
we can construct a covariance with Dirichlet boundary condition in $x = b$ using the image charge method.
The new covarinace $C^b(u,x_1,x_2)$ is non zero only if $x_i \geq b $ $\forall i\in\{1,2\}$ or $x_i \leq b$ $\forall i\in\{1,2\}$. 
Furthermore, if $x_i \geq b$ $\forall i\in\{1,2\}$ we have 
\begin{equation}\label{eq:Cb}
\begin{aligned}
C^{b}(u,x_1,x_2) &=  C^\beta_m(u,x_1-x_2) - C^\beta_m(u,x_1-(-(x_2-b) +b))  \\
&=  C^\beta_m(u,x_1-x_2) - C^\beta_m(u,x_1 + x_2-2b ) 
\end{aligned}
\end{equation}
if we now consider an interval  $[a,b]$, a covariance with Dirichlet boundary condition is formed by three components. If both $x_1>b$ and $x_2>b$ just use the covariance obtained above, if both $x_1 < a$ and $x_2 < a$ we just adapt the previous convariance, if both $\{x_1,x_2\}\subset (a,b)$ we use recursively the image charge method. To obtain
\begin{equation}\label{eq:CX}
C_{(a,b)}(u,x_1,x_2) =  \sum_{n=-\infty }^\infty \left( C^\beta_m(u,x_1-x_2  +2(b-a)n) - 
C^\beta_m( u,x_1 + x_2-2b   +2(b-a)n)\right)
\end{equation}
In view of the exponential decay for large $|x|$ of $C^\beta_m(u,x)$ the infinite sum present above is convergent. 
The resulting covariance decays exponentially. 

\bigskip

Consider now $\Gamma$, a finite set of elements of $\mathbb{Z}$, and decompose $\mathbb R \setminus \Gamma$ in its connected components $X_i$
\[
\mathbb R \setminus \Gamma = \bigcup_{i=0}^{|\Gamma|} X_i .
\]
For simplicity we also assume that $X_i<X_{i+i}$, (namely for $x\in X_i$ and $y\in X_{i+1}$ we have $x<y$).
In $\{ X_i\mid i\in \{0,\dots, |\Gamma|\}\}$ there are two non compact elements $X_0$ and $X_{|\Gamma|}$ and the other $|\Gamma|-1$ elements are compact. 
The distribution $C^{\Gamma}(u,x_1,x_2)$ is now the covariance with $0-$Dirichlet data on $\Gamma$, hence 
\begin{itemize}
\item $C^{\Gamma}(u,x_1,x_2)$  is non vanishing only if $x_1,x_2$ lie in the same $X_i$
\item if $x_1,x_2 \in  X_i$ for some $0<i<|\Gamma|$ 
\[
C^{\Gamma}(u,x_1,x_2) = C_{X_i}(u,x_1,x_2)
\]
where $C_{X_i}$ is given in \eqref{eq:CX},
\item if $x_1,x_2 \in  X_i$ for $=0$ or $i=\Gamma$
\[
C^{\Gamma}(u,x_1,x_2) = C^{\partial X_i}(u,x_1,x_2)
\]
where $C^{\partial X_i}$ is given in \eqref{eq:Cb}
 \end{itemize}
Later we shall use $C^\Gamma$ at the place of $C^\beta_m$ in \eqref{eq:interacting-state} to to analyze the adiabatic limits. Hence, in the next we shall derive some estimates satisfied by $C^\beta_m$ and by $C^\Gamma$

\subsection{Path integral representation of \texorpdfstring{$C^\beta_m$}{Cbetam}}

For later purposes it will be useful to have another representation of $C^\beta_m$. We recall that, using the Wiener measure,
we can rewrite the (one dimensional) fundamental solution of the Heat equation  
\[
K(x,y; T) = \frac{1}{\sqrt{2\pi T}}e^{ - \frac{(x-y)^2}{2T}}
\]
as a sum over the paths joining $x$ and $y$ in a time $T$. 
\begin{equation}\label{eq:wiener-measure}
K(x,y; T)  = \int  d\mu_{x,y}^T(\omega)  
\end{equation}
where $\mu_{x,y}^T $ is the Wiener measure. 

We observe that the spatial Fourier transform\footnote{Here we use the convention that $\hat{f}(p)=\int  {e^{\mathrm{i}px} f(x)}dx$.} of $K(x,y; T)$ is $\hat{K}(p,T) = e^{-\frac{p^2 T}{2}}$, hence, to write $\frac{1}{p^2 }$ as a sum over paths we integrate $K(x,y; T)$ in time, actually
\[
\hat{G} = \frac{1}{p^2 } =  \lim_{T\to\infty}\frac{1}{2}\int_0^T  \hat{K}(p,t) dt, \qquad p\neq 0,
\]
and
\[
G(x,y) = \frac{1}{4\pi}\int_0^\infty  \int  d\mu_{x,y}^t(\omega)      dt. 
\]
With these observations we now construct $C^\beta_m(u,x-y)$, we compute its Fourier transform in $x-y$ and its Fourier coefficients in $u$. We get
\[ 
\hat{C}^\beta_m(n,p) = \frac{c}{w(p)^2+\left(\frac{2\pi}{\beta} n\right)^2}  = \lim_{T\to\infty}\frac{c}{4\pi}\int_0^{T} e^{- \left(w(p)^2+\left(\frac{2\pi}{\beta} n\right)^2\right) \frac{t}{2}}dt
\]   
where $c$ is a suitable constant,
hence recalling that $w(p)^2 = p^2 + m^2$ and proceeding as before we have that 
\[
\tilde{C}(n;x,y)  =  \lim_{T\to\infty}\frac{c}{4\pi}\int_0^{T} \int  e^{- \left(m^2+\left(\frac{2\pi}{\beta} n\right)^2\right) \frac{t}{2}}  d\mu_{x,y}^t(\omega)  dt
\] 
hence 
\begin{equation}\label{eq:path-int-rep}
\begin{aligned}
{C^\beta_m}(u,x,y)  &=  \sum_{n} e^{\mathrm{i}\frac{2\pi}{\beta} n  u} \tilde{C}(n,x,y)  =  \lim_{T\to\infty} \frac{c}{4\pi} \sum_{n} e^{\mathrm{i}\frac{2\pi}{\beta} n  u}
\int_0^{T} \int  e^{- \left(m^2+\left(\frac{2\pi}{\beta} n\right)^2\right) \frac{t}{2}}  d\mu_{x,y}^t(\omega)  dt\\
&=  
\int_0^{\infty} \int  e^{- \frac{m^2}{2} t} W(t,u)  d\mu_{x,y}^t(\omega)  dt   
\end{aligned}
\end{equation}
where $W(t,u)$ is a suitable positive weight.
Actually for $q = e^{- \left(\frac{2\pi}{\beta} \right)^2 \frac{t}{2}} $ and $z=e^{\mathrm{i}\frac{2\pi}{\beta} u}$
\[
W(t,u) = \sum_{n=-\infty}^\infty e^{\mathrm{i}\frac{2\pi}{\beta} n  u} e^{- \left(\frac{2\pi}{\beta} n\right)^2 \frac{t}{2}} 
=\sum_{n=-\infty}^\infty z^n q^{n^2} =  \prod_{m=1}^\infty (1-q^m)(1+z q^{2m-1} )(1+z^{-1} q^{2m-1} )
\]
where the last equation is nothing but the Jacobi triple product formula. Observing that 
$(1+z q^{2m-1} )(1+z^{-1} q^{2m-1} ) = (1+ (z+z^{-1})q^{2m-1} + q^{4m-2})$ we get that $W(t,u)$ is positive because $0< q\leq 1$ and $z+z^{-1}=2\cos(\frac{2\pi}{\beta} u)$ and thus $W(t,u)$ is a product of positive factors.

Now if a boundary condition is imposed we have simply to restrict the set of paths over which the Wiener integral is taken  \cite{Roep,GJ12,GJS,spencer1974,Ciesielski}. In particular, let $\Gamma$ be a set of points of $\mathbb{R}$. To represent $C_\Gamma(x,y)$ the covariance which vanishes for $y$ in $\Gamma$, the Wiener integral is restricted to paths which do not touch the elements of  $\Gamma$.  
So 
let $\omega(t)$ be a path in $\mathbb{R}$ starting in  $x$ and ending in $y$, and let $b \in\mathbb{R}$ a point where we impose boundary conditions,   we introduce the following
\[
J^T_b(\omega) = \begin{cases}
0 \qquad \text{if} \quad \omega(t) = b  \quad \text{for some } t\in(0,T) \\ 
1 \qquad  \text{otherwise}
\end{cases}
\]
hence
\begin{equation}\label{eq:Cb-path}
C^b(u,x;y) =  
\int_0^{\infty} \int  e^{- \frac{m^2}{2} T} W(T,u) J^T_b(\omega)d\mu_{x,y}^T(\omega)  dT  
\end{equation}
and
\begin{equation}\label{eq:Cab-path}
C_{(a,b)}(u,x;y) =  
\int_0^{\infty} \int  e^{- \frac{m^2}{2} T} W(T,u) \prod_{p\in\{a,b\}}
J^T_p(\omega)d\mu_{x,y}^T(\omega)  dT   
\end{equation}

\subsubsection{Decay properties of $C^\beta_m$ for large separation}

We derive some bounds satisfied by $C^\beta_m$ given in \eqref{eq:defC}. 

\begin{prop}\label{pr:decayC}
Consider $C^\beta_m$ given in \eqref{eq:defC}, for $x>\alpha$ it holds that
\[
|C_m^\beta(u,x)| \leq c_\beta \frac{e^{-\frac{m}{\sqrt{2}} |x|}}{m}
\] 
where the bounds holds uniformly in $u$ and the constant is such that
\[
c_\beta  =  
\frac{2}{\beta}   
\frac{1}{1-e^{  -\frac{\alpha}{\sqrt{2}}   \frac{2\pi }{\beta}}}.
\]
\end{prop}
\begin{proof}
Since $C_m^\beta(u,x)$ is symmetric we need to check only the case $x>0$.
Notice that thanks to Lemma \ref{le:positiveC} we know that $C$ is positive, furthermore, from \eqref{eq:defC}, we have that 
\begin{align*}
C_m^\beta(u,x) &= \frac{1}{2\pi}\int_{-\infty}^\infty   d p     \frac{1}{{\rm w}_p}  \frac{\cosh ( (\frac{\beta}{2}-u) {\rm w}_p )}{\sinh(\frac{\beta}{2} {\rm w}_p)}   e^{ipx}\\
&= \frac{1}{2\pi}\int_{-\infty}^\infty   d p     \frac{1}{{\rm w}_p^2}  \frac{{\rm w}_p \cosh ( (\frac{\beta}{2}-u) {\rm w}_p )}{\sinh(\frac{\beta}{2} {\rm w}_p)}   e^{ipx} 
\end{align*}
where ${\rm w}_p=\sqrt{p^2+m^2}$.
We use some complex analysis to estimate this integral. We observe that the integrand is an holomorphic function in $p$ which has simple poles where either ${\rm w}_p^2$ or $\sinh(\frac{\beta}{2} {\rm w}_p)/{\rm w}_p$ vanish. 
Compute a contour integral over $\gamma_R$ which is a counterclockwise oriented
contour formed by the interval $[-R,R]$ of the real line and by a semicircle of radius $R$ centred in the origin and lying in the upper half plane. 
We notice that for $x>0$,  in the limit $R\to\infty$ the contribution on the semicircle vanishes because  $e^{ipx}$ vanishes exponentially for $p$ with large imaginary part and so in that limit the integral over $\gamma_R$ gives the desired result. 
We use residue theorem to estimates these integral. 
The poles contained in $\gamma_\infty$ are located at $p=P_n:=\mathrm{i}\sqrt{m^2+ \left(\frac{2\pi n}{\beta}\right)^2}$ for $n\in \mathbb{N}$. 
More precisely $1/{\rm w}_p^2$ has a pole in $P_0$ whose residue is $\frac{1}{2im}$ while 
$f=\frac{{\rm w}_p}{\sinh(\beta {\rm w}_p/2)}$ has poles for $p=P_n$ with $n\geq1$ and the corresponding residues are 
\[
\text{Res}(f)_{P_n} =  (-1)^n \frac{4n\pi}{\beta^2} \frac{1}{\sqrt{m^2+ \left(\frac{2\pi n}{\beta}\right)^2}}.
\]
Hence we get the following estimate
\begin{align*}
|C_m^\beta(u,x)| &\leq  \frac{1}{\beta m} e^{-m|x|} + \frac{2 }{\beta} \sum_{n\geq 1} \frac{1}{\sqrt{m^2+ \left(\frac{2\pi n}{\beta}\right)^2}}e^{-|x| \sqrt{m^2+ \left(\frac{2\pi n}{\beta}\right)^2}}
\\
& \leq  \frac{1}{\beta m} e^{-m|x|} + \frac{2 }{\beta m }  
\sum_{n\geq 1} e^{- |x| \sqrt{m^2+ \left(\frac{2\pi n}{\beta}\right)^2}}
\\
& \leq  \frac{1}{\beta m} e^{-m|x|} + \frac{2 }{\beta m}  
\sum_{n\geq 1} e^{-  \frac{|x|}{\sqrt{2}}  \left(  m +   \frac{2\pi n}{\beta}\right)}
\\
& \leq  \frac{1}{\beta m} e^{-m|x|} + \frac{2 }{\beta m}e^{-m\frac{|x|}{\sqrt{2}}}  
\frac{1}{e^{  \frac{|x|}{\sqrt{2}}   \frac{2\pi }{\beta}}-1}.
\end{align*}
This concludes the proof.
\end{proof}

For later purposes it will be useful to derive a bound for the the following
\begin{align} \label{eq:DGammaC}
D_\Gamma C(u,x,y) = 
\int_0^{\infty} \int  e^{- m^2 T} W(T,u) 
\prod_{i\in\Gamma}(1-J^T_i(\omega))
d\mu_{x,y}^T(\omega)  dT
\end{align}
where $\Gamma$ is some set of points in $\mathbb{Z}\subset \mathbb{R}$,  
where $d\mu_{x,y}^T(\omega)$ is the Wiener measure and $W(T,u) $ is the weight apperaing in the path integral representation of $C^\beta_m$
which we recall here
\begin{align*}
C^\beta_m(u,x,y)  &=  \frac{1}{4\pi}
\int_0^{\infty} \int  e^{- \frac{m^2}{2} T} W(T,u)   d\mu_{x,y}^T(\omega)  dT. 
\end{align*}
We observe that the factors 
$\prod_{i\in\Gamma}(1-J^T_i(\omega))$ in $D_\Gamma C$
force the Wiener integral to be taken over the paths that touches all the points in $\Gamma$.

We observe that the following proposition holds
\begin{prop}\label{pr:boundDGammaC}
Consider $\Gamma$ a subset of $\mathbb{Z}\subset \mathbb{R}$ formed by at least two elements and $D_{\Gamma} C$ given in \eqref{eq:DGammaC}.
Then, for every couple of points $x,y \in \mathbb{R}$ and $u\in[0,\beta]$,
\[
|D_\Gamma C(u,x,y)| \leq  c \frac{1}{m} e^{-\frac{m}{\sqrt{2}} \Delta(\Gamma)} e^{-\frac{m}{\sqrt{2}} d(x,\Gamma)} e^{-\frac{m}{\sqrt{2}} d(y,\Gamma)} 
\]
where $\Delta(\Gamma)$ is the maximum of the distances of between any two points in $\Gamma$, $d(x,\Gamma)$ is the minimum distance between $x$ and the extreme points of $\Gamma$ and where $c$ is a suitable constant which depends only on $\beta$. 
\end{prop}
\begin{proof}
The points in $\Gamma$ are contained on $\mathbb{R}$ and thus they can be ordered according to their real value. 
Without loosing generality we suppose that $x<y$ and we consider some cases.
  Suppose for now that the two points $x,y$ are chosen in such a way that $x<\Gamma$ and $y> \Gamma$. In this case, because of continuity and because $\mathbb{R}$ is one dimensional, the paths  joining $x$ and $y$ in the  Wiener integral representation of $D_\Gamma C$ are forced to pass through all the points of $\Gamma$, hence $\prod_{i\in\Gamma}(1-J^T_i(\omega))$ gives one for every path joining $x$ and $y$ and thus in this case
\begin{align*}
D_\Gamma C(u,x,y) = 
\int_0^{\infty} \int  e^{- m^2 T} W(T,u) 
\prod_{i\in\Gamma}(1-J^T_i(\omega))
d\mu_{x,y}^T(\omega)  dT
= C(x,y)
\end{align*}
where $C=C^\beta_m$.
Hence, in this case Proposition \ref{pr:decayC} yields the claim.
Suppose now that $x< \Gamma$, $x<y$ but there is one point $z$  in  $\Gamma$ (the largest extreme) which is larger then $y$.
In this case the paths $\omega$ contributing to the Wiener integral in \eqref{eq:DGammaC}  joining $x$ and $y$ are forced to pass through $z$. Let $\Omega_z$ be the set of all the paths joining $x$ and $y$ in a time $T$ passing through $z$. We notice a path $\omega\in\Omega_z$ can be divided in two parts. $\omega$ is actually formed by a path $\omega_1$ joining $x$ and $z$ in a time $T'<T$ and a path $\omega_2$ joining $z$ and $y$ in a time $T-T'$. Here $T'$ is the first time at which $\omega$ reaches $z$. Notice that thanks to the properties of the Wiener measure, we can reflect the second branch  of the path, $r_z\omega_2 = \omega_2^r$ with respect to $z$ without altering the weight associated by the Wiener measure. Actually, if we denote by $\omega'$ the path formed by $\omega_1$ and $\omega_2^r$ we have that
\[
d\mu_{x,y}^T(\omega) = d\mu_{x,{2z-y}}^T(\omega').
\]
If we do this for all the paths in the set $\Omega_z$ we get $\Omega'$ which is formed all the paths joining $x$ and $2z-y$ in a time $T$. The map we have described is a bijection from $\Omega_z$ to $\Omega'$, hence 
\begin{align*}D_\Gamma C(u,x,y) 
&= 
\int_0^{\infty} \int_{\Omega}  e^{- m^2 T} W(T,u) 
d\mu_{x,y}^T(\omega)  dT
= 
\int_0^{\infty} \int_{\Omega'}  e^{- m^2 T} W(T,u) 
d\mu_{x,2z-y}^T(\omega')  dT
\\&=D_\Gamma C(u,x,2z-y) = C(u,x,2z-y)
\end{align*}
where as before $C=C^\beta_m$.
We have reduced this case to the first one, hence the claim holds also in this case.
The same arguments can be applied also in the symmetric situations. It remains to discuss the case where 
both $x<y$ are contained in $\Gamma$, namely where 
$z_{\text{min}} < x< y<z_{\text{max}}$ where $z_{\text{min}}$ and $z_{\text{max}}$ are the minimum and maximum of $\Gamma$.
In this case we shall divide the set of considered paths $\Omega_{z_{\text{min}},z_{\text{max}}}$ in two $\Omega_1$ and $\Omega_2$ where 
$\Omega_1$ contains the paths which start at $x$ and reach $z_{\text{min}}$ for the first time before $z_{\text{max}}$ and
$\Omega_2=\Omega_{z_{\text{min}},z_{\text{max}}}\setminus \Omega_1$.
Now a path in $\Omega_1$ can be divided in two contributions. The one from $x$ to $z_{\text{max}}$, and from $z_{\text{max}}$ to  to $y$. It is then possible to reflect the second branch of the path with respect to $z_{\text{maz}}$ 
without altering the weight associated to the Wiener measure. If we do this reflection for all the paths in $\Omega_1$ we get a set of paths joining $x$ to $2z-y$ in a time $T$ and passing through $z_{\text{min}}$ before reaching $z_{\text{max}}$. 
This set is strictly smaller than the $\Omega_{z_{\text{min}}}(x,2z_{\text{max}}-y)$, the set of paths joining $x,2z_{\text{max}}-y$ and passing through $z_{\text{min}}$. Despite this fact we can bound 
the Wiener integral over $\Omega_1$ with the integral over $\Omega_{z_{\text{min}}}(x,2z_{\text{max}}-y)$. Then operating as in the second analyzed case  (namely the case where $z_{\text{min}}<x<z_{\text{max}}<y$), this last integral is again equivalent to the integral over $\Omega(2z_{\text{min}}-x,2z_{\text{max}}-y)$, namely
\[
\int_{\Omega_1} d\mu^T_{x,y} \leq \int_{\Omega_{z_{\text{min}}}(x,2z_{\text{max}}-y)} d\mu^T_{x,2z_{\text{max}}-y}
\leq \int d\mu^T_{2z_{\text{min}}-x,2z_{\text{max}}-y}.
\]
The contribution over $\Omega_2$ can be treated similarly. We have thus that 
\[
D_\Gamma  C(u,x,y) \leq  D_{\Gamma} C(u,2z_{\text{min}}-x,2z_{\text{max}}-y) + D_{\Gamma} C(u,2z_{\text{min}}-y, 2z_{\text{max}}-x)
\]
and both contributions at the right hand side of the inequality can be bounded as in the first part of the proof hence concluding the proof.
\end{proof}

For the case where $\Gamma$ is formed by a single element we have the following proposition. 
\begin{prop}\label{pr:boundDGammaC1point}
Consider $\Gamma=\{b\}$ with $b\in \mathbb{Z}\subset \mathbb{R}$ and $D_{\Gamma} C$ given in \eqref{eq:DGammaC}.
Then, for every couple of points $x,y \in \mathbb{R}$ with $|x-b|+|b-y|>\alpha$
\[
|D_\Gamma C(u,x,y)| \leq  c \frac{1}{m} e^{-\frac{m}{\sqrt{2}} d(x,\Gamma)} e^{-\frac{m}{\sqrt{2}} d(y,\Gamma)} 
\]
where  $d(x,\Gamma)$ is the distance between $x$ and $b$ and where $c$ is a suitable constant which depends on $\alpha$ and on $\beta$. ($c$ diverges is $\alpha\to0$.)
\end{prop}
\begin{proof}
We proceed as in the proof of Proposition \ref{pr:boundDGammaC}. In the expansion of $D_\Gamma C$ as a sum over paths joining $x$ and $y$, we have that all paths needs to pass through $b$. Furthermore, 
 the constraint $|x-b|+|b-y|>\alpha$ implies that the minimal length of these paths is $\alpha$. 
Using the symmetry of the problem in the same way as in the proof of Proposition \ref{pr:boundDGammaC} 
and the decay properties of $C^\beta_m$ stated in Proposition \ref{pr:decayC} we have proved the claim.
\end{proof}

For later purposes we need also the following bounds
\begin{lemma}\label{le:lpnormC}
Consider $C^\beta_m$ given in \eqref{eq:defC}. We have for $p\geq 1$  and for large values of $m$ that 
\[
\|C^\beta_m\|^u_p \leq \frac{c}{m^{\frac{1}{p}}}
\]
where the sub script $u$ indicates that the integral in the $p$ norm are computed on the $x$ direction only, keeping $u$ fixed and where the constant $c$ is uniform for $u\in [0,\beta]$. 
\end{lemma}
\begin{proof}
Because of the symmetry of $C^\beta_m$ given in \eqref{eq:C-symmetry}, without loosing generality we restrict ourself to the case $u\in[0,\frac{\beta}{2}]$. 
Furthermore, triangle inequality implies that 
\[
\|C^\beta_m\|^u_p \leq \|C^\infty_m\|^u_p  +  \|C^\beta_m - C^\infty_m\|^u_p
\]
Let us recall that  
\[
C^\infty_m(u,x) = \Delta^+_m(iu,x) = \frac{1}{2\pi} K_0(m\sqrt{x^2+u^2})
\]
where $K_0$ is the modified Bessel function, which is positive and monotonically decreasing function, hence,
\[
|C^\infty_m(u,x)| \leq  \frac{1}{2\pi} K_0(m\sqrt{x^2})
\]
hence 
\[
\int dx \left|  C^\infty_m(u,x)\right|^p  \leq  \int dx \left| \frac{1}{2\pi} K_0(m\sqrt{x^2})\right|^p  \leq \frac{c^p}{m}
\]
where $c$ is a positive constant which does not depend on $m$. Hence

\[
\|C^\infty_m\|^u_p \leq \frac{c}{m^{\frac{1}{p}}}.
\] 
We discuss now the second contribution $C^\beta_m-C^\infty_m$, notice that, since $u\in[0,\beta/2]$, we have
\begin{align*}
|C^\beta_m(u,x)-C^\infty_m(u,x)|
&\leq  \frac{1}{4\pi}\int_{-\infty}^\infty   d p    \frac{e^{-\beta {\rm w}_p }}{{\rm w}_p (1-e^{-\beta {\rm w}_p}) } \left( e^{- u {\rm w}_p  }   +  e^{u {\rm w}_p }\right)   
\\
&\leq  \frac{1}{4\pi}\int_{-\infty}^\infty   d p    \frac{e^{-\frac{\beta}{2} {\rm w}_p }}{{\rm w}_p} \frac{\cosh(u {\rm w}_p)}{\sinh(\frac{\beta}{2}  {\rm w}_p) }
\\
&\leq  \frac{1}{4\pi}\int_{-\infty}^\infty   d p    \frac{e^{-\frac{\beta}{2} {\rm w}_p }}{{\rm w}_p} \left(1+\frac{2}{\beta {\rm w}_p}   \right) 
\\
&\leq  \frac{1}{2\pi}\int_{0}^\infty   d p    \frac{e^{-\frac{\beta}{2} p }}{m} \left(1+\frac{2}{\beta m}   \right) 
\\
&\leq  \frac{1}{2\pi} \frac{2}{\beta m }  \left(1+\frac{2}{\beta m}   \right)  
\end{align*}
Hence
\[
\|C^\beta_m-C^\infty_m\|_\infty^u \leq \frac{1}{2\pi} \frac{2}{\beta m }  \left(1+\frac{2}{\beta m}   \right) .
\]
We furthermore observe that $C^\beta_m(u,x)-C^\infty_m(u,x)$ is positive. This can be proved as in the proof of Lemma \ref{le:positiveC} noticing that 
\[
C^\beta_m(u,x)-C^\infty_m(u,x) = \int_0^\infty dx f({\rm w}_p) \cos (px) 
\]
where  $f({\rm w}_p) = \frac{1}{2\pi}\frac{e^{-\frac{\beta}{2} {\rm w}_p }}{{\rm w}_p} \frac{\cosh(u {\rm w}_p)}{\sinh(\frac{\beta}{2}  {\rm w}_p) }$ which is a positive monotonically decreasing function of ${\rm w}_p$ for every $u\in [0,\frac{\beta}{2}]$. Hence, 
\begin{align*}
\|C^\beta_m-C^\infty_m\|_1^u=\int_{-\infty}^\infty dx (C^\beta_m(u,x)-C^\infty_m(u,x))  
=
\frac{1}{2}    \frac{e^{-\frac{\beta}{2} m }}{m} \frac{\cosh(u m)}{\sinh(\frac{\beta}{2}  m) }
\leq    \frac{1}{2m} \left(1+\frac{2}{\beta m}   \right)  
\end{align*}
hence, for $u\in[0,\frac{\beta}{2}]$
\[
\|C^\beta_m-C^\infty_m\|_p^u \leq   \left((\|C^\beta_m-C^\infty_m\|_\infty^u)^{{p-1}}  \|C^\beta_m-C^\infty_m\|_1^u\right)^{\frac{1}{p}} \leq        \frac{c}{m} \frac{1}{\beta^{1-\frac{1}{p}}} \left(1+\frac{2}{\beta m}   \right)   \leq \frac{\tilde{c}}{m^{\frac{1}{p}}} 
\]
where the last inequality holds for
for $\beta m>1$ and for a suitable redefinition of the constant.
Combining these bounds we have proved the claim.
\end{proof}

\section{Cluster expansion and the adiabatic limit}\label{se:cluster-expansion}

In this section we discuss the existence of the following limit
\begin{equation}\label{eq:limith-1}
\lim_{h\to 1} \frac{\omega^\beta(A(\psi,\psi') U(\mathrm{i}\beta))}{\omega^\beta(U(\mathrm{i}\beta))},
\end{equation}
where the numerator and denominator of the ration at the right hand side are computed as in \eqref{eq:Fexp-Zexp}. Here $A(\psi,\psi')$ is given in \eqref{eq:A} where its connection with the generating functional $G(f)$ given in \eqref{eq:generating-functional} is also discussed.
We shall adapt to the present context, the cluster expansion proposed by Glimm Jaffe and Spencer \cite{GJS} which was also used in the context of Sine-Gordon theories in the euclidean domain in \cite{FS}.   See \cite{AR95} and reference therein for further generalizaiton of the original cluster expansion.

\subsection{Basic ideas of the method}
The basic idea of the cluster expansion is to analyze the limit \eqref{eq:limith-1} for a theory where the covariance $C^\beta_m(u,x-y)$ is replaced with a 
covariance $C^\Gamma(u,x,y)$ of a theory with vanishing boundary conditions on $\Gamma\subset\mathbb{Z}\subset \mathbb{R}$ which was constructed in section \ref{se:covariance-bc}. 
Notice that $\mathbb{R}\setminus \Gamma$ is a union of disjoint open intervals and
the correlation functions among two elements which are supported on two disjoint intervals vanishes if they are computed with $C^\Gamma(u,x,y)$ for this reason the elements of $\Gamma\subset\mathbb{Z}\subset \mathbb{R}$ are called {\bf bonds}. 
Furthermore the limit $h\to1$ can easily be taken in the case of $C^\Gamma$ when $\Gamma$ has sufficiently many elements. 
The idea is thus to control the limit obtained in \eqref{eq:limith-1} with a particular expansion of similar contributions computed with $C^\Gamma$ for various $\Gamma$.
In order to see how the limit can be taken for $C^\Gamma$ with a $\Gamma$ with sufficiently many elements we proceed as follows.
Consider now a generic finite set of bonds $\Gamma\subset\mathbb{Z}$,  $\mathbb{R}\setminus \Gamma$ is the union of $r=|\Gamma |+1$ disjoint open connected components $X_i$ 
\[
\bigcup_{i=1}^r X_i = \mathbb{R}\setminus \Gamma.
\]
Let us discuss 
$\omega^h_{C^\Gamma}(AU(\mathrm{i}\beta))$
where $\omega^{h_\Lambda}_{C^\Gamma}$ is the positive functional obtained with the covariance $C^\Gamma$ and $h_\Lambda$ is the characteristic function of the set $\Lambda \subset \mathbb{R}$.
We shall denote this contribution by
\[
F_{C^\Gamma}(\Lambda) := \omega^{h_{\Lambda}}_{C^\Gamma}(AU(\mathrm{i}\beta)).
\]
Hence, its explicit expression is as in the right hand side of equation \eqref{eq:Fexp-Zexp} where $C^\beta_m$ is replaced by $C^\Gamma$ and where $h_\Lambda$, the characteristic function of the set $\Lambda$, is used at the place of $h$. 
Similarly, the corresponding contribution with $A=1$ is denoted by
\[
Z_{C^\Gamma}(\Lambda) := \omega^{h_{\Lambda}}_{C^\Gamma}(U(\mathrm{i}\beta)).
\]
Its explicit expression is again as in \eqref{eq:Fexp-Zexp} with $C^\Gamma$ at the place of $C^\beta_m$ and with $h_\Lambda$ at the place of $h$.
With this notation, the adiabatic limit \eqref{eq:limith-1} we aim to control is 
\[
\lim_{h\to 1} \frac{\omega^\beta(A(\psi,\psi') U(\mathrm{i}\beta))}{\omega^\beta(U(i\beta))} = \lim_{\Lambda\to\mathbb{R}}  \frac{F_{C^\beta_m}(\Lambda)}{Z_{C^\beta_m}(\Lambda)}.
\]
Notice that in view of the clustering properties of $C^\Gamma$, $Z_{C^\Gamma}(\Lambda)$, factorizes in a product over the connected components $X_i$ of $\mathbb{R}\setminus \Gamma$,
namely
\[
Z_{C^\Gamma}(\Lambda)= \prod_i  Z_{C^{\Gamma\cap \overline{X_i}}}(\Lambda\cap X_i).
\]
Furthermore, we denote by $Y = \left. J(\supp(W(f))) \right|_{\Sigma_0}= \supp(A(\psi,\psi'))$, and we introduce the following two disjoint sets
$Y_1$  $Y_2$ formed by suitable unions of the sets $X_i$  introduced above 
\begin{align*}
Y_1:=\bigcup_{X_j \cap Y \neq \emptyset}  {{X}_j} ,
\qquad \qquad  
Y_2:=\bigcup_{X_j \cap Y = \emptyset}    {{X}_j}   .
\end{align*}
Notice that $Y_2$ is non empty only if $\Gamma \not\subset Y$. 
With this decomposition, if $Y_2\neq \emptyset$, $Y_2$ is not compact and furthermore we have that 
also $F_{C^\Gamma}(\Lambda)$ 
factorizes in a product of two contributions
namely
\[
F_{C^\Gamma}(\Lambda) = F_{C^{\Gamma}}(\Lambda\cap Y_1) F_{C^\Gamma}(\Lambda\cap Y_2)  
\]
we furthermore observe that $F_{C^\Gamma}(\Lambda\cap Y_2) = Z_{C^\Gamma}(\Lambda\cap Y_2) $ 
because $A$ is $1$ on $Y_2$, hence
\[
F_{C^\Gamma}(\Lambda) = F_{C^{\Gamma}}(\Lambda\cap Y_1) Z_{C^\Gamma}(\Lambda\cap Y_2).
\]
If $J\supp(W(f))|_{\Sigma_0}$ is contained in $Y_1$, which is a union of $X_i$ of compact support, 
the limit $\Lambda\to \mathbb{R}$ which correspond to the limit $h\to 1$ can be easily taken, namely
\[
\lim_{\Lambda\to \mathbb{R} } \frac{F_{C^\Gamma}(\Lambda)}{Z_{C^\Gamma}(\Lambda)}
=
\frac{F_{C^\Gamma}(Y_1)}{Z_{C^\Gamma}(Y_1)}
\]
where $h_{Y_1}$ is the characteristic function of the set $Y_1$.
In the next section we shall see how obtain \eqref{eq:limith-1} from the previous limit by means of the cluster expansion of $\frac{\omega^\beta(A(\psi,\psi') U(\mathrm{i}\beta))}{\omega^\beta(U(\mathrm{i}\beta))} = \frac{F_{C^\beta_m}}{Z_{C^\beta_m}}$.

\subsection{The details of the method}
We use now the cluster expansion and its properties discussed above to construct the adiabatic limit of the original theory.
We consider a covariance operator $C(s)$ which is a convex sum of $C^\Gamma$ over all possible $\Gamma \subset \mathbb{Z}$. More precisely, $s$ is the collection of elements, $s_b\in[0,1]$ for every $b\in \mathbb{Z}$,  
\[
s=\{s_b\}_{b\in\mathbb{Z}}.
\]
The covariance operator $C(s)$ depending on all these parameters is
\begin{equation}\label{eq:covariance-convex-sum}
C(s) = \sum_{\Gamma\subset\mathbb{Z}} 
\prod_{i\in \Gamma} s_i \prod_{j\in \Gamma^c} (1-s_j)  C^{\Gamma^c}
\end{equation}
where $\Gamma^c$ is the complement of $\Gamma \in \mathbb{Z}$.
With this definition we have that 
$C(\{0,0,0,\dots\})$ is the solution with zero Dirichlet data on every bond $i\in\mathbb{Z}$, while $C(\{1,1,1\dots\} )$ corresponds to the the original covariance $C^\beta_m$.
We substitute $C^\beta_m$ with $C(s)$ in the explicit expressions of $\omega^\beta(A(\psi,\psi') U(\mathrm{i}\beta))$ and of $\omega^\beta(U(\mathrm{i}\beta))$,  given in \eqref{eq:Fexp-Zexp} and  we use the following notation 
\begin{align*}
F(s,\Lambda)&:=F_{C(s)}(\Lambda)=\omega_{C(s)}^{h_\Lambda}(A(\psi,\psi') U(\mathrm{i}\beta))\\
Z(s,\Lambda)&:=Z_{C(s)}(\Lambda)=\omega_{C(s)}^{h_\Lambda}(U(\mathrm{i}\beta)).
\end{align*}
Notice that $F(0,0,\dots)$ and $Z(0,0,\dots )$ it is completely clustered while while $F(1,1,\dots) = \omega^\beta(A(\psi,\psi') U(\mathrm{i}\beta))$ and $Z(1,1,\dots ) = \omega^\beta(U(\mathrm{i}\beta))$.

\begin{prop}\label{pr:cluster-exp-00}
The {\bf cluster expansion} is then a sum over al possible finite set of bonds $\Ga\subset \mathbb{Z}$ and it is
\begin{equation}\label{eq:cluster-exp1}
F(s) = \sum_{\Ga\subset\mathbb{Z}; |\Ga|<\infty}  \int_{0\leq \sigma \leq s(\Ga)}  \partial^{\Ga} F(\sigma(\Ga))
d\sigma
\end{equation}
where
\[
s(\Ga) = \begin{cases}
s_b,\qquad b\in\Ga
\\
0,\qquad b\in\Ga^c.
\end{cases}
\] 
Furthermore $\sigma\leq s$ if and only if $\sigma_b\leq s_b$ for all $b$ and
\[
\partial^\Ga  = \prod_{b\in\Ga}   \frac{\partial}{\partial s_b}.
\]
Hence, at the right hand side of \eqref{eq:cluster-exp1}, the integral is taken over a region of dimension $|\Ga|$.
\end{prop}
The proof of \eqref{eq:cluster-exp1} is given in Proposition 3.2 of \cite{GJS}, here we recall the main steps of this proof.
\begin{proof}
The essential idea of the proof is the following, notice that 
\begin{align*}
F(s_0,0,0\dots) &= F(0,\dots) + \int_0^{s_0} d\tilde{s}_0  \frac{d}{d\tilde{s}_0}F(\tilde{s}_0,0,\dots).
\end{align*}
This formula can be written as
\[
F(s_0,0,0\dots) = (E_0 +\delta_0) F
\]
where $E_i$ and $\delta_i$ are two operators acting on the $i-$th components of the $s$ coordinates.
$E_iF(s):=F(\dots,s_{i-1},0,s_{i+1}\dots)$ is the evaluation in $0$ of the $i-$th component of the entry of $F$ and $\delta_i F(s) =  \int_0^{s_i}   \frac{d}{ds_i} F(\dots,s_{i-1},\sigma,s_{i+1}\dots) d\sigma$.
Hence, if one uses this formula for every bond one gets the desired expansion. 
\begin{align*}
F(s) &= \prod_{i\in\mathbb{Z}} (E_i +\delta_i) F(s)\\
&= \sum_{\Gamma\subset\mathbb{Z}} \prod_{i\in\Gamma}\delta_i F(s(\Gamma))\\
&=\sum_{\Ga\subset\mathbb{Z}; |\Ga|<\infty}  \int_{0\leq \sigma \leq s(\Ga)}  \partial^{\Ga} F(\sigma(\Ga))
d\sigma
\end{align*}
\end{proof}

Consider now $\Gamma\subset \mathbb{Z}$ with finite number of elements and let $\Gamma^c$ be its complement in $\mathbb{Z}$. Divide the space $\mathbb{R}\setminus \Gamma^{c}$ in connected components $X_i$ 
\begin{equation}\label{eq:space-decomposition}
\mathbb{R}\setminus \Gamma^{c} = \cup_{i}X_i.
\end{equation}
Notice that $s(\Gamma)_j$ vanishes for every $j$ on  $\Gamma^c$ hence, if $\supp A(\psi,\psi')\subset X_0$, we have that 
\begin{equation}\label{eq:combine}
F(s(\Gamma),\Lambda) = \prod_i F(s(\Gamma \cap X_i),\Lambda \cap {X_i})
\end{equation}
if this property holds,  we say that $F(s)$ {\bf decouples at $s=0$}. Similarly, we have also that $Z(s)$ decouples at $s=0$.
We use now this property in \eqref{eq:cluster-exp1} in order to partially resum the series.
 
We operate as follows, suppose now that $Y_0 = \supp A(\psi,\psi')$, for a generic $\Gamma$, consider the decomposition $X_i$ of $\mathbb{R}\setminus \Gamma^c$ given in \eqref{eq:space-decomposition}. We combine these sets in two disjoint sets
$Y_1$  $Y_2$ formed by suitable unions of the sets $X_i$ introduced above 
\begin{equation}\label{eq:Y12}
Y_1:=\bigcup_{X_j \cap Y_0 \neq \emptyset}  {{X}_j} ,
\qquad \qquad  
Y_2:=\bigcup_{X_j \cap Y_0 = \emptyset}    {{X}_j}   . 
\end{equation}
Introduce now $\Gamma_i= \Gamma\cap Y_i$ we have that  \eqref{eq:cluster-exp1} with \eqref{eq:combine} gives
\begin{equation}\label{eq:cluster-exp2}
F(s,\Lambda) = \sum_{\Ga} \prod_{i\in\{1,2\}}\int_{0\leq \sigma \leq s(\Ga_i)}  \partial^{\Ga_i} F(\sigma(\Ga_i),\Lambda\cap Y_i)
d\sigma.
\end{equation}
A similar decomposition holds also for $Z(s,\Lambda)$ and, because of the support of $A(\psi,\psi')$ we have that the second factors in each element of the sum in $F$ coincides with those of $Z$.

There are various subsets $\Gamma$ which lead to the same  $Y_1$ and $\Gamma_1$,  we partially resum over all these $\Gamma$s. 
We obtain 
\[
F(s,\Lambda) = \sum_{Y_1,\Gamma_1}  \left(\; \int\limits_{0}^{s(\Ga_1)}  \partial^{\Ga_1} F(\sigma_1(\Ga_1),   
\Lambda\cap Y_1)  d\sigma_1
\left(\sum_{\Gamma_2}
\int\limits_{0}^{s(\Ga_2)}  \partial^{\Ga_2} Z(\sigma_2(\Ga_2),    \Lambda\cap Y_2)  d\sigma
\right)\right).
\]
The sum over all possible $\Gamma_2 \subset \mathbb{Z}\setminus Y_1$ can be taken and using \eqref{eq:cluster-exp1} it gives 
\[
\sum_{\Gamma_2}\int\limits_{0}^{s(\Ga_2)}  \partial^{\Ga_2} Z(\sigma_2(\Ga_2),    \Lambda\cap Y_2)  d\sigma = F(s(\mathbb{Z}\setminus Y_1 ),\Lambda\setminus Y_1)
\] furthermore, if $s_b=1$ for every $b\in\mathbb{Z}$,
\begin{align*}
F(s(\mathbb{Z}\setminus Y_1 ),\Lambda\setminus Y_1) &= Z(s(\mathbb{Z}\setminus Y_1 ),\Lambda\setminus Y_1)\\
&=Z_{C^{\partial Y_1}}(\Lambda\setminus Y_1)
\end{align*}
where $Z_{C^{\partial Y_1}}(\Lambda\setminus Y_1)$ is computed with the covariance $C^{\partial Y_1}$ which has bonds in $\partial Y_1$. 
Hence, calling $Y=Y_1$ and $\Gamma=\Gamma_1$ 
\begin{equation}\label{eq:cluster-exp3}
F(s,\Lambda) = \sum_{{Y,\Ga}} 
 \int\limits_{0}^{s(\Ga)}  \partial^{\Ga} F(\sigma(\Ga),   
\Lambda\cap Y)  d\sigma   Z_{C^{\partial Y}}(\Lambda\setminus Y).
\end{equation}
The cluster expansion is now obtained evaluating $F$ at $s=1$ and dividing by $Z(\Lambda)$. We have actually the following proposition
\begin{prop}\label{pr:cluster-exp}
The expectation value of $\omega^{\beta,V^{h_\Lambda}}(W(f))$ where $h_\Lambda$ is the characteristic function of the interval $\Lambda\subset \mathbb{R}$ can be written  as
\begin{equation}\label{eq:cluster-exp4} 
S(\Lambda):= \frac{\omega^{\beta, h_\Lambda}(A(\psi,\psi') U(i\beta))}{\omega^{\beta, h_\Lambda}(U(i\beta))}  = \sum_{{Y,\Ga}} 
 \int\limits_{0}^{s(\Ga)}  \partial^{\Ga} 
F(\sigma,\Lambda\cap Y)    d\sigma \;   \frac{Z_{C^{\partial Y}}(\Lambda\setminus Y)}{Z(\Lambda)}
\end{equation}
where $s(\Gamma)_i$ is $1$ if $i\in\Gamma$ and $0$ otherwise.
The sum in the previous formula is taken over all possible $Y,\Ga$ with the following properties following from 
\eqref{eq:Y12}. Let $Y_0$ be a fixed compact interval which contains the support of $A(\psi,\psi')$ on $\Sigma_0$ then
\begin{itemize}
\item[a)] $Y_0\subset Y$ and $Y=(a,b)\subset \mathbb{R}$ with $a,b\in\mathbb{Z}$, 
\item[b)] $\Ga\subset Y$, 
\item[c)] $(\Ga^c\cap Y) \subset Y_0$.
\end{itemize}
\end{prop}
The last requirement in the choice of $\Gamma$ over which the sum in \eqref{eq:cluster-exp4} is taken implies in particular that for 
a given $Y$, $\Gamma$ is than chosen in such a way that all the bonds in $Y$ outside $Y_0$ are in $\Gamma$, namely 
\[
(\mathbb{Z} \cap  ({Y} \setminus Y_0))\subset \Ga
\]

\begin{thm}\label{thm:adiabatic-limit}
Under the hypothesis of the Proposition \ref{pr:cluster-exp},
it holds that the limit 
\[
\omega^{\beta,V}(W(f)) = \lim_{\Lambda \to \mathbb{R}} S(\Lambda) 
\]
exists.
\end{thm}
\begin{proof}
The limit $\La\to \mathbb{R}$ in each term of the sum on $Y,\Ga$ in \eqref{eq:cluster-exp4} can be taken because of the clustering  properties discussed above.
To prove the converge of the sum we make the following observations.
Condition $a)$ in Proposition \ref{pr:cluster-exp} implies that 
the number of $Y$ in the sum \eqref{eq:cluster-exp3} with fixed volume $|Y|$ is bounded by $|Y|$. 
Furthermore, once $Y$ is chosen for $Y_0$ of compact support, conditions $b)$ and $c)$ satisfied by the elements of the sum given in equation \eqref{eq:cluster-exp3} in Proposition \ref{pr:cluster-exp}, imply that the number of possible $\Ga$ is bounded by $2^{|Y_0|}$ which does not depend on $\Lambda$. 
With this two observations, Lemma \ref{le:cluster-convergence-1} and Lemma \ref{le:cluster-convergence-2} imply now that for sufficiently large mass $m$ the sum in 
\eqref{eq:cluster-exp3} converges. (Notice in particular that for each term in the sum \eqref{eq:cluster-exp3} every $h$ is bounded to be supported in $Y$. Hence the estimate of Lemma \ref{le:cluster-convergence-2} is used with $Y$ at the place of $\Lambda$).
In the proof we have also used that $|\Gamma|$ is essentially the volume $|Y|$, more precisely
\[
|Y| - |Y_0| \leq |\Gamma| \leq |Y|.
\]
and $|Y_0|$ is the volume of $Y_0$ which depends only on the support of $\psi$ and $\psi'$ or of $f$.
\end{proof}

We now estimate $Z_{C^{\partial Y}}(\Lambda\setminus Y)/Z(\Lambda)$ and the sum over $Y,\Ga$ with fixed $|Y|$.
We can do it using proposition 5.2 of \cite{GJS}, see also a similar analysis in \cite{FS}.

\begin{lemma}\label{le:cluster-convergence-1}
It holds that
\[
\left|\frac{Z_{C^{\partial Y}}(\Lambda\setminus Y)}{Z(\Lambda)}\right| \leq e^{k |Y|}
\]
where $k$ does not depend on $\beta$, on $\Lambda$ and on $m$ for sufficiently large $m$ and for small $a$. 
\end{lemma}
\begin{proof}
We present a detailed proof in the case $Y=[l_1,l_2] \subset \Lambda$ with $l_1<l_2$ and $l_1,l_2\in\mathbb{Z}$.
This proof can be obtained noticing that 
\begin{equation}\label{eq:ineq-setp}
C^{\partial Y} \leq C^\beta_m.
\end{equation}
This inequality can be proved using the path integral representation of $C^\beta_m$ given in \eqref{eq:path-int-rep} and observing the $C^{\partial Y}$ is a similar integral \eqref{eq:Cb-path} which is however taken over a smaller set of paths. Furthermore, as discussed after \eqref{eq:path-int-rep}, the weight $e^{- \frac{m^2}{2} t} W(t,u)$ which appears in \eqref{eq:path-int-rep} and in \eqref{eq:Cb-path} is  positive.
Now, knowing \eqref{eq:ineq-setp}, conditioning and the clustering properties of $C^{\partial Y}$ used in a similar way as in the proof of Lemma \ref{le:omega(U)} imply that 
\[
Z(\Lambda)= \omega^{\beta, h_\Lambda}(U(i\beta)) \geq \omega_{C^{\partial Y}}^{\beta,h_\Lambda}(U(i\beta))
=
\omega_{C^{\partial Y}}^{\beta, h_Y}(U(i\beta))\omega_{C^{\partial Y}}^{\beta,h_{\Lambda \setminus Y}}(U(i\beta))
= Z_{C^{\partial Y}}(\Lambda\setminus Y)Z_{C^{\partial Y}}(Y)
\]
from which
\begin{equation}\label{eq:ineq-1}
\left|\frac{Z_{C^{\partial Y}}(\Lambda\setminus Y)}{Z(\Lambda)}\right|\leq \left|\frac{1}{Z_{C^{\partial Y}}(Y)}\right|.
\end{equation}
The last step can be accomplished as follows. Similarly to \eqref{eq:ineq-setp} we have that 
\[
C^{\partial Y} \geq C(0,\dots)
\]
hence again conditioning and the clustering properties of $C(0)$ give 
\[
Z_{C^{\partial Y}}(Y)\geq Z_{C(0)}(Y)= (Z_{C(0)}(Q))^{|Y|} 
\]
where $Q=[0,1]$ and $|Y|=l_2-l_1$ is the volume of $|Y|$, this chain of inequalities implies that
\begin{equation}\label{eq:ineq-2}
\left|\frac{1}{Z_{C^{\partial Y}}(Y)}\right| \leq e^{k| Y|} \leq e^{k|\Lambda|}
\end{equation}
for some constant $k$.
Combining \eqref{eq:ineq-1} with  \eqref{eq:ineq-2} we have proved the claim.
\end{proof}

The limit $\La\to \mathbb{R}$ in each term of the sum on $Y,\Ga$ in \eqref{eq:cluster-exp3} can be taken because of the clustering  properties discussed above.
The sum can be proven to converge thanks to the following observation (Proposition 5.3 in \cite{GJS} and Proposition 2.2 of \cite{FS})

\begin{lemma} \label{le:cluster-convergence-2}
Consider $\Gamma$ a finite set of elements of $\mathbb{Z}$,
\begin{equation}\label{eq:basic}
\left| \int\limits_{0}^{s(\Ga)}  \partial^{\Ga} 
F(\sigma,\Lambda)
   d\sigma \;   
\right|  
\leq e^{- k_1|\Ga| + k_2 |\La|}
\end{equation}
$k_2$ does not depend on $m$. 
$k_1$ depends on the mass $m$. For any $k_1$ there exists $m_0(k_1)$ such that for any $m>m_0(k_1)$ the bound holds.
\end{lemma}

\begin{proof} 
We start analyzing how $F(\sigma,\Lambda)$ depends on $\psi$ and $\psi'$. 
To this end, recalling \eqref{eq:Fexp-Zexp}, we observe that $(\psi,\psi')$ enters 
$\omega^\beta(A(\psi,\psi') U(i\beta))$ as exponential of terms proportional to
\begin{equation}\label{eq:k1}
k(u, {x}) =  {-  \int d{x}_0 \psi({x}_0) C^\beta_m(u,{x}_0-{x})}   
{+\mathrm{i} \int d{x}_0 \psi'({x}_0) \partial_u C^\beta_m(u,{x}_0-{x})}. 
\end{equation}
Notice that this contribution can be written as 
\begin{equation}\label{eq:k2}
k(u, {x}) =   \int dx_0 du_0 f(u_0,x_0) \tilde{C}(u_0-u,{x}_0-{x})   + g(u,x)
\end{equation}
for two suitable smooth functions $f$ and $g$ supported on $[0,\beta]\times \mathbb{R}$ with compact support in the second variable. Here $\tilde{C}(u,x) = {C}^\beta_m(u,x)$ if $u\in[0,\beta]$ and  $\tilde{C}(u,x) = {C}^\beta_m(-u,x)$ if $u\in[-\beta,0)$. 
To find the precise forms of $f$ and $g$ we construct a current $J$ on $[0,\beta]\times \mathbb{R}$ whose component are
$$
J_j(u_0,x_0) = p(u_0,x_0) \partial_j  \tilde{C}(u_0-u,{x}_0-{x}) - \partial_j p(u_0,x_0)   \tilde{C}(u_0-u,{x}_0-{x}), \qquad j\in \{0,1\}.
$$
$p\in C^\infty([0,\beta\times \mathbb{R}])$ and $p$ is compact in space.
We notice that the divergence of the corresponding current gives 
\begin{align*}
\partial_j J_j(u_0,x_0) 
&= p(u_0,x_0) \Delta_{u,x}   \tilde{C}(u_0-u,{x}_0-{x}) - \Delta_{u,x} p(u_0,x_0)  \tilde{C}(u_0-u,{x}_0-{x}) 
\\
&=   m^2  p(u_0,x_0)    \tilde{C}(u_0-u,{x}_0-{x}) + p(u_0,x_0)  \delta(u_0-u)\delta({x}_0-{x})
 - \Delta_{u,x} p(u_0,x_0)  \tilde{C}(u_0-u,{x}_0-{x})
\end{align*}
where $\Delta_{u,x}= \partial_u^2+\partial_x^2$ and where we used the fact that $\Delta_{u,x}\tilde{C}(u,x) = m^2\tilde{C}(u,x)  + \delta(u) \delta(x)$. Here $\delta$ are ordinary delta functions.
We observe also that, for $u\in[0,\beta]$ in view of the symmetry satisfied by $C^\beta_m$,  
\[
\tilde{C}(\beta-u,{x}) =\tilde{C}(-u,{x}), 
\qquad
\partial_u \tilde{C}(\beta-u,{x}) =  \partial_u\tilde{C}(-u,{x}).
\]
Hence divergence theorem implies that 
\begin{align*}
\int_{\mathbb{R}}\int_{(0,\beta)} du_0 dx_0 \partial_j J_j(u_0,x_0) 
=& 
-\int dx_0  \left( p(\beta,x_0) - p(0,x_0) \right) \partial_u  \tilde{C}^\beta_m(u,{x}_0-{x}) 
\\
&- \int dx_0 \left(\partial_u p(\beta,x_0) - \partial_u p(0,x_0) \right) \tilde{C}^\beta_m(u,{x}_0-{x}).
\end{align*}
Choosing 
\[
p(u_0,x_0) =    \frac{\mathrm{i}}{\beta} (u_0-\frac{\beta}{2})  \psi'(x_0)- \frac{1}{2 \beta}(u_0-\frac{\beta}{2})^2 \psi(x_0)
\]
the boundary integrals coincide with $k(u,x)$ in \eqref{eq:k1}. The searched functions $f$ and $g$ in 
\eqref{eq:k2}
  are thus given in terms of $p$ as
\[
f=m^2p - \Delta_{u,x} p , \qquad g=p.
\] 
Recalling \eqref{eq:Fexp-Zexp}, we notice that the terms $e^{g}$ obtained as the exponential of \eqref{eq:k2} can be seen as a rescaling of $h$ to $\tilde{h}=e^{g}h$. For simplicity we shall perform the subsequent analysis discarding this factor, the correct estimate can be restored a posteriori.  
The relevant contribution which requires an estimate is thus 
$\omega(\mathcal{W}(f)U(\mathrm{i}\beta))$.
Here $\mathcal{W}(f)= \exp_{\tilde{C}}{\mathrm{i}\phi(f)}$, the exponential of the linear field $i\phi(f)$ taken with respect to $\tilde{C}$. We discuss the detail of the proof in the case of $f \in C^{\infty}_0([0,\beta]\times \mathbb{R})$. 
The same results can then be obtained also in the desired limit.
Hence, by direct inspection, we have that 
\begin{equation}
F(\sigma,\Lambda)  = {\rm{ev}}_0\left(\mathcal{W}(f)\cdot_{\tilde{C}(\sigma)} \exp_{\tilde{C}(\sigma)} 
\frac{2\lambda \dot{V}_{\beta}^{{h}}}{\hbar}\right)
\end{equation}
for $\mathcal{W}(f)= \exp_{\tilde{C}(\sigma)}{\mathrm{i}\phi(f)}$ where $f$ is a smooth compactly supported function on $[0,\beta]\times\mathbb{R}$.
Furthermore, $\dot{V}_\beta^h = \int dx \int_0^\beta du h(x)  \lambda \cos(a \varphi(x)) $.
We have that
\begin{align*}
F(\sigma,\Lambda)  &= {\rm{ev}}_0\left(\exp_{\tilde{C}(\sigma)}{\mathrm{i}\phi(f)}\cdot_{\tilde{C}(\sigma)} \exp_{\tilde{C}(\sigma)} \frac{2\lambda \dot{V}_{\beta}^h}{\hbar}\right)\\
&= {\rm{ev}}_0\left(\exp_{\tilde{C}(\sigma)} \left(\mathrm{i}\phi(f)+\frac{2\lambda \dot{V}_{\beta}^h}{\hbar}\right)\right)\\
&= {\rm{ev}}_0
\left(
e^{\sum_{i<j}\Gamma^{ij}_{\tilde{C}(\sigma)}} \exp_\otimes{ \left( \mathrm{i}\phi(f) +
\frac{2\lambda \dot{V}_{\beta}^h}{\hbar}\right)}\right)
\end{align*}
its derivative along a generic $\Gamma$ is
\begin{align*}
 \partial^\Gamma F(\sigma,\Lambda)  &=  \sum\limits_{\pi\in\mathcal{P}(\Gamma)}  
 {\rm{ev}}_0
\left(\prod_{\gamma \in \pi} \left( \sum_{i<j}\Gamma^{ij}_{\partial^\gamma \tilde{C}(\sigma)} \right)
e^{\sum_{i<j}\Gamma^{ij}_{\tilde{C}(\sigma)}} \exp_\otimes{ \left( \mathrm{i}\phi(f) +
\frac{2\lambda \dot{V}_{\beta}^h}{\hbar}\right)}\right)
\end{align*}
we can now localize the spatial variables of the points $x,y$ in the factors $\partial^\gamma \tilde{C}(\sigma)(x,y)$ present in $ \partial^\Gamma F(\sigma,\Lambda) $ 
 on spatial intervals of unit length. Namely on elements of
\[
\mathcal{I}:=\{[l,l+1]\subset \mathbb{R} \mid l\in\mathbb{Z}\}
\]
 and further decompose it as a sum over all possible intervals, 
 actually,
we are using
\[
\partial^\gamma C(s)_{ij} = \sum_{I,J \in \mathcal{I}} \partial^\gamma C(s)_{ij}^{I,J}
\]
where 
$\partial^\gamma C(s)_{ij}^{I_1,I_2}=\partial^\gamma C(s)_{ij} h_{I_1} h_{I_2}$ 
and $h_{I}$ is the characteristic function of the interval $I$. Hence $\partial^\gamma C(s)_{ij}^{I,J}$ is supported on $(0,\beta) \times I \times (0,\beta) \times J $.
With this decomposition
\begin{align}
 \partial^\Gamma F(\sigma,\Lambda)  &=  \sum\limits_{\pi\in\mathcal{P}(\Gamma)}  
 {\rm{ev}}_0
\left(\prod_{\gamma \in \pi} \left( \sum_{I,J}\sum_{i<j}\Gamma^{ij}_{\partial^\gamma \tilde{C}^{I,J}(\sigma)} \right)
e^{\sum_{i<j}\Gamma^{ij}_{\tilde{C}(\sigma)}} \exp_\otimes{ \left( \mathrm{i}\phi(f) +
\frac{2\lambda \dot{V}_{\beta}^h}{\hbar}\right)}\right)
\notag
\\
 &=  \sum\limits_{\pi\in\mathcal{P}(\Gamma)}  
 \sum_{\{(I_\gamma,J_\gamma) \}}
 {\rm{ev}}_0
\left(       \prod_{\gamma \in \pi} \left( \sum_{i<j}\Gamma^{ij}_{\partial^\gamma \tilde{C}^{I_\gamma,J_\gamma}(\sigma)} \right)
e^{\sum_{i<j}\Gamma^{ij}_{\tilde{C}(\sigma)}} \exp_\otimes{ \left( \mathrm{i}\phi(f) +
\frac{2\lambda \dot{V}_{\beta}^h}{\hbar}\right)}\right)
\notag
\\
&=  \sum\limits_{\pi\in\mathcal{P}(\Gamma)}  
\sum_{\{(I_\gamma,J_\gamma)\}}
 {\rm{ev}}_0 
\left( R_{\pi,\{(I_\gamma,J_\gamma)\}} \cdot_{\tilde{C}(\sigma)}
\exp_{\tilde{C}(\sigma)} \left(\mathrm{i}\phi(f)+\frac{2\lambda \dot{V}_{\beta}^h}{\hbar}\right)
\right)
\label{eq:decompositioninR}
\end{align}
where in the second equality we have interchanged the product over $\gamma\in\pi$ with the sum over copies of intervals.
Hence, after the interchange, the sum is over all possible $|\pi|$ couples of intervals and each element in a generic couple is labelled by $\gamma$. Namely the set over which this sum is taken is
\[
{\{(I_\gamma,J_\gamma) \}_{\gamma\in \pi}   \in (\mathcal{I}\times \mathcal{I})^{|\pi|}}.
\]
Hence, $ \partial^\Gamma F(\sigma,\Lambda)$ is a sum of terms indexed by all possible partitions and by the $|\pi|$ couples of localization intervals $\{(I_\gamma,J_\gamma)\}$. Furthermore $R_{\pi,\{(I_\gamma,J_\gamma)\}}$ is itself a sum of contributions of the form 
\begin{equation}\label{eq:L}
L(w')   =  \int_{( (0,\beta) \times \mathbb{R})^k}  \mathcal{W}_{\theta_1}(\mathbf{x}_1)\cdot_{\tilde{C}(\sigma)} \dots \cdot_{\tilde{C}(\sigma)} \mathcal{W}_{\theta_k}(\mathbf{x}_k)  w'(\mathbf{x}_1,\dots,\mathbf{x}_k)\prod_{i=1}^k d^i\mathbf{x}_i 
\end{equation}
where $\mathbf{x}_i=(u_i,x_i)$ and $\theta_i \in \{-a,0,a\}$ and
\[
\mathcal{W}_\theta(\mathbf{x}) =
\begin{cases}
 \phi(\mathbf{x}), \qquad  &\text{ if } \theta=0\\
 \frac{\lambda \theta \exp(i a\phi(\mathbf{x}))}{\hbar}, \qquad  &\text{ if } \theta=a > 0 \\
 \frac{\lambda \theta \exp(-i a\phi(\mathbf{x}))}{\hbar}, \qquad  &\text{ if } \theta=-a <  0 
\end{cases}
\]
 and
\begin{equation}\label{eq:w'}
w'(\mathbf{x}_1,\dots, \mathbf{x}_k) = w(\mathbf{x}_1,\dots, \mathbf{x}_k) \prod_{\gamma\in\pi} \partial^\gamma \tilde{C}^{I_\gamma,J_\gamma}(\sigma)(\mathbf{x}_{\gamma},\mathbf{y}_{\gamma}) 
\end{equation}
$w$ being a tensor product of $h$ or $f$ and $\mathbf{x}_\gamma\in (\mathbf{x}_1,\dots, \mathbf{x}_k)$ and $\mathbf{y}_\gamma\in(\mathbf{x}_1,\dots, \mathbf{x}_k)$.
With the constraint that $\mathbf{x}_\gamma\neq \mathbf{y}_\gamma$ and  for every $\mathbf{x}\in(\mathbf{x}_1,\dots, \mathbf{x}_k) $ there exists at least one $\gamma\in\pi$ such that either $\mathbf{x}=\mathbf{x}_\gamma$ or $\mathbf{x}=\mathbf{y}_\gamma$. Furthermore $k\leq 2|\Gamma|$ and there are at most $3^{2|\Gamma|}|\pi|!$ contributions of the form $L$ in $R_{\pi,\{(I_\gamma,J_\gamma)\}}$. These contributions,  
arise by application of $\prod_{\gamma \in \pi} \left( \sum_{i<j}\Gamma^{ij}_{\partial^\gamma \tilde{C}^{I_\gamma,J_\gamma}(\sigma)} \right)$ to
$\exp_\otimes{ \left( \mathrm{i}\phi(f) +
\frac{2\lambda \dot{V}_{\beta}^h}{\hbar}\right)}$
 and thus by using the Leibniz rule.

We study one of this contribution to $\partial^\Gamma F(\sigma,\Lambda)$, namely we analyze
\[
K = {\rm{ev}}_0 \left( L\cdot_{\tilde{C}(\sigma)} \exp_{\tilde{C}(\sigma)} \left(-\mathrm{i}\phi(f)+\frac{2\lambda \dot{V}_{\beta}^h}{\hbar}\right)  \right)
\]  
We can use Cauchy-Schwarz inequality to bound the resulting term
\begin{align}\label{eq:LexpU}
K
& \leq  
\left( {\rm{ev}}_0 
\left( L^* \cdot_{\tilde{C}(\sigma)} L\right)\right)^{\frac{1}{2}}
\left( {\rm{ev}}_0  \left(\exp_{\tilde{C}(\sigma)} \left(\frac{4\lambda \dot{V}_{\beta}^h}{\hbar}\right)
\right)
\right)^{\frac{1}{2}}.
\end{align}
We have that the second factor does not depend on $\pi$ nor on the particular form of $L$ and it can be controlled using results of Lemma \ref{le:omega(U)}
\begin{equation}\label{eq:expU}
{\rm{ev}}_0  \left(\exp_{\tilde{C}(\sigma)} \left(\frac{4\lambda \dot{V}_{\beta}^h}{\hbar}\right)
\right)
\leq 
c e^{c|\Lambda|}.
\end{equation}
We are thus left with the remaining contribution
\[
K_1= 
{\rm{ev}}_0 
\left( L^* \cdot_{\tilde{C}(\sigma)} L\right).
\]
Notice that we can use conditioning and inverse conditioning to control ${\rm{ev}}_0 
\left( L\right).$
Actually, recalling \eqref{eq:expU} and by the Cauchy-Schwarz inequality as in \eqref{eq:LexpU}, 
we get a corollary of Lemma \ref{le:primo}, namely we have that 
\[
K_1\leq  c \|w'\|_{2p} \prod_{I} (N(I)!)^q c^{N(I)}  e^{\frac{c}{2}|\Lambda|}.
\]
Furthermore, if $w'$ is of the form \eqref{eq:w'} by an application of H\"older inequality we have that 
\[
\|w'\|_{2p} 
\leq
\|w\|_{p_0} \left\| \prod_{\gamma\in\pi} \partial^\gamma \tilde{C}^{I_\gamma,J_\gamma}(\sigma)\right\|_r 
\leq  \beta^{k/{p_0}} e^{\frac{k}{p_0}\|f\|_\infty }    \left\| \prod_{\gamma\in\pi} \partial^\gamma \tilde{C}^{I_\gamma,J_\gamma}(\sigma)\right\|_r 
\] 
where $1/r+1/p_0 = 1/(2p)$ and where we used the fact that $w$ is a product of various $h$ and $f$ hence $\|w\|_\infty \leq (1+\|f\|_\infty)^k\leq \exp{k\|f\|_\infty }$. To estimate $\| \prod_{\gamma\in\pi} \partial^\gamma \tilde{C}^{I_\gamma,J_\gamma}(\sigma)\|_r $ we use 
the decay properties of $\partial^\gamma C(s)$ established in 
Proposition \ref{pr:boundDGammaC} and Proposition \ref{pr:boundDGammaC1point} for distant $I_\gamma$ and $J_\gamma$. 
The detailed estimate needed in the proof can be found in Lemma \ref{le:terzo}.
We furthermore need an estimate of $M=M(\pi,\{(I_\gamma,J_\gamma)\} )$ the number of contributions of the form $L$ in $R_{\pi,\{(I_\gamma,J_\gamma)\}}$. This is given in Lemma \ref{le:secondo}.

With the Lemmata given in section \ref{se:lemmata}, we can now conclude the proof. Actually, from Lemma \ref{le:omega(U)}, Lemma \ref{le:primo} and by Cauchy-Schwarz inequality, we have for a suitable constant 
\begin{align*}
|\partial^\Gamma F(\sigma,\Lambda)|
\leq \sum\limits_{\pi\in\mathcal{P}(\Gamma)}  \sum_{\{(I_\gamma,J_\gamma) \}}
M(\pi,\{(I_\gamma,J_\gamma)\} )  \|w'\|_{2p} 
\prod_I (N(I)!)^{q} e^{c |\Lambda| }.
\end{align*}
Where $N(I)$ is the the number of $I_{\gamma}$ or $J_{\gamma}$ equal to $I$ in $\{(I_\gamma,J_\gamma) \}$.
Recalling the form $w'$, we obtain
\begin{align*}
|\partial^\Gamma F(\sigma,\Lambda)|
\leq \sum\limits_{\pi\in\mathcal{P}(\Gamma)}  \sum_{\{(I_\gamma,J_\gamma) \}}
M(\pi,\{(I_\gamma,J_\gamma)\} )\left\| \prod_{\gamma\in\pi} \partial^\gamma \tilde{C}^{I_\gamma,J_\gamma}(\sigma)\right\|_r \|w\|_{p_0} 
\prod_I (N(I)!)^{q} e^{c |\Lambda| }.
\end{align*}
Hence, using Lemma \ref{le:terzo}, we have
\begin{align*}
|\partial^\Gamma F(\sigma,\Lambda)|
\leq \sum\limits_{\pi\in\mathcal{P}(\Gamma)}  \sum_{\{(I_\gamma,J_\gamma) \}}
 \|w\|_{p_0} 
\prod_I (N(I)!)^{q} e^{c |\Lambda| }
e^{-c_2|\Gamma|} \prod_{\gamma\in\pi} c e^{-\frac{K(m)}{c} d(\gamma,I_\gamma,J_\gamma) }M(\pi,\{(I_\gamma,J_\gamma)\} )
\end{align*}
hence we can decouple the difficult sums
\begin{align*}
|\partial^\Gamma F(\sigma,\Lambda)|
\leq & 
 e^{\frac{2|\Gamma|}{p_0}(\|f\|_\infty+\log \beta) }  e^{c |\Lambda| }
e^{-c_2|\Gamma|} 
\left( \sum\limits_{\pi\in\mathcal{P}(\Gamma)}  
e^{-\frac{m}{2\sqrt{2}} \Delta(\gamma) }
\right)\\
&\cdot
\sup_{\pi\in\mathcal{P}(\Gamma)}
\left( 
\sum_{\{(I_\gamma,J_\gamma) \}}
 \prod_I (N(I)!)^{q}  \prod_{\gamma\in\pi} c e^{-\frac{K(m)}{2c} d(\gamma,I_\gamma,J_\gamma) }M(\pi,\{(I_\gamma,J_\gamma)\} )
\right)
\end{align*}
where we have used the fact that $w$ is a tensor product of $f$ and $h$.
The sum over $\mathcal{P}(\Gamma)$ is controlled by the results of Lemma \ref{le:combinatoric}. We have
\begin{align*}
|\partial^\Gamma F(\sigma,\Lambda)|
\leq & 
 e^{c |\Lambda| }
e^{-c_2|\Gamma|} 
c_1 e^{c_2' |\Gamma|}
\\
&\cdot
\sum_{\{(I_\gamma,J_\gamma) \}}
\sup_{\pi\in\mathcal{P}(\Gamma)}
\left( 
\prod_I (N(I)!)^{q}  \prod_{\gamma\in\pi} c e^{-\frac{K(m)}{2c} d(\gamma,I_\gamma,J_\gamma) }M(\pi,\{(I_\gamma,J_\gamma)\} )
\right)
\end{align*}
with Lemma \ref{le:secondo} we can bound $M(\pi,\{(I_\gamma,J_\gamma)\}$ and
 have that 
\begin{align*}
|\partial^\Gamma F(\sigma,\Lambda)|
\leq & 
 \|w\|_{p_0} e^{c |\Lambda| }
e^{-c_2|\Gamma|} 
c_1 e^{c_2' |\Gamma|}
e^{c_1 |\Gamma|} 
\\
&\cdot
\sum_{\{(I_\gamma,J_\gamma) \}}
\sup_{\pi\in\mathcal{P}(\Gamma)}
\left( 
\prod_I (N(I)!)^{q+1}  \prod_{\gamma\in\pi} c e^{-\frac{K(m)}{2c} d(\gamma,I_\gamma,J_\gamma) }
\right)
\end{align*}
Finally Lemma \ref{le:quarto} can be used to control the sum over $\{(I_\gamma,J_\gamma)\}$. It gives  
 that
\begin{align*}
|\partial^\Gamma F(\sigma,\Lambda)|
\leq & 
 e^{c |\Lambda| }
e^{-c_2|\Gamma|} 
c_1 e^{c_2' |\Gamma|}
e^{c_1 |\Gamma|} 
\end{align*}
$c_2$ is the only constant which depends on the mass and furthermore it diverges in the limit $m\to \infty$.
We  have thus concluded the proof of Lemma \ref{le:cluster-convergence-2}.
\end{proof}

\subsection{Technical lemmata}\label{se:lemmata}

We collect in this subsection some more technical lemma used in particular in the proof of Lemma \ref{le:cluster-convergence-2}.

\begin{lemma}\label{le:primo} (Lemma 2.4 in \cite{FS})
Consider $w' \in L^{p}(\mathbb{R}^{2k})$ with $p>1$ a function which is supported in the cartesian product
$
\prod_{i=1}^k (0,\beta)\times I_i
$
where $I_i\in \mathcal{I}$ are intervals of unit length whose extreme points are in $\mathbb{Z}$. Let $L$ be as in \eqref{eq:L} and denote by $N(I)$ the number of elements $j$ of $\{1,\dots, k\}$ such that $I_j=I$. It holds that 
\[
{\rm{ev}}_0 (L) \leq \|w'\|_p \prod_{I} (N(I)!)^q c^{N(I)}
\] 
where $1/p+1/q =1$ and $c$ is a suitable constant.
\end{lemma}

\begin{proof}
We need to estimate
\[
X=\int_{(0,\beta)\times I_1}  d\mathbf{x}_1  
\dots
\int_{(0,\beta)\times I_k}  d\mathbf{x}_k
\;
|{\rm{ev}}_0 ( \mathcal{W}_{\theta_1}(\mathbf{x}_1)\cdot_{\tilde{C}(\sigma)} \dots \cdot_{\tilde{C}(\sigma)} \mathcal{W}_{\theta_k}(\mathbf{x}_k))||w'(\mathbf{x}_1,\dots, \mathbf{x}_k)|.
\]
We can find a finite series of elements of non overlapping support of the form $F^j=f_1^j\otimes \dots \otimes f_k^j$ with $0\leq f_i^j\in L^p((0,\beta)\times I_i)$, which converges to $|w'|$ in $L^p$. Hence $\|w'\|_p = \lim_{N\to\infty}\sum_{j=1}^N \|F^j\|_p$   and  $\|F^j\|_p=\prod_{i=1}^k \|f^j_i\|_p$ .  
We thus start considering 
\[
X^j=
|{\rm{ev}}_0 ( \mathcal{W}_{\theta_1}(f^j_1)\cdot_{\tilde{C}(\sigma)} \dots \cdot_{\tilde{C}(\sigma)} \mathcal{W}_{\theta_k}(f^j_k))|.
\]
Notice that $\{I_i\}_i$ is made of $l\leq k$ different disjoint open intervals $\{J_i\in\mathcal{I}\}$ which are labeled in such a way that $J_i< J_{i+1}$.
We use a technique similar to the one used to prove Lemma \ref{le:omega(U)} to factorize the expectation values of products of $\mathcal{W}_\theta$ in a product of expectation values of $\mathcal{W}_\theta$ supported in the same interval. 
First of all we use Cauchy-Schwarz inequality to factorize the expectation values of fields supported on even intervals $E=\cup_i J_{2i}$ from those supported on odd intervals $O\cup_i J_{2i+i}$. In this way we are sure that intervals are not touching. 
We now factorize the expectation values of different odd intervals using inverse conditioning and a covariance with boundary conditions in a collection of points, one for every even interval. By conditioning and inverse  conditioning, these single factors can be controlled with a covariance $a\log((x^2+u^2)/\mu^2)$ where $\mu$ is a suitable constant as in the proof of Lemma \ref{le:omega(U)}. We do the same with the even intervals. 
Finally, in each factor, we use again Cauchy-Schwarz inequality to separated the fields $\mathcal{W}_0=\phi$ from $\mathcal{W}_\theta$ with $\theta\in\{a,-a\}$.  The factor which contains only $\mathcal{W}_0$ can be directly estimated, with  $C^n n!$ where $n$ is the number of fields in the factor. The factor which contains only $\mathcal{W}_\theta$ can be treated using results similar to the one obtained in Lemma \ref{le:Cauchy-Determinant}. Collecting all these estimate and observing that by Stirling approximation 
 $\sqrt{(2n)!}
 \leq c^{n} n!$ for a constant $c=2e$ we have that
\[
X^j \leq  \prod_{I} \left( (N(I)!)^q c^{M(I)}  \right) \prod_{i=1}^k {\|f^j_i\|_p}
\] 
where $1/p+1/q =1$ and $c$ is a suitable constant.
Summing over $j$ we get the assertion.
\end{proof}

\begin{lemma}\label{le:secondo} (replacement of Lemma 2.6 of \cite{FS} or Lemma 10.1 of \cite{GJS})
Consider $R_{\pi,\{(I_\gamma,J_\gamma)\}}$ of the form \eqref{eq:decompositioninR}. Let $M=M(\pi,\{(I_\gamma,J_\gamma)\}$ the number of contributions of the form $L$ given in \eqref{eq:L} in $R_{\pi,\{(I_\gamma,J_\gamma)\}}$.
Let $N(I)$ be the number of $I_{\gamma}$ or $J_{\gamma}$ in $\{(I_\gamma,J_\gamma)\}$ such that $I=I_{\gamma}$ or $I=J_\gamma$. It holds that 
\[
M \leq e^{c_1 |\Gamma|} \prod_I N(I)!
\]
for a suitable constant.
\end{lemma}

\begin{proof}
Applying $N(I)$ functional derivatives to $\exp \mathcal{W}$ we obtain at most
$N(I)!$ terms which contain products of functional derivatives of $\mathcal{W}$. 
Furthermore we recall that $\mathcal{W}$ is a linear combination of three local fields,  $\phi(f)$, $V_a$ and $V_{-a}$, hence, taking into account this decomposition in total there are at most $3^{N(I)} N(I)!$ contributions of the form $L$ given in \eqref{eq:L}.
We observe that in $R_{\pi \{(I_\gamma,J_\gamma)\}}$, 
for a given partition $\pi$ of $\Gamma$ in the product $X=\prod_{\gamma \in \pi} \Gamma_{\partial^\gamma \tilde{C}^{I_\gamma,J_\gamma}}$ there is a the product of $2|\pi|\leq 2|\Gamma|$ functional derivatives. 
Notice that we can combine the functional derivatives present in $X$ as a product over all possible intervals $I\in \mathcal{I}$ and for every interval $I$ we have $N(I)$ functional derivatives. 
Hence, we have that $\sum N(I) \leq 2|\Gamma|$, and we conclude that there exists a suitable constant $c_1$ such that
\[
M \leq \prod_I  3^{N(I)} N(I)! \leq  e^{c_1 |\Gamma|} \prod_I N(I)!.
\]
\end{proof}

\begin{lemma}\label{le:terzo} (replacement of Lemma 2.7 of \cite{FS}) Consider the product $\prod_{\gamma\in\pi} \partial^\gamma \tilde{C}^{I_\gamma,J_\gamma}(\sigma)$ present in $L$ in \eqref{eq:L} through $w'$ in \eqref{eq:w'}. It holds that
\[
\left\| \prod_{\gamma\in\pi} \partial^\gamma \tilde{C}^{I_\gamma,J_\gamma}(\sigma)\right\|_r  \leq 
e^{-c_2|\Gamma|} \prod_{\gamma\in\pi} c e^{-\frac{K(m)}{c} d(\gamma,I_\gamma,J_\gamma) }
\]
where $d(\gamma,I_\gamma,J_\gamma) =\Delta(\gamma) + \sup_{b\in\gamma} dist(b,I_\gamma)+ dist(b,J_\gamma)  $ and
$\Delta(\gamma)$ is  the minimal length of a path in $\mathbb{R}$ which starts and ends in $\gamma$ and which touches all the points of $\gamma$. $K(m)$ is a positive continuous function of $m$ which vanishes in the limit $m\to\infty$.
Furthermore $c_2$ tends to $\infty$ when $m\to \infty$ too. 
\end{lemma}

\begin{proof}
To prove this Lemma we need some further control on $C(s)$ and on $\partial^\Gamma C(s)$. 
 Let us start recalling the form of $C(s)$ given in \eqref{eq:covariance-convex-sum}
\[
C(s) = \sum_{\Gamma\subset\mathbb{Z}} 
\prod_{i\in \Gamma} s_i \prod_{j\in \Gamma^c} (1-s_j)  C^{\Gamma^c}.
\]
Proceeding as in \eqref{eq:Cb-path} and \eqref{eq:Cab-path} we have that 
\[
C^{\Gamma^c} =  
\int_0^{\infty} \int  e^{- \frac{m^2}{2} T} W(T,u) \prod_{p\in \Gamma^c}
J^T_p(\omega)d\mu_{x,y}^T(\omega)  dT   
\]
and thus substituting in \eqref{eq:covariance-convex-sum}, considering the combinatoric of a product of binomials we get
\begin{align*}
C(s) &= \sum_{\Gamma\subset\mathbb{Z}} 
\prod_{i\in \Gamma} s_i \prod_{j\in \Gamma^c} (1-s_j)  \int_0^{\infty} \int  e^{- \frac{m^2}{2} T} W(T,u) \prod_{p\in \Gamma^c}
J^T_p(\omega)d\mu_{x,y}^T(\omega)  dT 
\\
&= 
\int_0^{\infty} \int  e^{- \frac{m^2}{2} T} W(T,u) \prod_{b\in \mathbb{Z}}
(s_b+(1-s_b)J^T_b(\omega))d\mu_{x,y}^T(\omega)  dT.
\end{align*}
Now we can take the derivatives with respect to any $s_i$ to obtain
\begin{align*}
\frac{\partial}{\partial s_i} C(s) =
\int_0^{\infty} \int  e^{- \frac{m^2}{2} T} W(T,u) 
(1-J^T_i(\omega))
\prod_{b\in \mathbb{Z}, b\neq i}
(s_b+(1-s_b)J^T_b(\omega))d\mu_{x,y}^T(\omega)  dT
\end{align*}
furthermore
\begin{align*}
\frac{\partial^2}{\partial s_i^2} C(s)  = 0 
\end{align*}
and
\begin{align}\label{eq:defPartialGammaC}
\partial^\Gamma  C(s) = 
\int_0^{\infty} \int  e^{- \frac{m^2}{2} T} W(T,u) 
\prod_{i\in\Gamma}(1-J^T_i(\omega))
\prod_{b\in \Gamma^c}
(s_b+(1-s_b)J^T_b(\omega))d\mu_{x,y}^T(\omega)  dT.
\end{align}
The factors $\prod_{i\in\Gamma}(1-J^T_i(\omega))$ forces the Wiener integral to be taken over the paths that touches all the bonds in $\Gamma$.

In view of the positivity of $W(u,T)$ discussed after \eqref{eq:path-int-rep},
we thus have that $\partial^\Ga C(x,y)$ is controlled by $D_\Gamma C$ given \eqref{eq:DGammaC} and, if $\Gamma$ has more than one point
the latter is bounded thanks to Proposition \ref{pr:boundDGammaC}
\begin{equation}\label{eq:estimatepartialGamma}
0\leq \partial^\Ga C(x,y) 
\leq
c \frac{1}{m} e^{-\frac{m}{\sqrt{2}} \Delta(\Gamma)} e^{-\frac{m}{\sqrt{2}} d(x,\Gamma)} e^{-\frac{m}{\sqrt{2}} d(y,\Gamma)}, \qquad |\Ga|>1
\end{equation}
where here $\Delta(\Gamma)$ is the maximum distance of any two points in $\Gamma$. 
If $\Gamma$ is formed by a single element, we may use Proposition \ref{pr:boundDGammaC1point} to get that 
\begin{equation}\label{eq:estimatepartialGamma1pt}
0\leq \partial^\Ga C(x,y) 
\leq
c \frac{1}{m} e^{-\frac{m}{\sqrt{2}} d(x,\Gamma)} e^{-\frac{m}{\sqrt{2}} d(y,\Gamma)} , \qquad \Ga=\{b\}, |x-b|+|b-y|>\alpha
\end{equation}
Furthermore, recalling Lemma \ref{le:lpnormC} we can derive the following bounds   
\[
\|\partial^\Ga C(x,y)\|_p \leq \frac{c}{m^{\frac{1}{p}}}
\]
which is then used when $\Ga$ contains a single element $b$ and when $|x-b|+|b-y|<\alpha$ to get that also in this case the limit $m\to\infty$ vanishes. 
The proof is a direct consequence of these estimates, with the observation that there exists an $m_0$ such that for every $m>m_0$ $1/m^p\leq e^{-cm_0}$.
\end{proof}

\begin{lemma}\label{le:quarto} (replacement of Lemma 10.2 of \cite{GJS}) 
Let $\pi \in \mathcal{P}(\Gamma)$ and $r>0$, there exists a constant $c$ such that
\[
\sum_{\{(I_\gamma,J_\gamma)\}} \prod_{\gamma\in \pi} e^{- \frac{K(m)}{c} d(\gamma,I_\gamma,J_\gamma) }  \prod_{I} (N(I)!)^r \leq e^{c |\Gamma|}.
\]
Where $N(I)$ be the number of $I_{\gamma}$ or $J_{\gamma}$ in $\{(I_\gamma,J_\gamma)\}$ such that $I=I_{\gamma}$ or $I=J_\gamma$. and where $K(m)$ is the positive continuous function of $m$ which vanishes in the limit $m\to\infty$ obtained in Lemma \ref{le:terzo}.
\end{lemma}

\begin{proof}
The proof follows in a similar way as the proof of Lemma 10.2 in \cite{GJS} with some differences due to the one dimensional nature of the clustering we are considering here.

We shall prove that there exists a suitable constant, independent on $m$ such that
\[
\prod_{I} (N(I)!)^r  \leq  e^{c|\Gamma|}\prod_{\gamma\in\pi } e^{ c d(\gamma,I_\gamma,J_\gamma) } 
\]
with this estimate, if $m$ is large enough, the sum over various $I$ can be taken thanks to the weights $e^{- \frac{K(m)}{c} d(\gamma,I_\gamma,J_\gamma) }$.
We observe that for a fixed $I$ we can control $N(I)$ with the distance of $\gamma$ from $I$ namely 
\[
N(I) \leq C s , \qquad s = \sup_{\{\gamma| I\in \{I_\gamma,J_\gamma\}\} } d(\gamma,I_\gamma,J_\gamma).
\]
Furthermore,   since different $\gamma$s in $\pi$ are non overlapping,  the number of $\gamma$ in $\{\gamma\in\pi| I\in \{I_\gamma,J_\gamma\} ,  d(\gamma,I_\gamma,J_\gamma) < r \}$, is controlled by $r$, hence
\[
\sum_{\{\gamma| I\in \{I_\gamma,J_\gamma\}\}}  d(\gamma,I_\gamma,J_\gamma)        \geq  C N(I)^2
\]
hence 
\begin{align*}
\prod_I (N(I)!)^r&\leq \prod_I N(I)^{r N(I)} \leq e^{r \sum_I  N(I)^2} 
\\ &\leq e^{ c \sum_I \sum_{\{\gamma| I\in \{I_\gamma,J_\gamma\}\}} d(\gamma,I_\gamma,J_\gamma)} e^{c|\Gamma|}
\\ &\leq e^{c|\Gamma|}\prod_{\gamma\in\pi } e^{ c d(\gamma,I_\gamma,J_\gamma)}  .
\end{align*}
\end{proof}

We finally prove a combinatoric lemma which is a simplified version of Proposition 8.2 in \cite{GJS}. 
\begin{lemma}\label{le:combinatoric}
Let $\Gamma\subset \mathbb{Z}$, if $m$ is large enough, there exist two positive constants $c_1$ and $c_2$ such that
\[
\sum\limits_{\pi\in\mathcal{P}(\Gamma)}  
 \prod_{\gamma \in \pi}  e^{-\frac{m}{2\sqrt{2}} \Delta(\gamma)} \leq c_1 e^{c_2 |\Gamma|}
\]
where the sum is taken over all possible partitions $\pi$ of $\Gamma$. Furthermore $\Delta(\gamma)=\sup_{x,y\in\gamma}|x-y|$.  
\end{lemma}
\begin{proof}
Fix $r\in \mathbb{N}$ and let $A=\{a_i\in\mathbb{N}\cup 0\}_i$ be a finite sequence of natural number such that $\sum_i a_i = r$. 
The number of elements of $A$ is bounded by $|\Gamma|$. 
Let us estimate $n(A)$, which is the number of partitions $\pi \in \mathcal{P}(\Gamma)$ such that $\Delta(\gamma_i) = a_i$ where $\{\gamma_i\}_i=\pi$ and $\gamma_i\neq \gamma_j$.
Furthermore, $\gamma_i$ are labelled in such a way that the smallest element in $\gamma_i$ is smaller than the smallest element in $\gamma_{i+1}$.
Let $b_1$ be the smallest element in $\gamma_1$, notice that it coincides with the smallest element of $\Gamma$. 
We observe that since $\Delta(\gamma_1) = a_1$ every $b'\in \Gamma$ such that $b'>b_1+a_1$ cannot be in $\gamma_1$, furthermore every $b''\in (b_1,b_1+a_1)\cap \Gamma$ may or may not be in $\gamma_1$. Hence there are at most $2^{a_1-1}$ possible choices of $\gamma_1$.
Let $b_2$ be the smallest element of $\gamma_2$, notice that $b_2$ must be either in $I= (b_1,b_1+a_1)\cap \Gamma$ or it must be the smallest element of the complement of $\overline{I}= (b_1,b_1+a_1]$, namely in $\Gamma\setminus \overline{I}$.  Hence there are at most $a_1+1$ possibilities to chose $b_2$. If we continue in this way we get that the number of possible partitions $n(A)$ is bounded by
\begin{equation}\label{eq:n(A)}
n(A) \leq 2^{\sum_i a_i-1} \prod_i (1+a_i) \leq  e^{(1+\log(2))\sum_i a_i }= e^{(1+\log(2)) r}.
\end{equation}
Let $n_{|\Gamma|, r}$ be the number of possible sequences  $\{a_i\in\mathbb{N}\}_i$ such that $\sum_{i}a_i=r$ with $a_i\geq 0$ and with at most $|\Gamma|$ elements. 
We have that 
\begin{equation}\label{eq:n_r}
n_{|\Gamma|, r} \leq 2^{|\Gamma|}2^{r-1}
\end{equation}
Actually, we observe preliminarily that the number of possible sequences $\{c_i \in \mathbb{N}\setminus 0\}_i$ such that  $\sum_i c_i = r$ is bounded by $2^{r-1}$, actually, the first unit in $r$ is for sure in $c_1$, suppose that the $j-$th unit in $r$ is in $c_i$ for some $i$ then the $j+1-$th unit in $r$ may be either in $c_i$ or in $c_{i+1}$, there are only $2$ possibilities, hence the number of possible $\{c_i\}_i$ is $2^{r-1}$.  
Now, let $A=\{a_i\in \mathbb{N}\}_i$ be a finite sequence of positive natural numbers. We denote by $c(A)$ the subsequence obtained extracting the non vanishing elements of $A$. 
Fix now $C=\{c_i \in \mathbb{N}\setminus 0\}_i$ be a subsequence of strictly positive natural number such that $\sum_i {c}_i = r$. 
We observe that the number of possible $A$s with at most $|\Gamma|$ elements is bounded by $2^{|\Gamma|}$ because the $i-$th element of $A$ may be either $0$ or in $(0,\infty)$ and there are at most $|\Gamma|$ elements in $A$. 
Combining these observations we get the bound in \eqref{eq:n_r}.

With, \eqref{eq:n(A)} and \eqref{eq:n_r} we can now estimate
\[
\sum\limits_{\pi\in\mathcal{P}(\Gamma)}  
 \prod_{\gamma \in \pi}  e^{-\frac{m}{2\sqrt{2}} \Delta(\gamma)} \leq 
\sum_{r=0}^\infty 2^{|\Gamma|}2^{r-1}  e^{(1+\log(2)) r}  e^{-\frac{m}{2\sqrt{2}} r} 
\]
notice that the sum at the right hand side of the previous inequality can be taken if $m$ is large enough. In this way we get the assertion.
\end{proof}

\begin{lemma}\label{le:omega(U)}
Let $h$ be the characteristic function of the connected closed set $\Lambda = [-l,k]\subset \mathbb{R}$ with $l,k\in\mathbb{N}$. 
There exists a suitable positive constant $c$ such that
\begin{equation}\label{eq:stability}
|\omega^{\beta,h}(U(i\beta))| = |Z_{C^\beta_m}(\Lambda)|  \leq e^{c|\Lambda|}
\end{equation}
where $|\Lambda|$ denotes the volume of the set $\Lambda$ which here coincides with $k+l$.
\end{lemma}
\begin{proof} 
Recalling \eqref{eq:Fexp-Zexp} we have that
\[
 \begin{aligned}
\omega^\beta(U(i\beta))
&=
\sum_n  \frac{1}{2^n} \frac{\lambda^n}{\hbar^n}\int_{\beta  \mathcal{S}_n}  dU  
\int dX h^{\otimes n}(X)
\sum_{a_i = \pm a}
e^{-\sum_{1\leq i<j\leq n} a_ia_j\hbar C^\beta_m(u_i-u_j,x_i-x_j)} 
\\
&=
\sum_n\frac{1}{n!} \frac{\lambda^n}{\hbar^n} \int_{\frac{\beta}{2}  \mathcal{S}_n}  dU  
\int dX h^{\otimes n}(X)
\sum_{a_i = \pm a}
e^{-\sum_{1\leq i<j\leq n} a_ia_j\hbar {C}^\beta_m(u_i-u_j,x_i-x_j)} 
\end{aligned}
\]
where we used the symmetry property of ${C}_m^\beta$ (KMS condition) to reduce the domain over which the $dU$ integration is taken.
Furthermore, introducing the covariance $\tilde{C}_m^\beta(u,x)$ which is $\tilde{C}_m^\beta(u,x) = {C}_m^\beta(u,x)$ if $u>0$ and  $\tilde{C}_m^\beta(u,x) = {C}_m^\beta(-u,x)$ for negative $u$, we have that
\[
 \begin{aligned}
\omega^\beta(U(i\beta))
&=
\sum_n  \frac{1}{n!} \frac{\lambda^n}{\hbar^n} \int_{  (0,\frac{\beta}{2})^n}  dU  
\int dX h^{\otimes n}(X)
\sum_{a_i = \pm a}
e^{-\sum_{1\leq i<j\leq n} a_ia_j\hbar \tilde{C}^\beta_m(u_i-u_j,x_i-x_j)} 
\end{aligned}
\]
Hence, we may observed that
\[
\omega^\beta(U(i\beta)) = {\rm{ev}}_0(\exp_{\tilde{C}_m^\beta} \frac{2\lambda \dot{V}_{\beta}^h}{\hbar})
\]
where the exponential is computed with respect to the product obtained with $\hbar \tilde{C}_m^\beta$ and where
\[
\dot{V}_{\beta}^h = \int_0^{\frac{\beta}{2}} du \int  d\mathbf{x} \cos(a \phi(u,\mathbf{x})) h(\mathbf{x}).
\]
We now decompose $\dot{V}_{\beta}^h = \dot{V}_{\beta}^{h_1}+\dot{V}_{\beta}^{h_2}$ in contributions with disjoint supports, 
which are such that 
\begin{align*}
\supp{\dot{V}_{\beta}^{h_1}} \subset D_e, \qquad
\supp{\dot{V}_{\beta}^{h_2}} \subset D_e+(0,1), \qquad
\end{align*}
where $D_e=\{ (u,x)\in (0,\beta)\times \mathbb{R} \mid  x \mod 2 \in [0,1/2] \cup [-1,-1/2]\}$.
Notice that $\dot{V}_{\beta}^{h_2}$ coincides to $\dot{V}_{\beta}^{h_1}$ up to a spacelike translation. 
We now observe that in view of translation invariance of $\tilde{C}_m^\beta$, 
we can proceed as in the prof of Theorem \ref{thm:convergence-V}.
Actually, $\dot{V}_{\beta}^{h_i}$ can be obtained as (the limit of) a sum of a positive functional plus a constant functional similarly to \eqref{eq:decomposition-of-V}.
\[
\dot{V}_{\beta}^{h_i} = \lim_{n\to\infty} P_n^{h_i} + c_n^{h_i}. 
\]
Hence, using the fact that $c_n^{h}$ is a constant functional, we get
\[
{\rm{ev}}_0\left(\exp_{\tilde{C}_m^\beta} \frac{2\lambda \dot{V}_{\beta}^h}{\hbar}\right)
= 
  \lim_{n\to\infty}
  e^{\frac{2\lambda}{\hbar}c^{h}_n}
  {\rm{ev}}_0
  \left( \exp_{\tilde{C}_m^\beta}
   \frac{2\lambda}{\hbar}P^{h}_n 
  \right).
\]
Observe that the $N-$th order in $\lambda $contribution to $  {\rm{ev}}_0
  \left( \exp_{\tilde{C}_m^\beta}
   \frac{2\lambda}{\hbar}P^{h}_n 
  \right)$ can be bounded by an application of Proposition \ref{prop:multi-holder}.
 We get  
\[
{\rm{ev}}_0(\underbrace{P_n^{h} \cdot_{\tilde{C}_m^\beta} \dots \cdot_{\tilde{C}_m^\beta} P_n^{h}}_{N}) 
=
{\rm{ev}}_0(\underbrace{\sum_{i=1}^2 P_n^{h_i}  \cdot_{\tilde{C}_m^\beta} \dots \cdot_{\tilde{C}_m^\beta} \sum_{i=1}^2 P_n^{h_i} }_{N}) 
\leq {{2}}^N {\rm{ev}}_0(\underbrace{P_n^{h_1} \cdot_{\tilde{C}_m^\beta} \dots \cdot_{\tilde{C}_m^\beta} P_n^{h_1}}_{N}).
\]
Taking the sum of the power series in $\lambda$ and the limit $n\to\infty$ we get
\begin{align*}
|\omega^\beta(U(i\beta))| 
= {\rm{ev}}_0
  \left( \exp_{\tilde{C}_m^\beta}
   \frac{2\lambda}{\hbar}P^{h}_n 
  \right)
&\leq {\rm{ev}}_0\left( \exp_{{\tilde{C}_m^\beta}}\left(\frac{4 \lambda}{\hbar}\dot{V}_{\beta}^{h_1}\right)\right).
\end{align*}
Modify now the product $\cdot_{\tilde{C}_m^\beta}$ imposing boundary conditions on $\mathcal{B}=\{2\mathbb{Z}-\frac{1}{2}\}$ and denote by 
 $\overline{C}_m^\beta$ the corresponding covariance.
Since $\mathcal{B}$ does not intersect $D_e$, the points over which $\overline{C}_m^\beta-\tilde{C}_m^\beta$ diverges are outside of the support of $\dot{V}_{\beta}^{h_1}$. 
Notice that also that the covariance $\overline{C}(u,x)$ with boundary conditions on $\mathcal{B}$ is smaller than $\tilde{C}_m^\beta$ as can be seen from its definition in terms of the sum over paths similarly to what it is done in \eqref{eq:Cab-path}.
Hence, if we apply inverse conditioning, namely by Proposition \ref{pr:inverse-conditioning}, we have that 
\[
|\omega^\beta(U(i\beta))| \leq 2 {\rm{ev}}_0\left(\exp_{\overline{C}_m^\beta}  \left( 8e^{a^2 \frac{K}{2}} \frac{\lambda \dot{V}_{\beta}^{h_1}}{\hbar} \right)\right).
\]
however, we have also that $\overline{C}_m^\beta(u,x_i,x_j)$ vanishes if $|x_i-x_j|>1$, if we denote by $\mathcal{I}$ the set of connected components of $D_e$, we have that the expectation values of the product of two elements supported on two different elements of $\mathcal{I}$ factorizes, hence 
\begin{equation}\label{eq:sss}
|\omega^\beta(U(i\beta))| \leq 2 \prod_{\{X \in \mathcal{I} | X \subset \Lambda \}}
{\rm{ev}}_0\left(\exp_{\overline{C}_m^\beta}  \left( 8e^{a^2 \frac{K}{2}} \frac{\lambda \dot{V}_{\beta}^{h_X}}{\hbar} \right)\right)
\end{equation}
where $h_X$ is the characteristic function of the set $X$. We observe that, in view of the periodicity of $\overline{C}_m^\beta$ ($\overline{C}_m^\beta(u,x_1,x_2)=\overline{C}_m^\beta(u,x_1+2n,x_2+2n)$, $n\in\mathbb{Z}$, 
${\rm{ev}}_0(\exp_{\overline{C}_m^\beta} \frac{\lambda \dot{V}_{\beta}^{h_X}}{\hbar})$ does not depend on $X$.
Observing that there are $|\Lambda|/2$ elements $\{X \in \mathcal{I} | X \subset \Lambda \}$ and thus there are $|\Lambda|/2$ equal factors in the product at the right hand side of \eqref{eq:sss} the claim holds if ${\rm{ev}}_0(\exp_{\overline{C}_m^\beta} \frac{\lambda\dot{V}_{\beta}^{h_X}}{\hbar})$ is finite for every $\lambda$. 
To prove that ${\rm{ev}}_0(\exp_{\overline{C}_m^\beta} \frac{\lambda \dot{V}_{\beta}^{h_X}}{\hbar})$ is finite we use conditioning and inverse condition as in the proof of Theorem \ref{thm:convergence} three times. Actually,  we control ${\rm{ev}}_0(\exp_{\overline{C}_m^\beta} \frac{\lambda \dot{V}_{\beta}^{h_X}}{\hbar})$ with  ${\rm{ev}}_0(\exp_{{C}_m^\beta} \frac{\lambda \dot{V}_{\beta}^{h_X}}{\hbar})$ by means of conditioning because ${C}_m^\beta-\overline{C}_m^\beta$ is positive (as can be seen using the path integral representation). Afterwords, we control 
${\rm{ev}}_0(\exp_{{C}_m^\beta} \frac{\lambda \dot{V}_{\beta}^{h_X}}{\hbar})$ with 
${\rm{ev}}_0(\exp_{{C}_m^\infty} \frac{\lambda \dot{V}_{\beta}^{h_X}}{\hbar})$ by means of inverse conditioning because 
${C}_m^\infty-{C}_m^\beta$ is negative.
Finally, ${\rm{ev}}_0(\exp_{{C}_m^\infty} \frac{\lambda \dot{V}_{\beta}^{h_X}}{\hbar})$ can be controlled with
${\rm{ev}}_0(\exp_{{C}_{0,\mu}} \frac{\lambda \dot{V}_{\beta}^{h_X}}{\hbar})$ for a sufficiently large $\mu$ where ${C}_{0,\mu}(u,x) = A \log(\frac{u^2+x^2}{\mu^2})$ by means of inverse conditioning because
${C}_m^\infty-{C}_{0,\mu}$ is positive.
The obtained contribution ${\rm{ev}}_0(\exp_{{C}_{0,\mu}} \frac{\lambda \dot{V}_{\beta}^{h_X}}{\hbar})$ can be shown to be bounded as in  Theorem \ref{thm:convergence} or as in the similar analysis performed in the Euclidean case, see e.g. \cite{Froe76}.
\end{proof}

\section*{Acknowledgments}

The authors would like to thank D. Cadamuro,  K. Fredenhagen and J. Fr\"ohhlich for illuminating discussions on the subject of the present paper.
The authors are also grateful to the referee for pointing out some weaknesses in an earlier version of the  manuscript and for his help in improving the presentation of the obtained results.
The research presented in this paper was possible and was conducted in part thank to the hospitality of various institutions.
In particular, DB and KR would like to thank Centre de recherches mathematiques in Montreal for the hospitality during the Thematic Semester ``Mathematical Challenges in Many-Body Physics and Quantum Information (fall 2018)''.
NP and KR would like to thank Mainz Institute for Theoretical Physics for the hospitality during the program ``The Mysterious Universe: Dark Matter - Dark Energy - Cosmic Magnetic Fields'' (summer 2019).
All three authors would like to thank Perimeter Institute for the hospitality during the Emmy Noether Workshop: ``The Structure of Quantum Space Time'' (fall 2019).

\appendix 
\section{Jensen inequality in the functional approach}

The result of this section are valid in the commutative $*-$algebra of regular functionals with finitely many non vanishing functional derivatives equipped with the pointwise product or a deformation of the pointwise product. 
Under certain hypotheses, as we shall see in the last proposition, Jensen inequality holds also in the case of local functionals.

\begin{prop}\label{pr:jensen} (Jensen inequality)
Consider $G\in \mathcal{A}=(\mathcal{F}_{\normalfont\text{reg}},\cdot)$ a selfadjoint functional over a smooth field configuration $\varphi\in C^\infty(M;\mathbb{R})$.
Let $f:\mathbb{R} \to \mathbb{R}$ a smooth convex function.  Let $\omega$ a state (a positive normalized linear functional over $\mathcal{F}_{\normalfont\text{reg}}$).  
It holds that
\[
\omega(f(G)) \geq f (\omega(G)), 
\]
where $f(G)(\varphi) := f(G(\varphi))$. 
\end{prop}
\begin{proof}
Since $f$ is a smooth convex function, for every copy of selfadjoint elements $F,G \in \mathcal{F}_{\normalfont\text{reg}}$ and for every field configuration $\varphi$ it holds that 
\[
f(G(\varphi))-f(F(\varphi)) \geq f'(F(\varphi)) (G(\varphi)-F(\varphi)),
\]
this means that $f(G)-f(F) - f'(F) \cdot (G-F)$ is a positive element of $\mathcal{F}$.
Hence choosing $F$ to be a constant functional equal to $\omega(G)$, we have
\[
f(G)-f(\omega(G)) \geq f'(\omega(G)) (G-\omega(G)),
\]
the positivity of the state implies that
\[
\omega(f(G)-f(\omega(G))) \geq \omega(f'(\omega(G)) (G-\omega(G))),
\]
the state is linear and normalized hence
\[
\omega(f(G)) - f(\omega(G))) \geq 0.
\]
\end{proof}
The same proof holds also for more singular functionals and in the case of deformed algebras.
To see this in details we shall here consider the polynomial subalgebra $\mathcal{A}^p\subset \mathcal{A}=(\mathcal{F}_{\normalfont\text{reg}},\cdot)$ which is the smallest subalgebra of $\mathcal{A}$ which contains the linear fields.
Furthermore, let us consider $\mathcal{A}^p_w$ which is the algebra containing the same elements as 
$\mathcal{A}^p$ and whose product is $\cdot_w$, a deformation of the pointwise product realized in the following way
\begin{equation}\label{eq:product-exponential_2}
F\cdot_w G = \mathcal{M}\; e^{\Gamma_w}(F\otimes G), \qquad   \Gamma_w = \int d\mu_xd\mu_y w(x,y) \frac{\delta}{\delta \varphi(x)} \otimes \frac{\delta}{\delta \varphi(y)}
\end{equation}
where $w\in \mathcal{D}'(M^2)$  is a symmetric, real, positive distribution (for every $f,g\in\mathcal{D}(M)$$w(f,g)=w(g,f)=\overline{w(\overline{f},\overline{g})}$, $w(\overline{h},h)\geq 0$) of order $\hbar$. 
As usual, if $w$ is sufficiently regular, the product $\cdot_w$ can be extended to more singular objects like local observables.

Notice that $\mathcal{A}^p_w$ and $\mathcal{A}^p=$ are $*-$isomorphic. The isomorphism 
$\alpha_w: \mathcal{A}^p\to \mathcal{A}^p_w$ is realized by
\begin{equation}\label{eq:alpha}
\alpha_w (F)(\varphi) = e^{\frac{1}{2}\int d\mu_xd\mu_y w(x,y) \frac{\delta^2}{\delta \varphi(x)\delta\varphi(y)}}F(\varphi),  
\end{equation}
and
\[
F \cdot_w G = \alpha_w ( \alpha^{-1}_w(F) \cdot \alpha^{-1}_w(G)).
\]

The set of microcausal functional $\mathcal{F}_{\mu c}$ is linear, we equip it with the topology induced by the following seminorms labelled by $\varphi$, a field configuration, $l\in\mathbb{N}$ and $f\in {\mathcal{E}(M^l)}$, 
\[
\|G\|_{\varphi,l,f} = |G^{(l)}(\varphi,f) |, \qquad G\in \mathcal{F}_{\mu c}.
\] 
Hence, we say that $G_n \in \mathcal{F}_{\mu c}$ converges to $G\in \mathcal{F}_{\mu c}$ if for every $l$ and for every $\varphi\in\mathcal{C}$,  $G_n^{(l)}(\varphi)-G^{(l)}(\varphi)$ in the topology of $\mathcal{E}'(M^l)$ and, as usual,
we say that $D_n$ converges to $D$ in $\mathcal{E}'(M^l)$ if for every $f\in \mathcal{E}(M^l)$ $\lim_{n\to\infty} D_n(f)=D(f)$. With these definitions we have that the product $\cdot_w$ on $\mathcal{A}^p_w$ is continuous with respect to the topology of  $\mathcal{F}_{\mu c}$.

We observe that, every smooth function $f:\mathbb{C}\to \mathbb{C}$ can be promoted to a function on $f:\mathcal{A}\to\mathcal{A}$ $f(F)(\varphi):=f(F(\varphi))$, furthermore at the same time, if $f$ is a polynomial function we may obtain a function $f_w:\mathcal{A}_w \to \mathcal{A}_w$ applying $\alpha_w$ 
\[
f_w := \alpha_w \circ f \circ \alpha_w^{-1}.
\]
If $f$ is convex, Jensen inequality holds also in the deformed algebra. 
\begin{prop}(Jensen inequality (deformed algebra)) \label{pr:jensen2}
Let $f:\mathbb{R}\to\mathbb{R}$ be a real convex function, let 
$w\in\mathcal{D}'(M^2)$ be symmetric, real and positive.
 The following holds
\begin{itemize}
\item[(1)] If $f$ is a polynomial function, for every state $\omega$ on $\mathcal{A}^p_w$, and for every $G=G^*\in\mathcal{A}^p_{w}$
the following inequality holds
\begin{equation}\label{eq:Jensen-def}
\omega(f_w(G)) \geq f (\omega(G)), 
\end{equation}
\item[(2)]   Suppose that $G=G^*$ 
is a local functional which can be obtained as the limit of $G_n\in\mathcal{A}^p_w \subset \mathcal{F}_{\mu c}$. Suppose that $w$ is regular enough to have $G^{\cdot_w n}$ well defined for every $n\in\mathbb{N}$.
Let 
$\mathcal{A}^p_{G,w}$ the polynomial algebra generated by $G$ together with the identity.
If $\omega$ is a state for $\mathcal{A}^p_{G,w}$ which is obtained as the 
 extension of a state $\omega$ (continuous with respect to the topology of $\mathcal{F}_{\mu c}$) on $\mathcal{A}^p_w$  and if $f$ is a polynomial function \eqref{eq:Jensen-def} holds.

\item[(3)] In the hypothesis of point (2), suppose that the function $f$ can be obtained as the pointwise limit of a sequence  $\{f_n\}_{n\in\mathbb{N}}$ of convex polynomial functions.
Suppose that $\lim_{n\to\infty}\omega(f_{n,w}(G))=L\in \mathbb{R}$ then \eqref{eq:Jensen-def} holds
with $\omega(f_w(G))=L$. 
\end{itemize}
\end{prop}

\begin{proof}
In the hypothesis of item $(1)$ we have that $\omega\circ\alpha_w$ is a state on $\mathcal{A}^p$ 
and that $\alpha_w^{-1}(G)$ 
is selfadjoint in $\mathcal{A}^p$. By Proposition 
\ref{pr:jensen}, we have 
\[
\omega(f_w(G))=\omega(\alpha_w( f(\alpha_w^{-1}(G)))) \geq  f (\omega(G))
\]
thus concluding the proof of item $(1)$. 

To prove item $(2)$ we have by hypothesis that  there exists a sequence of elements $G_n\in\mathcal{A}^p_w$ which converges to $G$ in the topology of $\mathcal{F}_{\mu c}$. 
Furthermore, 
we have that $\omega(G_n)$ and $\omega(f_w(G_n))$ converge respectively to $\omega(G)$ and $\omega(f_w(G))$ for large $n$ in view of the continuity of the state $\omega$ and because $f$ is a polynomial functions. 
For every $n$ it holds that by the first part of this proposition 
\[
\omega(f_w(G_n)) \geq  f (\omega(G_n))
\]
hence taking the limit on both side, also in view of the continuity of $f$, we get the statement of item (2)

To prove item (3) notice that for every $f_n$ 
equation \eqref{eq:Jensen-def} holds 
\[
\omega(f_{n,w}(G)) \geq  f_n (\omega(G)).
\]
The claim follows by taking the limit $n$ to infinity on both sides of the inequality.
\end{proof}

The previous Proposition will be used in the following way. 
\begin{rem}\label{rem:how-use-it}
Consider $V_{g,a}= \int e^{\mathrm{i} a  \varphi(x)}  g(x) d\mu_x$ with sufficiently small $a$, $\varphi\in\mathcal{C}$ and $g\in{C}^\infty_0(M)$, we have that it generates a $*-$algebra $\mathcal{A}_{V_{g,a},w}$ where the product is $\cdot_{w}$ with $w(x,y)=-\frac{1}{4\pi} \log(|(x-y)^2|) + s(x,y)$ and where $s$ is a continuous function on $M^2$.
The same holds for $V_{g,-a}$ and for $V=(V_{g,a}+V_{g,-a})/2$.
Consider now $V_{g,a}^N = \sum_{n=0}^N \int \frac{(\mathrm{i}a\varphi_{x,N})^n}{n!} g(x)  d\mu_x$ where 
$\varphi_{x,N} = \int \varphi(y) \delta_{1/N}(y-x) d\mu_y$ and where 
$\delta_{\epsilon}$ is a set of compactly supported smooth functions which converge to the Dirac delta function for $\epsilon\to 0$. $V_{g,a}^N$ are polynomial local functionals in $\mathcal{A}^p_w$ which converges to $V_{g,a}$ in the topology of $\mathcal{F}_{\mu c}$.  
The evaluation functional ${\rm{ev}}_0(G)=G(0)$ is a state on $\mathcal{A}_{V,w}$ which is continuous with respect to the topology of $\mathcal{F}_{\mu c}$.
Consider the polynomial functions $f_n(x) = \sum_{l=0}^{2n}  \frac{x^{l}}{l!}$. Uniformly on every compact set, the function $f_n$ converges to $\exp$ for large $n$, furthermore, for $n>0$ $f_n$ are real positive and convex. If ${\rm{ev}}_0(f_{n,w}(V_{g,a}))$ is a convergent sequence, considering $V=(V_{g,a} + V_{g,-a})/2$
we are in the hypothesis of item (3) of Proposition \ref{pr:jensen2} and thus
\[
{\rm{ev}}_0( \exp_w(V) ) \geq \exp ({\rm{ev}}_0(V) )
\] 
where ${\rm{ev}}_0( \exp_w(V) )$ is understood as the limit for $n\to\infty $ of ${\rm{ev}}_0(f_{n,w}(V_{g,a}))$.
\end{rem}

\subsection{Jensen inequality \texorpdfstring{$C^*-$}{C*-}algebraic case} \label{ap:jensen}
Jensen inequality holds also on a (non-commutative) $C^*$-algebra.
\begin{prop}
Let $\mathfrak{A}$ be a $C^*-$algebra, $f$ a differentiable real, convex function and $\omega$ a state on $\mathfrak{A}$, then for every $A=A^*\in\mathfrak{A}$
\[
\omega(f(A))\geq f(\omega(A))
\]
where $f(A)\in\mathfrak{A}$ is obtained with the standard functional calculus for $C^*-$algebras. 
\end{prop}
\begin{proof}
Notice that $I$ and $A$ generate a $C^*-$subalgebra $\mathcal{B}$ of $\mathcal{A}$, furthermore $\mathcal{B}$ is commutative.
The element $f(A)$ is in $\mathcal{B}$. $\omega$ restricted on $\mathcal{B}$ is there a state. 
Jensen inequality can be proved as in the proof of Proposition \ref{pr:jensen}, alternatively we obtain the standard Jensen inequality noticing that by Gelfand duality $\mathcal{B}$ is isomorphic to $C(\sigma(A))$ the $C^*-$algebra of bounded functions on $\sigma(A)$ which is the spectrum of $A$. 
By the Riesz-Markov theorem a state $\omega$ on $C(\sigma(A))$ is represented by a probability measure $\mu_\omega$, such that
\[
\omega(f(A)) = \int_{\sigma(A)} f d\mu_\omega.
\]
Applying the known Jensen inequality for probability measures we get the assertion.
\end{proof}

\section{Conditioning}

The following result and its proof is the translation to the functional formalism of the Theorem 2.1 (1) in \cite{Froe76}. See also \cite{GR75a,GR75b}.
\begin{prop}\label{pr:conditioning}
Let $w, v\in\mathcal{D}'(M\times M)$ be two symmetric positive distributions such that $w-v$ is also positive.
Consider $f_n(x) = \sum_{l=0}^{2n}  \frac{x^{l}}{l!}$, and 
$V=V^*\in \mathcal{F}_{\normalfont\text{reg}}\cap \mathcal{A}^p$. 
Suppose that the sequences
${\rm{ev}}_0( f_{n,w}(V) )=f_{n,w}(V)(0)$
and
${\rm{ev}}_0( f_{n,v}(V) )=f_{n,v}(V)(0)$
converge for large $n$. If we denote respectively by ${\rm{ev}}_0( \exp_w(V) )$
and by 
${\rm{ev}}_0( \exp_v(V) )$
the resulting limits
%
it holds that 
\begin{equation}\label{eq:Jensen-V}
{\rm{ev}}_0( \exp_w(V) ) \geq {\rm{ev}}_0( \exp_v(V) ).
\end{equation}

If for every $n$, 
${\rm{ev}}_0(V^{\cdot_v n})$ is positive, the convergence of ${\rm{ev}}_0( f_{n,v}(V) )$
to 
${\rm{ev}}_0( \exp_v(V))$
follows from the convergence of
${\rm{ev}}_0( f_{n,w}(V) )$.

If $V=V^*\in \mathcal{F}_{\normalfont\text{loc}}\cap \mathcal{F}^V$
can be obtained as the limit of  polynomial regular functionals,
if ${\rm{ev}}_0( f_{n,w}(V) )$ is convergent,
  and ${\rm{ev}}_0(V^{\cdot_v n} )$ is positive for every $n$,
 \eqref{eq:Jensen-V} holds also in this case.
\end{prop}
\begin{proof}
To every functional $F$ on $C^\infty(M;\mathbb{R})$ we can associate a functional $\tilde{F}$ on $\mathcal{C}^2=C^\infty(M;\mathbb{R})\times C^\infty(M;\mathbb{R})$ in the following way
\[
\tilde{F}(\varphi_1,\varphi_2)=F(\varphi_1+\varphi_2).
\]
Furthermore, using the state $\widetilde{{\rm{ev}}}_0(F)=F(0,0)$ on the functionals over $\mathcal{C}^2$, we have that 
\[
\widetilde{{\rm{ev}}}_0(\tilde{F}) = {{\rm{ev}}}_0({F}).
\]
Let us discuss the case where $V=V^*$ is a polynomial regular functional. 

Notice that
\begin{align*}
{\rm{ev}}_0(f_{n,w}(V)) &= 
{\rm{ev}}_0(\alpha_w f_n(\alpha_w^{-1}V))=
{\rm{ev}}_0(\alpha_v\alpha_{w-v}f_n(\alpha_{w-v}^{-1}\alpha_{v}^{-1}V)) \\
&=
\widetilde{{\rm{ev}}}_0(\alpha^1_v\alpha^2_{w-v}f_n({\alpha^2}_{w-v}^{-1}{\alpha^1}_{v}^{-1}\tilde{V}))\\
&=
(\alpha^1_v\alpha^2_{w-v}f_n({\alpha^2}_{w-v}^{-1}{\alpha^1}_{v}^{-1}\tilde{V}))(0,0)
\end{align*}
where $\alpha^1_{v}$ and $\alpha^2_{w-v}$ are defined as in \eqref{eq:alpha} where the functional derivatives in $\alpha^i$ act only on $\varphi_i$. For this reason they commute.   
For every $n$, $f_n$ is positive, 
because in the point $x_m$ where 
where $f_n$ takes its minimum 
$f_{n}'(x_m)=0$ and $f_{n}(x_m)=f_{n}'(x_m)+\frac{x_m^{2n}}{(2n)!} = \frac{x_m^{2n}}{(2n)!} \geq 0$.
For every $n$, $f_n$ is also convex, actually $f''_n= f_{n-1}$ and $f_{n-1}\geq 0$.
Hence we may apply Jensen inequality.
For a fixed $\varphi_1$
Jensen inequality given in Proposition \ref{pr:jensen2} holds with respect of $\varphi_2$, and with respect to a state $\tilde\omega^2_0(\tilde{F}) = \tilde{F}(\varphi_1,0)$, namely  
\[
(\alpha^2_{w-v}f_n({\alpha^2}_{w-v}^{-1}{\alpha^1}_{v}^{-1}\tilde{V}))(\varphi_1,0)
\geq
f_n({\alpha^1}_{v}^{-1}{\tilde{V}})(\varphi_1,0)=f_n({\alpha}_{v}^{-1}{V})(\varphi_1)
\]
this implies that 
\[
F(\varphi):=(\alpha^2_{w-v}f_n({\alpha^2}_{w-v}^{-1}{\alpha^1}_{v}^{-1}\tilde{V}))(\varphi,0)
-f_n({\alpha}_{v}^{-1}{V})(\varphi)
\]
is a positive element of $\mathcal{F}_{\text{reg}}\cap \mathcal{A}^{p}$. Hence, since ${\rm{ev}}_0\circ\alpha_v$ is a positive state on $\mathcal{A}^{p}$, evaluating $F$ on ${\rm{ev}}_0\circ\alpha_v$ we get for every $n$
\begin{align*}
{\rm{ev}}_0(f_{n,w}(V)) \geq 
{\rm{ev}}_0(f_{n,v}(V)).
\end{align*}
The previous inequality holds for every $n$, 
both sequences converge, hence the previous inequality holds also in the limit $n\to\infty$.  
Hence we get
\begin{align*}
{\rm{ev}}_0(\exp_w(V)) \geq 
{\rm{ev}}_0(\exp_v(V)).
\end{align*}
If  ${\rm{ev}}_0(V^{\cdot_v n})$ is positive for every $n$, ${\rm{ev}}_0(f_{n,v}(V))$ is actually an absolutely convergent series because ${\rm{ev}}_0(f_{n,w}(V))$ converges.
This concludes the proof of the first part of the statement. 
To prove the second part of the statement we have now that $V=V^*$ is such that together with $\cdot_w$ generates a $*-$algebra, ${\rm{ev}}_0( f_{n,w}(V) )$ is convergent and
$V$ can be obtained as the limit $j\to\infty$ of  polynomial regular functionals $V_j=V_j^*$.
By the first part of the statement we have that, for every $j$ and every $n$
\begin{align*}
{\rm{ev}}_0(f_{n,w}(V_j)) \geq 
{\rm{ev}}_0(f_{n,v}(V_j)).
\end{align*}
Taking the limit $j\to\infty$ 
  we get that the previous inequality holds for every $n$ also in the limit 
\begin{align*}
{\rm{ev}}_0(f_{n,w}(V)) \geq 
{\rm{ev}}_0(f_{n,v}(V)).
\end{align*}
Since ${\rm{ev}}_0(f_{n,w}(V))$ is in $n$ a convergent sequence,
and   since ${\rm{ev}}_0(V^{\cdot_v n})$ is positive,
we have that the series defining ${\rm{ev}}_0(f_{n,v}(V))$ is also absolutely convergent and controlled by ${\rm{ev}}_0(\exp_{w}(V))$.
Hence, the previous inequality holds also in the limit $n\to\infty$ thus getting the assertion.
\end{proof}

In the present paper, we are interested in using the results of the previous Proposition for $V_g=\int \cos(a\varphi(x)) g(x)d\mu_x$, this can be done following a discussion similar to the one presented in 
Remark \ref{rem:how-use-it}.

\section{Inverse Conditioning}

This section is based on \cite{Froe76}. 
Consider $w\in\mathcal{D}'(M^2)$ 
which is symmetric and positive (for every $f,g\in\mathcal{D}(M)$,  $w(f \otimes g)=w(g  \otimes f)$ and $w(\overline{f}  \otimes f)\geq 0$).
We also require that 
the integral kernel of $w$ is described by a function on $M^2$ (which with a little abuse of notation we denote by $w$) which makes $e^{-a_1 a_2w(x,y)}$ locally integrable for $a_i\in(-\alpha,\alpha)$.
To any symmetric and positive two-point function $w_2:\mathcal{D}(M)^2\to \mathbb{C}$ analyzed in the reset of the paper we can associate an unique element $w\in \mathcal{D}'(M^2)$ such that $w(f\otimes g) = w_2(f,g)$. To better present the result of this appendix it is easier to work with two-point function seen as elements of $\mathcal{D}'(M^2)$. 
The distribution $w$ is then used to construct a symmetric product among regular functionals as in \eqref{eq:product-exponential} and suppose that $w$ is sufficiently regular to allow the extension of $\cdot_w$ to local functionals.
We shall thus work in $\mathcal{A}=(\mathcal{F},\cdot_w)$ where $\mathcal{A}$ is generated by local functionals. 
Consider the following functionals
\begin{equation}\label{eq:chia}
V^{g}_a(\varphi) := \int_M  e^{\mathrm{i} a \varphi(x)}  g(x)  d\mu_x
\end{equation}
which are elements of $\mathcal{A}$. 
We also use the following notation 
\[
V_\pm^g(\varphi) :=V_{\pm1}^g(\varphi).
\]
Notice that
\[
V_{a_1}^g\cdot_w V_{a_2}^g =\int_{M^2}  e^{\mathrm{i}a_1 \varphi(x)+\mathrm{i}a_2\varphi(y) } e^{-a_1a_2w(x,y)} g(x)g(y)d\mu_xd\mu_y.
\]
We shall consider $g\geq 0 $, hence we have that
\[
|V_a^g|^2 = V_a^g\cdot_w {V_a^g}^{*} = \int_{M^2}  e^{\mathrm{i}a \varphi(x)-\mathrm{i}a\varphi(y) } e^{a^2w(x,y)} g(x)g(y)d\mu_xd\mu_y
\] 
Consider the state ${\rm{ev}}_0$ defined by
\[
{\rm{ev}}_0(F) := F(0) , \qquad \forall F\in \mathcal{A}.
\]
We thus have that 
\[
{\rm{ev}}_0(|V_a^g|^2) = V_a^g\cdot_w {V_a^g}^{*}(0) = \int_{M^2} e^{a^2w(x,y)} g(x)g(y) d\mu_xd\mu_y.
\]
We shall now prove that ${\rm{ev}}_0(|V_a^g|^2)$ is larger than ${\rm{ev}}_0(V_a^g\cdot_w V_a^g)$.
To reach this result we need the following lemma which involves the set of symmetric distributions 
\[
\mathcal{D}'_s(M^2) = \{ w\in \mathcal{D}'(M^2) \mid w(f\otimes g) = w(g\otimes f), \forall f,g \in \mathcal{D}(M)  \}.
\]
 We call an element of $\mathcal{D}'_s(M^2)$ a symmetric {\bf two-point function}.
Notice that  $I(f\otimes g):= \int f d\mu \int g d\mu$ is an element of $\mathcal{D}'_s(M^2)$. Furthermore, $w\in \mathcal{D}'_s(M^2) $ is positive if $w(\overline{f}\otimes f)\geq 0$ for every $f\in\mathcal{D}$.
\begin{lemma}\label{eq:change-sign-quadratic}
Let $N$ be a positive normalized linear functional on the set of symmetric distributions $\mathcal{D}'_s(M^2)$, namely
\[
N(I)=1,  \qquad N(w)\geq 0 \qquad \text{if } w\geq0.
\]
Then, for every positive   $w$ in $\mathcal{D}'_s(M^2)$ such that $e^{w}(f\otimes g) = \langle e^{w} ,f\otimes g\rangle$ is also in $\mathcal{D}'_s(M^2)$, it holds that
\[
N(e^{w}) \geq N(e^{-w}).
\]
\end{lemma}
\begin{proof}
Consider $E(\lambda) =  N(e^{\lambda w}) $ for $-1\leq \lambda \leq 1$.
The exponential is a convex smooth function and $N$ is a positive normalized functional on symmetric distributions $\mathcal{D}'_s(M^2)$, hence by Jensen inequality  
\[
 N(e^{\lambda w}) \geq  e^{\lambda N( w)}.
\]
This inequality holds for every $|\lambda| \leq 1$ hence, since $w$ is positive, it implies that 
\[
N(w^n) \geq  N(w)^n \geq 0
\]
where ${w}^n(f\otimes g) = \langle {w}^n ,f\otimes g\rangle$ are well defined because $e^{\lambda w} \in \mathcal{D}'_s(M^2)$. 
Consider now $S(\lambda) = E(\lambda) - E(-\lambda)$ and
notice that for $\lambda >0$
\[
S(\lambda) = 2 \sum_{n\geq 0} \frac{\lambda^{2n+1}}{(2n+1)!} N(w^{2n+1})  \geq 
2 \sum_{n\geq 0} \frac{\lambda^{2n+1}}{(2n+1)!} N(w)^{2n+1} = 2\sinh(\lambda N(w)) \geq 0.
\]
Hence $S(1)\geq0$ and this implies the claim.
\end{proof}

The following proposition holds
\begin{prop}
Consider $g\geq 0$, and $w$ a real positive distribution which makes $\cdot_w$ well defined on local functionals. It holds that
\[
{\rm{ev}}_0(|V_a^g|^2) \geq {\rm{ev}}_0(V_a^g\cdot_w V_a^g)
\]
\end{prop}
\begin{proof}
Let us start observing that 
\[
E(\lambda) := \int_{M^2} e^{\lambda a^2w(x,y)} g(x)g(y) d\mu_xd\mu_y
\]
is positive for every $\lambda$ because the integrand is positive.
Furthermore $E(1)={\rm{ev}}_0(|V_a^g|^2)$ and $E(-1)={\rm{ev}}_0(V_a^g\cdot_w V_a^g)$, moreover
\[
E(\lambda) = \|g\|_1^2 N(e^{\lambda a^2w})
\]
where $N(A):=\int_{M^2} A(x,y) g_n(x)g_n(y) d\mu_xd\mu_y$ and where $g_n(x)=g(x)/\|g\|_1$.
Notice that $N$ is a positive normalized functional on symmetric distributions in $\mathcal{D}_s(M^2)$, hence Lemma \ref{eq:change-sign-quadratic} implies that   
\[
E(1) = \|g\|_1^2 N(e^{ a^2w}) \geq \|g\|_1^2 N(e^{-a^2w}) = E(-1)
\]
and thus we have proven the claim.
\end{proof}

Using the symmetry of $w$, the fact that $w$ is real and $g\geq0$ we have
\begin{equation}\label{eq:positiveprod}
{\rm{ev}}_0( {V_{a_1}^g\cdot_{w}} \dots \cdot_w V_{a_n}^g) = \int_{M^n} 
e^{-\sum_{i<j}a_ia_j w(x_i,x_j)} g(x_1)\dots g(x_n)   d\mu_{x_1}\dots d\mu_{x_n} \geq 0.
\end{equation}
Suppose now that $-a_{\hat{i}}a_{\hat{j}}>0$ for a fixed copy of indices $\hat{i},\hat{j}$. Then we may obtain a lower bound for ${\rm{ev}}_0( {V^{a_1}\cdot_{w}} \dots \cdot_w V^{a_n})$
 changing the sign of $-a_{\hat{i}}a_{\hat{j}} w(x_{\hat{i}},x_{\hat{j}})$ in the argument of the exponential of the right hand side of \eqref{eq:positiveprod}.
If $w$ is regular enough, we have actually that 
\[
\hat{N}(A) := \int_{M^n} A(x_{\hat{i}},x_{\hat{j}})  e^{-\sum_{i<j,{(i,j)\neq (\hat{i},\hat{j})}}a_ia_j w(x_i,x_j)} g(x_1)\dots g(x_n)   d\mu_{x_1}\dots d\mu_{x_n}
\]
gives origin to a positive normalized functional $N(A)=\hat{N}(A)/\hat{N}(1)$, hence Lemma \ref{eq:change-sign-quadratic} applied to this $N$ furnishes the desired bound.

Consider
\[
Z_n(w,g) := {\rm{ev}}_0((V^g_{+}\cdot_w V^g_{-})^{\cdot_w n} )
\]
where the product in the $n-$th power of
$V^g_{+}\cdot_w V^g_{-}$ is taken with respect to $\cdot_w$.  Notice that $Z_n$ can be understood as the canonical partition function of a system formed by $n$ positively and $n$ negatively charged particles, 
which interacts with the two-body potential $w$ each of which is confined by an external potential of the form $\log(g)$.   
Lemma 2.1 in \cite{Froe76} implies the following
\begin{prop}\label{pr:inverse-conditioning} (Inverse Conditioning)
Consider $w_1$ and $w_2$ symmetric positive two-point functions. 
If $w_1(f,f)\geq w_2(f,f)$ for any real $f$ such that $\int_M f d\mu = 0$, and if 
\[
\sup_{x} \{w_1(x,x)-w_2(x,x)\} \leq K
\]
it holds that 
\[
Z_n(w_1,g) \leq e^{ n a^2 K} Z_{n}(w_2,g).
\]
\end{prop}

\begin{proof}
The proof is exactly the same as the proof of Lemma 2.1 in \cite{Froe76}.
We need to compute
\[
Z_n(w_1,g) = \int dX^ndY^n g(X^n) g(Y^n) e^{-U_{w_1}(X^n,Y^n)}
\]
where $X^n=(x_1,\dots, x_n)$, $g(X^n):=\prod_{j=1}^n g(x_j)$ and $dX^n=dx_1\dots dx_n$. Furthermore, in terms of the two-point function $W$
\[
- U_W(X^n,Y^n) :=  - a^2 \sum_{1\leq i<j\leq n} \left(W(x_i,x_j)+W(y_i,y_j)\right)  + a^2 \sum_{i,j=1}^n W(x_i,y_j)
\]
Consider
\[
\Sigma(x) = \sum_{j=1}^n a (\delta (x-x_j)-\delta(x-y_j))
\]
Notice that 
\[
- U_W =  -\frac{1}{2} \int dx dy \Sigma(x) W(x,y) \Sigma(y) + \frac{a^2}{2} \sum_j \left(W(x_j,x_j)+W(y_j,y_j)\right)  
\]
hence, since $0\leq g \leq 1$
\begin{align*}
Z_n(w_1,g) &= \int dX^ndY^n g(X^n) g(Y^n) e^{-U_W(X^n,Y^n)} e^{-U_{w_1-w_2}(X^n,Y^n)}\\
&= \int dX^ndY^n g(X^n) g(Y^n) e^{-U_{w_2}(X^n,Y^n)}  e^{-\frac{1}{2} \int dx dy \Sigma(x) (w_1-w_2)(x,y) \Sigma(y)}  
e^{\frac{a^2}{2} \sum_j \left((w_1-w_2)(x_j,x_j)+(w_1-w_2)(y_j,y_j)\right) }
\\
&\leq 
\int dX^ndY^n g(X^n) g(Y^n) e^{-U_{w_2}(X^n,Y^n)}  e^{-\frac{1}{2} \int dx dy \Sigma(x) (w_1-w_2)(x,y) \Sigma(y)}  
e^{{a^2} n \;\text{supp}(w_1-w_2)(x,x) }
\\
&\leq 
Z(w_2,g) e^{n{a^2} K}
\end{align*}

\end{proof}

Consider
\begin{equation}\label{eq:Z2n2q}
Z_{2n}^{2q}(w,g) :=   {\rm{ev}}_0((V^g_{+})^{\cdot_w (n+q)}\cdot_w(V^g_{-})^{\cdot_w (n-q)})
\end{equation}
notice that $Z^0_{2n}(w,g)  = Z_n(w,g) $.
Equation \eqref{eq:positiveprod} implies that $Z_{2n}^{2q}(w,g) \geq 0$, hence, since
the algebraic product we are considering is commutative and the state $\omega_w$ is positive, we have  the following proposition

\begin{prop}
It holds that 
\begin{equation}\label{eq:Zinequality}
0 \leq   Z_{2n}^{2q}(w,g)  \leq  Z_n(w,g) 
\end{equation}
\end{prop}
\begin{proof}
By direct inspection, 
\[
Z_{2n}^{2q}(w,g)  = \int_{M^{2n}} 
\exp{ \left(-\!\!\!\sum_{i<j\leq n+q}\!\!\! w(x_i,x_j) -\!\!\!\sum_{i<j\leq n-q}\!\!\! w(y_i,y_j) 
+\!\!\!\!\!\sum_{i\leq n+q; j\leq n-q}\!\!\!\!\!\!\!w(x_i,y_j)  \right)}
\!\!\prod_{i\leq n+q}\!\!g(x_i)\!\!\prod_{j\leq n-q}\!\! g(y_{j}) d\mu_{x_i}d\mu_{y_j} 
\]
hence  $Z_{2n}^{2q}(w,g)\geq0$ because $g\geq0$.
Furthermore, 
$Z_n(w,g) =Z^{0}_{2n}(w,g) $
hence the number of $w$  with positive signs in the exponential in $Z_n$ is $n^2$ which is larger than in $Z_{2n}^{2q}(w,g) $ where it is $n^2-q^2$. Hence Lemma \ref{eq:change-sign-quadratic} applied $q^2$ times as in the discussion after \eqref{eq:positiveprod} implies the claim.
\end{proof}
 
Furthermore, we observe that 
\begin{equation}\label{eq:Vpotential}
V^g:=\cos(\varphi)(g):=\int_M \cos(\varphi(x)) g(x)d\mu_x = \frac{1}{2}\left( V^g_{+}+V^g_- \right)
\end{equation}
hence, expanding the products and comparing with \eqref{eq:Z2n2q}
\begin{equation}\label{eq:sum}
\frac{1}{(2n)!}{\rm{ev}}_0(\cos(\varphi)(g)^{\cdot_w 2 n})=\frac{1}{2^{2n}}\sum_{q=0}^{2n} \frac{1}{q! (2n-q)!}Z^{2q-2n}_{2n}
\end{equation}
hence, as also discussed in \cite{Froe76}, we have
\begin{equation}\label{eq:Zineq-n}
\frac{1}{2^{2n}}\frac{1}{(n!)^2}Z_n(w,g)  \leq  \frac{1}{(2n)!}{\rm{ev}}_0(\cos(\varphi)(g)^{\cdot_w 2 n}) 
\leq  \frac{1}{(2n!)}Z_n(w,g) 
< \frac{1}{(n!)^2}Z_n(w,g).
\end{equation}
where in the first inequality we used the fact that all the elements in the sum \eqref{eq:sum} are positive hence we considered only the element $q=n$ of the sum
while in the second inequality, we used \eqref{eq:Zinequality} and the binomial formula for $(1+1)^{2n}$.
Now, we define 
\[
\Xi(w,zg):= \sum_{n=0}^\infty \frac{z^n}{n!} {\rm{ev}}_0(\cos(\varphi)(g)^{\cdot_w n})  = {\rm{ev}}_0(\exp_w(zV)),
\]
and
\[
\Xi_{\text{cosh}}(w,zg):= \sum_{n=0}^\infty \frac{z^{2n}}{(2n)!} {\rm{ev}}_0(\cos(\varphi)(g)^{\cdot_w 2n}) = {\rm{ev}}_0(\cosh_w(zV^g))
\]
where $2\cosh_w(V^g) = \exp_w(V^g)+\exp_w(-V^g)$, as well as
\[
Z(w,zg) := \sum_{n=0}^\infty \frac{z^{2n}}{(n!)^2} Z_n(w,g),
\]
which correspond to the gran canonical partition function. 
By direct inspection and using \eqref{eq:Zineq-n}, we then have
\[
|\Xi(w,zg)| \leq  \Xi(w,|z|g) \leq 2 \Xi_{\text{cosh}}(w,zg) \leq 2 \sum_n \frac{z^{2n}}{(2n)! } Z_n(w,zg)
 \leq 2 \sum_n \frac{z^{2n}}{(n!)^2 } Z_n(w,g) = 2Z(w,zg)
\]
and again by \eqref{eq:Zineq-n}
\[
Z(w,zg) =\sum_n\frac{z^{2n}}{(n!)^2 } Z_n(w,g)
\leq \sum_n \frac{(2z)^{2n}}{(2n)! }\omega_w(\cos(\varphi)(g)^{\cdot_w 2 n}) 
=  \Xi_{\text{cosh}}(w,2z g) 
\leq 2 \Xi(w,|2z|g).
\]
Let us summarize this discussion in the following theorem
\begin{thm}\label{th:partitions} Let $g$ be a positive function, then it holds
\[
\Xi_{\normalfont\text{cosh}}(w,g) \leq Z(w,g) \leq \Xi_{\normalfont\text{cosh}}(w,2g)
\]
and
\[
\Xi(w,g) \leq 2 \Xi_{\normalfont\text{cosh}}(w,g) \leq 2\Xi(w,g).
\]
\end{thm}

Combining Proposition \ref{pr:inverse-conditioning} with Theorem \ref{th:partitions} we have that 
\begin{prop}\label{pr:inverse-conditioning2} (Inverse Conditioning)
Consider $g\geq0$ and let $w_1$ and $w_0$ be two positive, symmetric two-point functions, such that 
\[
w_1(f,f)\geq w_0(f,f)
\] 
for every real valued function $f$. Assume furthermore that 
\[
\sup_x \{w_1(x,x)-w_0(x,x) \}\leq K
\]
then
\[
\Xi(w_1,g) \leq 2 \Xi(w_0,2 e^{a^2 \frac{K}{2}}g)
\]
or using another notation
\[
\omega_{w_1}(\exp_{w_1}(V^g))\leq 2 \omega_{w_0}(\exp_{w_0}(2e^{a^2 \frac{K}{2}}V^g))
\]
\end{prop}

\section{Cauchy-Schwarz and H\"older inequalities}

\subsection{Cauchy-Schwarz inequality}

Notice that $\exp_{\hbar H_m}\left(\frac{\lambda}{\hbar}V\right)$ is the exponential computed in the commutative $*-$algebra
\[
\mathcal{A}_{H_m}=(\mathcal{F}^V,\cdot_{H_m}).
\]
Furthermore, the functional ${\rm{ev}}_0(F):=F(0)$ is a linear normalized positive functional on $\mathcal{A}_{H_m}$. 
In view of the positivity of ${\rm{ev}}_0$ we have in $\mathcal{A}_{w}$ that Cauchy-Schwarz inequality holds
\[
|{\rm{ev}}_0(A \cdot_w B)| \leq  \sqrt{|{\rm{ev}}_0(A^* \cdot_w A)|}\sqrt{|{\rm{ev}}_0(B^* \cdot_w B)|}
\]
hence, with $B=1$ we obtain
\begin{equation}\label{eq:C*inequality}
|{\rm{ev}}_0(A)|^2 \leq |{\rm{ev}}_0(A^* \cdot_w A)|
\end{equation}
and
\begin{align*}
|{\rm{ev}}_0(A+B)|^2 &\leq {\rm{ev}}_0((A^*+B^*) \cdot_w (A+B))
\\
&=
{\rm{ev}}_0(A^*\cdot_w A)+{\rm{ev}}_0(A^*\cdot_w B)
+{\rm{ev}}_0(B^*\cdot_w A)
+
{\rm{ev}}_0(B^*\cdot_w B)
\\
&\leq
(\sqrt{{\rm{ev}}_0(A^*\cdot_w A)}+\sqrt{{\rm{ev}}_0(B^*\cdot_w B)})^2
\end{align*}
if now $B=A_x$ is obtained translating $A$ and in the direction $x$ and if $w$ is invariant under spacetime translations, we have that 
$|{\rm{ev}}_0(A)|^2=|{\rm{ev}}_0(A_x)|^2$ and $|{\rm{ev}}_0(A\cdot_w B)|^2=|{\rm{ev}}_0(A_x\cdot_w B_x)|^2$ hence
\begin{align*}
|{\rm{ev}}_0(A_x+A_y)|^2 
&
\leq
4|{\rm{ev}}_0(A^*\cdot_w A)|
\end{align*}
Iterating this procedure we have that 
\begin{align*}
|{\rm{ev}}_0(\sum_{j=1}^n A_{x_j})|^2 
&
\leq
n^2|{\rm{ev}}_0(A^*\cdot_w A)|
\end{align*}

\subsection{H\"older inequality}
 
Let us start considering $\mathcal{A}_e$ which is a commutative $*-$algebra of (not necessarily smooth) functionals over smooth field configurations  constructed with the pointwise product and with the involution defined by the complex conjugation. Consider also a generic state $\tilde\omega$ on $\mathcal{A}_e$. 
We say that an element of $A
\in\mathcal{A}_e$ is positive if  
$A=C C^*$ for some $C\in \mathcal{A}_e$. 
If $A$ and $B$ are positive elements of $\mathcal{A}$
H\"older inequality holds, namely
\[
\tilde\omega(A B)\leq \tilde\omega(A^p)^{\frac{1}{p}}\tilde\omega(B^q)^{\frac{1}{q}}
\]
for $1/p+1/q=1$. 

This inequality can be proven for the algebra of functionals proceeding in the following way:
Consider $\tilde{A} :=  A^p / \tilde\omega(A^p)$
and $\tilde{B} :=  B^p / \tilde\omega(B^p)$ where $A^p(\varphi)= A(\varphi)^p$. 
It holds that $\tilde\omega(\tilde{A}) =  1 =  \tilde\omega(\tilde{B})$.
The functionals $\tilde{A}$ and $\tilde{B}$ are positive, 
hence for a generic $\varphi$ we can write $\tilde{A}(\varphi)=e^{x}$ and 
$\tilde{B}(\varphi)=e^{y}$. 
Since $e^{x}$ is a convex function we have that 
\[
e^{\frac{x}{p}+ \frac{y}{q}} \leq \frac{1}{p} e^{x} +\frac{1}{q}  e^y  ,
\qquad 
\text{namely}
\qquad
\tilde{A}^{\frac{1}{p}}   \tilde{B}^{\frac{1}{q}} \leq
\frac{1}{p} \tilde{A} + \frac{1}{q} \tilde{B}.
\]
$\tilde{\omega}$ is a positive 
functional, hence the previous inequality holds evaluating both sides on $\tilde{\omega}$. 
We thus obtain
\[
\tilde\omega(\tilde{A}^{\frac{1}{p}} \tilde{B}^{\frac{1}{q}}) \leq
\frac{1}{p}  + \frac{1}{q}, 
\]
which is the desired H\"older inequality. 
For $n,l\in\mathbb{N}$ with $l<n$, H\"older inequality implies that 
for $A_i$ positive elements of $\mathcal{A}$
\[
\tilde\omega(\prod_{i=1}^l A_i^{\frac{n}{l}})^{\frac{l}{n}} \leq 
\tilde\omega(A_l^{n})^{\frac{1}{n}} 
\tilde\omega(\prod_{i=1}^{l-1} A_i^{\frac{n}{l-1}})^{\frac{l-1}{n}}
\] 
Using recursively this relation we have that
\begin{equation}\label{eq:n-holder-commutative}
\tilde\omega(\prod_{i=1}^nA_i)\leq \prod_{i=1}^n \tilde\omega(A_i^n)^{\frac{1}{n}}.
\end{equation}
Contrary to $p$ and $q$ above $n$ in \eqref{eq:n-holder-commutative} is an integer, hence the very same inequality \eqref{eq:n-holder-commutative} holds also for $A_i$ positive elements in $\mathcal{A}\subset \mathcal{A}_e$ the subalgebra of $\mathcal{A}_e$ formed by smooth functionals.  
Furthermore, a similar inequality holds also in the deformed algebra $\mathcal{A}_w$ where $w$ is a symmetric, positive  distribution ($w(f,\overline{f})\geq 0$ for every $f\in C^\infty_0(M)$) whose integral Kernel is a smooth function.
We recall here that an element  $A\in\mathcal{A}_w$ is positive if $A =C\cdot_w C^*$ for some $C\in\mathcal{A}_w$. The following proposition holds.
\begin{prop}\label{prop:multi-holder}
Let $\omega$ a state on $\mathcal{A}_w$
where $w$ is a positive symmetric bidistribution whose integral kernel is described by a smooth function.
Consider $A_i$, $i\in \{1,\dots, n\}$ positive elements of $\mathcal{A}_w$.
Then it holds that
\begin{equation}\label{eq:multi-holder}
\omega(A_1 \cdot_w\dots \cdot_wA_n) \leq \prod_{i=1}^n \omega (A_i \cdot_w\dots \cdot_w A_i)^\frac{1}{n}.
\end{equation}
Inequality \eqref{eq:multi-holder} holds also in  
$(\mathcal{F}^V,\cdot_{\hbar H_m})$, namely, consider now $A_i=C_i\cdot_{\hbar H_m} C_i^*$ positive elements of $(\mathcal{F}^V,\cdot_{\hbar H_m})$
then
\begin{equation}\label{eq:multi-holder-ev}
{\rm{ev}}_0(A_1 \cdot_{\hbar H_m}\dots \cdot_{\hbar H_m}A_n) \leq \prod_{i=1}^n {\rm{ev}}_0 (A_i \cdot_{\hbar H_m}\dots \cdot_{\hbar H_m} A_i)^\frac{1}{n}.
\end{equation}
\end{prop}
\begin{proof}
Let us start considering the isomorphism of algebras $\alpha_w:\mathcal{A}\to\mathcal{A}_w$ given in \eqref{eq:alpha} and observe that 
\[
\omega(A_1 \cdot_w \dots \cdot_w A_n)
=
\omega(\alpha_w(
\alpha_w^{-1}(A_1) \dots \alpha_w^{-1}(A_n))),
\]
Notice that that 
for every state $\omega$ on $\mathcal{A}_w$, $\tilde{\omega} := \omega \circ \alpha_w$ is a positive normalized linear functional hence a state on $\mathcal{A}$, furthermore, 
 $B_i:=\alpha_w^{-1}(A_i)$ are positive elements of $\mathcal{A}$, hence we may use \eqref{eq:n-holder-commutative} to obtain
 \[
\omega(A_1 \cdot_w\dots \cdot_wA_n) \leq \prod_{i=1}^n \tilde\omega (B_i^{\cdot_w n})^\frac{1}{n}.
 \]
To conclude the proof of the first part of the proposition we use the properties of $\alpha_w$ to observe that the previous inequality is \eqref{eq:multi-holder}.

To prove the validity of \eqref{eq:multi-holder-ev}, consider a sequence of smooth functions $w_j$ which converges to $\hbar H_m$ in $\mathcal{D}'(M^2)$ for large $j$. 
For every $j$ we have that 
$(\mathcal{F}^V,\cdot_{w_j})$ is a subalgebra of $\mathcal{A}_{w_j}$
hence, applying 
\eqref{eq:multi-holder}
we have that for every $j$
\[
{\rm{ev}}_0((C_1\cdot_{w_j}C_1^*)
\cdot_{w_j}
\dots
\cdot_{w_j}
(C_n\cdot_{w_j}C_n^*))
\leq \prod_{i=i}^n
{\rm{ev}}_0(
(C_i\cdot_{w_j}C_i^*)^{\cdot_{w_j}n})^{\frac{1}{n}}
\]
and the inequality holds also in the limit of large $j$.
However, in the limit $j\to\infty$ the left hand side converges ${\rm{ev}}_0(A_1
\cdot_{\hbar H_m}\dots \cdot_{\hbar H_m})$ while the right hand side converges to $\prod_{i=1}^n {\rm{ev}}_0(A_i^{\cdot_{\hbar H_m}n})^{1/n}$.
Hence we get \eqref{eq:multi-holder-ev}. 
\end{proof}


%


\bibliographystyle{amsalpha}
\bibliography{References}

\providecommand{\bysame}{\leavevmode\hbox to3em{\hrulefill}\thinspace}
\providecommand{\MR}{\relax\ifhmode\unskip\space\fi MR }
\providecommand{\MRhref}[2]{%
  \href{http://www.ams.org/mathscinet-getitem?mr=#1}{#2}
}
\providecommand{\href}[2]{#2}
\begin{thebibliography}{GRS75b}

\bibitem[AR95]{AR95}
A.~Abdesselam and V.~Rivasseau, \emph{Trees, forests and jungles: A botanical
  garden for cluster expansions}, Constructive Physics Results in Field Theory,
  Statistical Mechanics and Condensed Matter Physics (Berlin, Heidelberg)
  (V.~Rivasseau, ed.), Springer Berlin Heidelberg, 1995, pp.~7--36.

\bibitem[Ara73]{Araki73}
H.~Araki, \emph{Relative {H}amiltonian for faithful normal states of a von
  {N}eumann algebra}, Publications of the Research Institute for Mathematical
  Sciences \textbf{9} (1973), no.~1, 165--209.

\bibitem[BDF09]{BDF}
R.~Brunetti, M.~D{\"u}tsch, and K.~Fredenhagen, \emph{Perturbative algebraic
  quantum field theory and the renormalization groups}, Adv. Theor. Math. Phys.
  \textbf{13} (2009), no.~5, 1541--1599.

\bibitem[BF84]{BF2}
G.A. Battle and P.~Federbush, \emph{A note on cluster expansions, tree graph
  identities, extra 1/n! factors!!!}, Letters in Mathematical Physics
  \textbf{8} (1984), 55–57.

\bibitem[BF00]{BF0}
R.~Brunetti and K.~Fredenhagen, \emph{Microlocal analysis and interacting
  quantum field theories}, Commun. Math. Phys. \textbf{208} (2000), no.~3,
  623--661.

\bibitem[BFK96]{BFK96}
R.~Brunetti, K.~Fredenhagen, and M.~K{\"o}hler, \emph{The microlocal spectrum
  condition and {W}ick polynomials of free fields on curved spacetimes},
  Commun. Math. Phys. \textbf{180} (1996), no.~3, 633--652.

\bibitem[BFM08]{BFM1}
G.~Benfatto, P.~Falco, and V.~Mastropietro, \emph{Massless sine-gordon and
  massive thirring models: Proof of coleman’s equivalence}, Communications in
  Mathematical Physics \textbf{285} (2008), no.~2, 713–762.

\bibitem[BFR21]{BFR17}
D.~Bahns, K.~Fredenhagen, and K.~Rejzner, \emph{Local nets of von {N}eumann
  algebras in the {S}ine-{G}ordon model}, Commun. Math. Phys. \textbf{383}
  (2021), 1--33, [arXiv:math-ph/1712.02844].

\bibitem[BFV03]{BFV}
R.~Brunetti, K.~Fredenhagen, and R.~Verch, \emph{The generally covariant
  locality principle---{A} new paradigm for local quantum field theory},
  Commun. Math. Phys. \textbf{237} (2003), 31--68.

\bibitem[BGN82]{BGN}
G.~Benfatto, G.~Gallavotti, and F.~Nicol{\`{o}}, \emph{On the massive
  sine-gordon equation in the first few regions of collapse}, Communications in
  Mathematical Physics \textbf{83} (1982), 387--410.

\bibitem[BK87]{BK}
D.~C. Brydges and T.~Kennedy, \emph{Mayer expansions and the hamilton-jacobi
  equation}, Journal of Statistical Physics \textbf{48} (1987), no.~1, 19--49.

\bibitem[BR18]{BR16}
D.~Bahns and K.~Rejzner, \emph{The quantum {S}ine {G}ordon model in
  perturbative {AQFT}}, Commun. Math. Phys. \textbf{357} (2018), 421--446,
  [arXiv:1609.08530].

\bibitem[BS59]{BS}
N.~Bogoliubov and D.~Shirkov, \emph{Introduction to the {T}heory of {Q}uantized
  {F}ields}, Introduction to the {T}heory of {Q}uantized {F}ields,
  Interscience, New York, 1959.

\bibitem[BW22]{bauerschmidt2022coleman}
Roland Bauerschmidt and Christian Webb, \emph{The coleman correspondence at the
  free fermion point}, 2022.

\bibitem[CF08]{ChF08}
B.~Chilian and K.~Fredenhagen, \emph{The time slice axiom in perturbative
  quantum field theory on globally hyperbolic spacetimes}, Commun. Math. Phys.
  \textbf{287} (2008), no.~2, 513--522.

\bibitem[Cie66]{Ciesielski}
Z.~Ciesielski, \emph{Lectures on brownian motion, heat conduction and potential
  theory}, Aarhus Universit{\"a}t, 1966.

\bibitem[DF01]{DF}
M.~D{\"u}tsch and K.~Fredenhagen, \emph{Perturbative algebraic field theory,
  and deformation quantization}, Mathematical Physics in Mathematics and
  Physics: Quantum and Operator Algebraic Aspects \textbf{30} (2001), 151--160.

\bibitem[DH93]{DimockHurd}
J.~Dimock and T.~Hurd, \emph{Construction of the two-dimensional sine-gordon
  model for $\beta<8\pi$}, Communications in mathematical physics \textbf{156}
  (1993), no.~3, 547--580.

\bibitem[DH00]{DimockHurd1}
J.~Dimock and T.R. Hurd, \emph{Sine-gordon revisited}, Annales Henri Poincaré
  \textbf{1} (2000), no.~3, 499–541.

\bibitem[{\relax DLMF}]{NIST:DLMF}
\emph{{\it NIST Digital Library of Mathematical Functions}},
  http://dlmf.nist.gov/, Release 1.1.6 of 2022-06-30, F.~W.~J. Olver, A.~B.
  {Olde Daalhuis}, D.~W. Lozier, B.~I. Schneider, R.~F. Boisvert, C.~W. Clark,
  B.~R. Miller, B.~V. Saunders, H.~S. Cohl, and M.~A. McClain, eds.

\bibitem[DM06]{DM06}
J.~Derezi{\'n}ski and K.~A. Meissner, \emph{Quantum massless field in 1+ 1
  dimensions}, pp.~107--127, Springer, 2006, In \textit{Mathematical Physics of
  Quantum Mechanics}, J.~Asch, A.~Joye Eds.

\bibitem[D{\"u}t19]{Due19}
M.~D{\"u}tsch, \emph{From classical field theory to perturbative quantum field
  theory}, Birkh{\"a}user, 2019.

\bibitem[FL14]{FL}
K.~Fredenhagen and F.~Lindner, \emph{Construction of {KMS} states in
  perturbative {QFT} and renormalized {H}amiltonian dynamics}, Commun. Math.
  Phys. \textbf{332} (2014), no.~3, 895--932.

\bibitem[FNW81]{FNW}
S.~A. Fulling, F.~J. Narcowich, and R.~M. Wald, \emph{Singularity structure of
  the two-point function in quantum field theory in curved spacetime, ii},
  Annals of Physics \textbf{136} (1981), no.~2, 243--272.

\bibitem[FR12]{FR}
K.~Fredenhagen and K.~Rejzner, \emph{{B}atalin-{V}ilkovisky formalism in the
  functional approach to classical field theory}, Communications in
  Mathematical Physics \textbf{314} (2012), no.~1, 93--127.

\bibitem[FR15]{LesHauches}
\bysame, \emph{Perturbative algebraic quantum field theory}, pp.~17--55,
  Springer, 2015, In \textit{Mathematical Aspects of Quantum Field Theories},
  D. Calaque, T. Strobl (Eds.).

\bibitem[Fr{\"o}76]{Froe76}
J.~Fr{\"o}hlich, \emph{Classical and quantum statistical mechanics in one and
  two dimensions: Two-component {Y}ukawa -- and {C}oulomb systems},
  Communications in Mathematical Physics \textbf{47} (1976), no.~3, 233--268.

\bibitem[FS76]{FS}
J.~Fr{\"o}hlich and E.~Seiler, \emph{The massive {T}hirring-{S}chwinger model
  ({QED}${}_2$): convergence of perturbation theory and particle structure},
  Helvetica Physica Acta \textbf{49} (1976), no.~6, 889--924.

\bibitem[GJ12]{GJ12}
J.~Glimm and A.~Jaffe, \emph{Quantum physics: a functional integral point of
  view}, Springer Science and Business Media, 2012.

\bibitem[GJS73]{GJS}
J.~Glimm, A.~Jaffe, and T.~Spencer, \emph{The particle structure of the weakly
  coupled {$P(\varphi)_2$} model and other applications of high temperature
  expansions}, pp.~199--242, Springer, 1973.

\bibitem[GJS74]{GJS74}
\bysame, \emph{The wightman axioms and particle structure in the
  {$P(\varphi)_2$} quantum field model"}, Annals of Mathematics \textbf{100}
  (1974), 585--632.

\bibitem[GRS75a]{GR75a}
F.~Guerra, L.~Rosen, and B.~Simon, \emph{The $p(\varphi)_2$ euclidean quantum
  field theory as classical statistical mechanics (part i)}, Ann. Math.
  \textbf{101} (1975), 111--189.

\bibitem[GRS75b]{GR75b}
\bysame, \emph{The $p(\varphi)_2$ euclidean quantum field theory as classical
  statistical mechanics (part ii)}, Ann. Math. \textbf{101} (1975), 191--259.

\bibitem[HHW67]{KMS}
R.~Haag, N.~Hugenholtz, and M.~Winnink, \emph{On the equilibrium state in
  quantum statistical mechanics}, Comm. Math. Phys. \textbf{5} (1967), 215.

\bibitem[H{\"o}r03]{Hoer1}
L.~H{\"o}rmander, \emph{The analysis of the linear partial differential
  operators {I}: Distribution theory and {F}ourier analysis}, Classics in
  Mathematics, Springer, Berlin, 2003.

\bibitem[HW01]{HW}
S.~Hollands and R.~M. Wald, \emph{Local {W}ick polynomials and time ordered
  products of quantum fields in curved spacetime}, Commun. Math. Phys.
  \textbf{223} (2001), no.~2, 289--326.

\bibitem[HW02]{HW01}
\bysame, \emph{Existence of local covariant time ordered products of quantum
  fields in curved spacetime}, Commun. Math. Phys. \textbf{231} (2002), no.~2,
  309--345.

\bibitem[NRS86]{NRS}
F.~Nicol{\`{o}}, J.~Renn, and A.~Steinmann, \emph{On the massive sine-gordon
  equation in all regions of collapse}, Communications in Mathematical Physics
  \textbf{105} (1986), 291--326.

\bibitem[Rad96]{Rad}
M.~J. Radzikowski, \emph{Micro-local approach to the {H}adamard condition in
  quantum field theory on curved space-time}, Commun. Math. Phys. \textbf{179}
  (1996), 529--553.

\bibitem[Rej16]{Book}
K.~Rejzner, \emph{Perturbative {A}lgebraic {Q}uantum {F}ield {T}heory. {A}n
  introduction for {M}athematicians}, Mathematical Physics Studies, Springer,
  2016.

\bibitem[Riv16]{Riv16}
V.~Rivasseau, \emph{Constructive tensor field theory}, SIGMA. Symmetry,
  Integrability and Geometry: Methods and Applications \textbf{12} (2016), 085.

\bibitem[Roe12]{Roep}
G.~Roepstorff, \emph{Path integral approach to quantum physics: an
  introduction}, Springer Science and Business Media, 2012.

\bibitem[Spe74]{spencer1974}
Thomas Spencer, \emph{The mass gap for the $p(\phi)_{2}$ quantum field model
  with a strong external field}, Comm. Math. Phys. \textbf{39} (1974), no.~1,
  63--76.

\bibitem[SR50]{SR}
E.~St{\"u}ckelberg and D.~Rivier, \emph{A propos des divergences en th{\'e}orie
  des champs quantifi{\'e}s}, Helv. Phys. Acta \textbf{23} (1950), no.~Suppl
  III, 236--239.

\bibitem[WG16]{GrW20}
R.~Wulkenhaar and H.~Grosse, \emph{Construction of a quantum field theory in
  four dimensions}, Frontiers of Fundamental Physics 14, vol. 224, SISSA
  Medialab, 2016, p.~151.

\end{thebibliography}
\end{document}